\documentclass[11pt]{article}
\usepackage{amssymb}
\usepackage{amsthm}
\usepackage{amsmath}
\usepackage[dvips]{epsfig}
\usepackage{graphicx}
\usepackage{color}
\usepackage{url}
\usepackage{tikz}
\usepackage{caption}
\usepackage{subcaption}
\usepackage{epstopdf}
\usepackage{float}
\usepackage{geometry}
\newtheorem{theorem}{Theorem}
\newtheorem{lemma}{Lemma}[section]
\newtheorem{proposition}[lemma]{Proposition}

\newtheorem{conjecture}[lemma]{Conjecture}
\newtheorem{remark}[lemma]{Remark}

\numberwithin{equation}{section}

\newcommand{\RE}{\mathbb R}

\newcommand{\ome}{\omega}

\DeclareMathOperator{\cn}{cn}

\newenvironment{proof1}%
{\begin{trivlist} \item[]{\em Proof }}%
{\hspace*{\fill}$\rule{.3\baselineskip}{.35\baselineskip}$\end{trivlist}}

\begin{document}

\title{\bf Bifurcations and stability of standing waves \\ in the nonlinear Schr\"{o}dinger equation \\ on the tadpole graph}

\author{Diego Noja$^{1}$, Dmitry Pelinovsky$^{2,3}$, and Gaukhar Shaikhova$^{2,4}$ \\
{\small $^{1}$ Dipartimento di Matematica e Applicazioni, Universit\`a di Milano Bicocca, }\\
{\small via R. Cozzi 55, 20125 Milano, Italy} \\
{\small $^{2}$ Department of Mathematics, McMaster University, Hamilton, Ontario, L8S 4K1, Canada} \\
{\small $^{3}$ Department of Applied Mathematics, Nizhny Novgorod State Technical University,} \\
{\small Nizhny Novgorod, Russia}\\
{\small $^{4}$ Department of General and Theoretical Physics, Eurasian National University,} \\
{\small Astana 010000, Kazakhstan}
}

\date{\today}
\maketitle

\begin{abstract}
We develop a detailed rigorous analysis of edge bifurcations of standing waves in
the nonlinear Schr\"odinger (NLS) equation on a tadpole graph
(a ring attached to a semi-infinite line subject to the Kirchhoff boundary conditions at the junction).
It is shown in the recent work \cite{CFN} by using explicit Jacobi elliptic functions that
the cubic NLS equation on a tadpole graph admits a rich structure of standing waves. Among these,
there are different branches of localized waves bifurcating from the edge of the essential spectrum
of an associated Schr\"{o}dinger operator.

We show by using a modified Lyapunov-Schmidt reduction method that the bifurcation of localized standing waves occurs
for every positive power nonlinearity. We distinguish a primary branch of never vanishing standing waves
bifurcating from the trivial solution and an infinite sequence of higher branches with
oscillating behavior in the ring. The higher branches bifurcate from
the branches of degenerate standing waves with vanishing tail outside the ring.

Moreover, we analyze stability of bifurcating standing waves. Namely, we show that the primary branch
is composed by orbitally stable standing waves for subcritical power nonlinearities, while all
nontrivial higher branches are linearly unstable near the bifurcation point. The stability character of
the degenerate branches remains inconclusive at the analytical level,
whereas heuristic arguments based on analysis of embedded eigenvalues of negative Krein signatures
support the conjecture of their linear instability at least near the bifurcation point.
Numerical results for the cubic NLS equation show that this conjecture is valid and that the degenerate
branches become spectrally stable far away from the bifurcation point.
\end{abstract}

\maketitle


\vspace{10pt}

\section{Introduction}

The study of existence and properties of standing waves of the nonlinear Schr\"odinger  (NLS) equation
constitutes a continuously developing subject. The NLS equation has potential applications to many realistic problems
such as signal propagation in optical fibers or Bose--Einstein condensation.
Standing waves are usually considered in unbounded homogeneous media \cite{Caz} or
in the periodically modulated media \cite{Pel}. Nevertheless, real systems can exhibit
strong inhomogeneities, due to different nonlinear coefficients in different regions
of the spatial domain or a specific geometry of the spatial domain.

A problem of general interest is the interaction between standing waves in spatially confined systems
and those in large or unbounded reservoirs.  Here we develop a rigorous analysis of bifurcation and stability
of standing waves for the NLS equation with power nonlinearity in the simplest geometry given by a ring attached to a semi-infinite
line. We refer to this model geometry as to the {\it tadpole graph}. At the junction
between the ring and the half line, suitable boundary conditions (referred typically to as the
{\em Kirchhoff boundary conditions})
are given to define the coupling. These boundary conditions ensure conservation
of the current flow through the network junction. The tadpole graph is an example of quantum graphs,
a much studied subject in the last decades (see \cite{BK} and references therein) with
many relevant physical applications.

The linear counterpart of this model, even in the presence
of a magnetic field, was studied by Exner \cite{E}. If the ring is placed on the interval $[-L,L]$
and the semi-infinite interval is $[L,\infty)$, then we define the Laplacian operator
by
\begin{equation}
\label{op-Delta}
\Delta \Psi = \left[ \begin{array}{l} u''(x), \quad x \in (-L,L) \\ v''(x), \quad x \in (L,\infty) \end{array} \right],
\end{equation}
acting on functions in the form
$$
\Psi = \left[ \begin{array}{l} u(x), \quad x \in (-L,L) \\ v(x), \quad x \in (L,\infty)\end{array} \right],
$$
where the primes stand for spatial derivatives. We equip the Laplacian operator (\ref{op-Delta})
with the domain
\begin{equation}
\label{domain-Delta}
\mathcal{D}(\Delta) = \Big\{ \begin{array}{l} (u,v) \in H^2(-L,L)\times H^2(L,\infty): \\
u(L)=u(-L)=v(L), \quad u'(L)-u'(-L)=v'(L) \end{array} \Big\},
\end{equation}
where the Kirchhoff boundary conditions have been used.

Let us show that $\Delta$ is a symmetric operator
from $\mathcal{D}(\Delta)$ to $L^2(-L,L) \times L^2(L,\infty)$.
Indeed, let $\Psi_1 = (u_1,v_1)$ and $\Psi_2 = (u_2,v_2)$ be two elements in $\mathcal{D}(\Delta)$.
The bilinear form for the Laplacian operator satisfies
\begin{eqnarray*}
\langle u_1, u_2'' \rangle_{L^2(-L,L)} + \langle v_1, v_2'' \rangle_{L^2(L,\infty)}
= \langle u_1'', u_2 \rangle_{L^2(-L,L)} + \langle v_1'', v_2 \rangle_{L^2(L,\infty)}
\end{eqnarray*}
if
\begin{eqnarray*}
u_1(L) u_2'(L) - u_1'(L) u_2(L) + v_1'(L) v_2(L) - v_1(L) v_2'(L) \\
 = u_1(-L) u_2'(-L) + u_1'(-L) u_2(-L).
\end{eqnarray*}
The above constraint is indeed satisfied under the Kirchhoff boundary conditions in (\ref{domain-Delta}).
It further follows (see, e.g., Theorem 1.4.4 in \cite{BK}) that the operator $\Delta$ is in fact 
self-adjoint on its domain $\mathcal{D}(\Delta)$.

The linear Schr\"odinger equation on the tadpole graph can be written in the compact form
\begin{equation}\label{eq00}
i\frac{\partial}{\partial t} \Psi = \Delta \Psi,
\end{equation}
where $\Psi = \Psi(t,x)$. We are interested in the nonlinear Schr\"odinger equation on the tadpole graph,
which is the natural generalization of the linear Schr\"{o}dinger equation \eqref{eq00},
\begin{equation}\label{eq0}
i\frac{\partial}{\partial t} \Psi = \Delta \Psi + (p+1)|\Psi|^{2p} \Psi\ ,
\end{equation}
where the nonlinear term $|\Psi|^{2p}\Psi$ is interpreted as a symbol for
$(|u|^{2p} u, |v|^{2p} v)$ defined piecewise on $(-L,L)$ and $(L,\infty)$.
For $p > 0$, the power nonlinearity is of the focusing type and it supports existence
of localized waves on infinite or semi-infinite lines.

Standard application of the fixed point theory shows that local well posedness
holds in the energy space
\begin{equation}
\label{energy-Delta}
{\mathcal E}(\Delta) = \{ (u,v)\in H^1(-L,L)\times H^1(L,\infty) : \quad u(L) = u(-L) = v(L) \}
\end{equation}
and in the operator domain space $\mathcal{D}(\Delta)$ (see \cite{Caz} for the classical theory and
\cite{ACFN2} for applications to the NLS equation on quantum graphs). In the case of subcritical
power nonlinearities
with $p \in (0,2)$, local solutions can also be extended to global solutions either in $\mathcal{E}(\Delta)$ or
in $\mathcal{D}(\Delta)$.

Standing waves of the focusing NLS equation (\ref{eq0}) on the tadpole graph
are given by the solutions of the form
$$
\Psi(t,x) = e^{i\ome t} \Phi(x),
$$
where $\omega$ and $\Phi \in \mathcal{D}(\Delta)$ are considered to be real.
This pair satisfies the stationary NLS equation
\begin{equation}\label{eq}
-\Delta \Phi - (p+1) |\Phi|^{2p} \Phi =  \ome \Phi \qquad \ome \in\RE\,,\;\Phi\in \mathcal{D}(\Delta).
\end{equation}
More explicitly, using $u$ and $v$ as components of the vector $\Phi$, we can write the
stationary NLS equation (\ref{eq}) as a system of two NLS equations, one on the ring and
the other one on the half line,
coupled by the Kirchhoff boundary conditions:
\begin{equation}
\label{stat-nls}
\left\{ \begin{array}{l} - u''(x) - (p+1) |u|^{2p} u = \omega u, \quad x \in (-L,L)\ , \\
- v''(x) - (p+1) |v|^{2p} v = \omega v, \quad x \in (L,\infty)\ ,\\
u(L) = u(-L) = v(L)\ ,\\ u'(L) - u'(-L) = v'(L)\ .
\end{array} \right.
\end{equation}

The subject of NLS equations on quantum graphs has seen many developments in the recent years.
From the physical point of view, the most promising interest is in the experimental creation
and management of various kinds of traps for Bose-Einstein condensates (see \cite{D+,GMM+,RAC+,SMKJ} and reference therein).
Various types of junctions have been modeled to show formation and trapping of localized waves and
existence of coherent structures with symmetry breaking \cite{[BS],[HTM],[Sob],[SSBM],[ZS]}.

At the rigorous mathematical level, the emphasis has been placed on the case of Y-junctions or more generally
on star graphs, where existence, variational properties, stability of standing waves, and
scattering of localized waves have been studied, e.g. in \cite{ACFN0,ACFN2,N}.
Very little is known about propagation or formation of standing waves in more complex structures,
where oscillations of waves can be present. For example, the authors of \cite{GnSD} demonstrate
numerically that a complex set of broad and narrow resonances shows up after inserting
a single nonlinear edge in a network, where linear Schr\"odinger propagation occurs.
A result on the absence of nonlinear ground states in networks with closed cycles is given in \cite{AST}
under a set of certain topological conditions.

A classification of standing waves in the present model (\ref{stat-nls}) of the NLS equation on the tadpole graph
is given in \cite{CFN} for the cubic case $p = 1$. The authors of \cite{CFN} showed
a rather unexpected and rich structure of the nonlinear waves of the system.
Several interesting bifurcations appear, giving rise to nonlinear standing waves
embedded in the essential spectrum of $\Delta$, a countable set of families of localized waves
bifurcating from the edge of the essential spectrum of $\Delta$,
and a wealth of families of standing waves which have no linear analogues.
Standing waves were constructed explicitly in \cite{CFN}, thanks to the known properties of Jacobi elliptic functions
related to the cubic NLS equation.

The task of the present work is to extend these results to the NLS equation with the generalized power
nonlinearity (\ref{eq0}). Our particular emphasis is on
the existence and stability of standing localized waves bifurcating from the edge of the essential spectrum
of $\Delta$. In contrast to the previous work in \cite{CFN}, we prove the existence
and bifurcation results without relying on the known exact solutions but using a modification
of the Lyapunov--Schmidt reduction method.
We also address the linear and orbital stability of the bifurcating standing waves near the bifurcation threshold,
which was not considered in \cite{CFN} even in the case of cubic nonlinearities.

As is well-known (see, e.g., Chapter 4 in \cite{Pel} for review), spectral and orbital stability of
the stationary solutions in the time evolution of the NLS equation (\ref{eq0})
is determined by the spectra of the self-adjoint operators $L_+$ and $L_-$,
which defines the energy quadratic form near the stationary solution.
In our context, the spectral problem for operator $L_-$ is defined by the boundary-value problem
\begin{equation}
\label{L-minus-nls}
\left\{ \begin{array}{l} - U''(x) - \omega U - (p+1) |u|^{2p} U = \lambda U, \quad x \in (-L,L), \\
- V''(x) - \omega V - (p+1) |v|^{2p} V = \lambda V, \quad x \in (L,\infty), \\
U(L) = U(-L) = V(L), \\
U'(L)-U'(-L) = V'(L),
\end{array} \right.
\end{equation}
where $(u,v)$ represents any stationary solution of the boundary-value problem (\ref{stat-nls}).
The spectral problem for operator $L_+$ is defined by the boundary-value problem
\begin{equation}
\label{L-plus-nls}
\left\{ \begin{array}{l} - U''(x) - \omega U - (2p+1)(p+1) |u|^{2p} U = \lambda U, \quad x \in (-L,L), \\
- V''(x) - \omega V - (2p+1)(p+1) |v|^{2p} V = \lambda V, \quad x \in (L,\infty), \\
U(L) = U(-L) = V(L), \\
U'(L)-U'(-L) = V'(L).
\end{array} \right.
\end{equation}
Similarly to the operator $\Delta$, operators $L_{\pm}$ are self-adjoint in $L^2(-L,L) \times L^2(L,\infty)$
with the domain $\mathcal{D}(\Delta)$ given in (\ref{domain-Delta}).

Our main result is that a countable set of standing localized waves bifurcates from the end point of the essential
spectrum of $\Delta$ at $\omega = 0$. This bifurcation, known as {\em the edge bifurcation},
was previously studied in the context of linear eigenvalue problems \cite{KS04}.
A {\em primary branch} of standing waves has no nodes
and it bifurcates from the vanishing state. The subsequent families of standing waves, or {\em higher branches},
can be ordered according to the increasing number of nodes and each family bifurcates from one of the standing waves which are identically zero outside the ring $(-L,L)$.

For small $\omega < 0$ and subcritical power nonlinearities with $p \in (0,2)$,
we show that the primary branch is composed by orbitally stable standing waves, whereas the higher branches
with non-vanishing localized tails are spectrally unstable. The spectral analysis is inconclusive
for the higher branches with the vanishing tail outside the ring. Numerical analysis
is developed for the cubic NLS equation with $p = 1$ to show linear instability of these higher branches
near the bifurcation point $\omega = 0$ and their spectral stability in the limit of large negative $\omega$.
At the same time, conclusions on the orbital stability of the standing waves along the primary branch
and the spectral instability of the higher branches with nonvanishing localized tails remain true for all
negative $\omega$. We note that orbital stability of the nodeless primary branch supports the idea
that it represents the ground state of the system, that is, it minimizes energy at constant mass,
as conjectured in \cite{CFN}.

The paper is organized as follows. In Section 2 we prove existence of a countable set
of standing waves which vanish on the tail of the tadpole graph.
This set is the basis for the subsequent analysis of bifurcation and stability of new standing waves.

In Section 3 we show the existence of the primary branch of
standing waves, which bifurcates from the trivial solution at $\omega=0$.
In Section 4 we construct the higher branches of standing waves, which bifurcate at
$\omega=0$ from each solution in the countable set constructed in Section 2, the bifurcating solutions
have non-vanishing localized tail outside the ring. Note that both perturbative results
are based on a non-trivial adaptation of the Lyapunov--Schmidt reduction method,
which holds near the edge bifurcation.

In Section 5 we consider orbital stability of primary branch near the bifurcation point,
using the general theory of Grillakis-Shatah-Strauss \cite{GSS1,GSS2} and Grillakis \cite{Gr}.
We count multiplicities of
the negative and zero eigenvalues of the linearized operators $L_-$ and $L_+$ and
verify the slope condition, which indicates if the number of negative eigenvalues of
$L_+$ is reduced by one on an orthogonal complement of the eigenvector of $L_-$.

In Section 6 we prove the linear instability of higher branches of standing waves
with non-vanishing localized tails outside the ring.
We note that the count of negative and zero eigenvalues of the linearized operators $L_-$ and $L_+$ is more difficult for the
higher branches and it involves eigenvalue problems with nonlinear dependence of the spectral
problem on the spectral parameter. We also show that the eigenvalue count is inconclusive
for the higher branches with vanishing tails outside the ring. Using heuristic arguments, we conjecture
that these higher branches are also unstable at least for small values of $\omega$.

Finally, in Section 7, we develop numerical approximations for the cubic case $p = 1$.
Numerical findings illustrate our analytical results near $\omega = 0$ and give
a complete picture on existence and stability of standing waves on the tadpole graphs
for larger values of the parameter $\omega$.

We use the following notations throughout the paper: 
\begin{itemize}
\item $H^2$ denotes the usual Sobolev space of square integrable functions with square integrable second derivatives;
\item $C^m$ denotes the space of $m$-times continuously differentiable functions with bounded derivatives up to the order $m$; 
\item ${\mathcal O}(\epsilon)$ denotes a quantity that converges to zero at the same rate as $\epsilon$  as $\epsilon \to 0$; 
\item ${\rm o}(\epsilon)$ denotes a quantity that converges to zero faster than $\epsilon$ as $\epsilon \to 0$; 
\item primes always denote derivatives with respect to the spatial variable.
\end{itemize}

\section{Standing waves prior to bifurcations}

Here we consider standing waves of the stationary NLS equation \eqref{stat-nls},
which are identically zero on the tail of the tadpole graph.
The decoupled stationary solutions with $v(x) \equiv 0$ for all $x \in [L,\infty)$
satisfy the nonlinear boundary value problem on the ring:
\begin{equation}
\label{stat-nls-scalar}
\left\{ \begin{array}{l} - u''(x) - (p+1) |u|^{2p} u = \omega u, \quad x \in (-L,L), \\
u(L) = u(-L) = 0, \\
u'(L) = u'(-L).
\end{array} \right.
\end{equation}
Although the boundary conditions over-determine the second-order boundary-value problem,
we can look for odd $2L$-periodic solutions in $H^2_{\rm per, odd}(-L,L)$, which satisfy the Dirichlet
boundary conditions at $x = \pm L$ as well as at $x = 0$. For each solution with $u'(0) > 0$,
we have also another solution with $u'(0) < 0$ because both $u$ and $-u$ satisfy the
same boundary-value problem (\ref{stat-nls-scalar}).

First, we characterize trajectories of the second-order differential equation
in the system (\ref{stat-nls-scalar}) on the phase plane $(u,u')$.
For any $\omega \in \mathbb{R}$, the differential equation is integrable
thanks to the energy invariant
\begin{equation}
\label{energy-invariant}
E = \left( \frac{d u}{dx} \right)^2 + \left( \omega + |u|^{2p} \right) u^2 = {\rm const}.
\end{equation}
For every $\omega \geq 0$, the level set (\ref{energy-invariant}) determines
closed trajectories on the phase plane $(u,u')$, which correspond to periodic solutions
with a minimal period, say $2T$. We claim the following.

\begin{lemma}
For every $p > 0$ and every $\omega \geq 0$, the period-to-energy map
\begin{equation}
\label{map-1}
\left(0,\frac{\pi}{\sqrt{\omega}} \right) \ni T \to E \in (0,\infty)
\end{equation}
associated with the closed periodic trajectories that surrounds the zero critical point
of the energy invariant (\ref{energy-invariant}) is a $C^1$
diffeomorphism with
\begin{equation}
\label{period-4}
E'(T) < 0, \quad \mbox{\rm for all } \; T \in \left(0,\frac{\pi}{\sqrt{\omega}} \right).
\end{equation}
\label{lemma-period}
\end{lemma}

\begin{proof}
By integrating the trajectory of the first-order equation (\ref{energy-invariant})
from the point $u(-T) = 0$ and $u'(-T) = \sqrt{E}$ to the point
$u(-T/2) = u_0$ and $u'(-T/2) = 0$, where $u_0$ is the positive root of the algebraic equation
\begin{equation}
\label{period-1}
( \omega + u_0^{2p} ) u_0^2  = E,
\end{equation}
we obtain the explicit representation of the half period $T$ in terms of $E$,
\begin{equation}
\label{period-2}
T = 2 \int_0^{u_0} \frac{du}{\sqrt{E - ( \omega + u^{2p}) u^2}}.
\end{equation}
Substituting $E$ from (\ref{period-1}) and using the change of variables $u = u_0 x$ with $x \in (0,1)$,
we rewrite (\ref{period-2}) in the equivalent form
\begin{equation}
\label{period-3}
T = 2 \int_0^{1} \frac{dx}{\sqrt{\omega (1-x^2) + u_0^{2p} (1 - x^{2p+2})}}.
\end{equation}
Thus, $E$ and $T$ are parameterized by $u_0$ in equations (\ref{period-1}) and (\ref{period-3}).
For every $\omega \geq 0$, we use these equations and obtain the asymptotic representations of $E$ and $T$
in the limits $u_0 \to 0$ and $u_0 \to \infty$, namely
$$
T = \frac{\pi}{\sqrt{\omega}} + \mathcal{O}(u_0^{2p}), \quad E = \omega u_0^2 + \mathcal{O}(u_0^{2p+2}),
\quad \mbox{\rm as} \quad u_0 \to 0,
$$
and
$$
T = \frac{2}{u_0^p} \int_0^1 \frac{dx}{\sqrt{1-x^{2p+2}}} + \mathcal{O}\left(\frac{1}{u_0^{3p}}\right), \quad
E = u_0^{2p+2} + \mathcal{O}(u_0^{2}), \quad \mbox{\rm as} \quad u_0 \to \infty,
$$
where the integral returns a finite value and we have used the fact that $x = 1$ is a simple root of both
$1-x^2$ and $1 - x^{2p+2}$.

Now we show that the period-to-energy map (\ref{map-1})
is $C^1$ and satisfies (\ref{period-4}). We use the chain rule, since
the map $(0,\infty) \ni u_0 \mapsto E \in (0,\infty)$
is monotonically increasing for every $\omega \geq 0$. Therefore, we only need to show that
the map $(0,\infty) \ni u_0 \mapsto T \in \left( 0,\frac{\pi}{\sqrt{\omega}}\right)$
is monotonically decreasing for every $\omega \geq 0$. This follows from the explicit computation,
\begin{equation}
\label{period-5}
\frac{d T}{d u_0} = -2 p u_0^{2p-1} \int_0^{1} \frac{(1 - x^{2p+2}) dx}{\sqrt{(\omega (1-x^2) + u_0^{2p} (1 - x^{2p+2}))^3}} < 0,
\end{equation}
where the integral returns a finite value because again $x = 1$ is a simple root of both
$1-x^2$ and $1 - x^{2p+2}$. The assertion of the lemma is proved.
\end{proof}

\begin{lemma}
For every $p > 0$ and every $\omega < 0$, the period-to-energy map
\begin{equation}
\label{map-2}
\left(0,\infty \right) \ni T \to E \in (0,\infty)
\end{equation}
associated with the closed periodic trajectories
that surround all critical points of the energy invariant (\ref{energy-invariant}) is a $C^1$ diffeomorphism with
\begin{equation}
E'(T) < 0, \quad \mbox{\rm for all } \; T \in \left(0,\infty \right).
\end{equation}
\label{lemma-period-2}
\end{lemma}

\begin{proof}
For $\omega < 0$, the phase-plane topology of the energy invariant (\ref{energy-invariant})
changes near the origin (for small values of $E$), where the zero critical point of $E$
becomes a saddle point and a pair of homoclinic orbits connecting
the zero critical point arise. The homoclinic orbits correspond to the zero value of $E$.
Outside of the homoclinic orbits, for $E > 0$, a family of closed periodic trajectories
surrounding all three critical points of $E$ exists, which correspond to the periodic solution with a minimal period $2T$.
As $E \to 0$, we have $T \to \infty$. The rest of the proof repeats the proof of Lemma \ref{lemma-period}.
\end{proof}

Using Lemmas \ref{lemma-period} and \ref{lemma-period-2}, we prove
the existence of suitable solutions to the homogeneous problem (\ref{stat-nls-scalar}).
The following proposition gives the relevant result.

\begin{proposition}
For every $p > 0$, the boundary-value problem (\ref{stat-nls-scalar}) with $\omega \in (-\infty,\omega_n)$
admits two solutions $u^{\pm}_{n,\omega}$ in $H^2_{\rm per, odd}(-L,L)$, where
$\omega_n := \frac{\pi^2 n^2}{L^2}$ and $n \in \mathbb{N}$.
Each pair of $u^{\pm}_{n,\omega}$ is uniquely defined by the conditions
$u_{n,\omega}^{+ \prime}(0) > 0$ and $u_{n,\omega}^{- \prime}(0) < 0$.
In fact, $u_{n,\omega}^- = -u_{n,\omega}^+$. Moreover,
the map $\omega \mapsto
u^{\pm}_{n,\omega} \in H^2_{\rm per, odd}(-L,L)$ is $C^1$ in $\omega$.
\label{proposition-countable}
\end{proposition}

\begin{proof}
For $\omega = \omega_n := \frac{\pi^2 n^2}{L^2}$, there exists a solution of
the linear version of the boundary-value problem (\ref{stat-nls-scalar}) in $H^2_{\rm per, odd}(-L,L)$
given by
\begin{equation}
\label{solution-bifurcation}
u(x) = u_n(x) := \sin\left(\frac{\pi n x}{L}\right).
\end{equation}
The eigenvalue $\omega = \omega_n$ is simple
in the space of odd $(2L)$-periodic solutions. By the standard Lyapunov--Schmidt reduction method,
there exists a unique $C^1$ continuation of the solution $u \in H^2_{\rm per, odd}(-L,L)$ of the nonlinear boundary-value problem
(\ref{stat-nls-scalar}) with $u'(0) > 0$ with respect to the parameter $\omega$ near $\omega_n$. Moreover,
a relatively straightforward computation shows that the continuation exists
for $\omega \lesssim \omega_n$. Let us denote this solution as $u_{n,\omega}^+$.
By the invariance of the boundary-value problem (\ref{stat-nls-scalar}) with respect
to the transformation $u \mapsto -u$, we can also construct another solution
$u_{n,\omega}^- = - u_{n,\omega}^+$, which satisfies the condition $u'(0) < 0$.

In order to prove that the solution branches $u^{\pm}_{n,\omega}$ persist for any $\omega < \omega_n$
and remain $C^1$ in $\omega$, we use the result of Lemmas \ref{lemma-period} and \ref{lemma-period-2}.
If $\omega \in [0,\omega_n)$, then $T_n := L/n$ belongs to the
range $\left(0,\frac{\pi}{\sqrt{\omega}}\right)$, hence, two odd $2L$-periodic
solutions $u^{\pm}_{n,\omega}$ exist for the energy level $E_n = E(T_n) > 0$.
If $\omega \in (-\infty,0)$, we can still find two odd $2L$-periodic solutions
$u^{\pm}_{n,\omega}$ for the given energy level $E_n = E(T_n) > 0$.
Consequently, the solutions $u^{\pm}_{n,\omega}$ are uniquely continued with respect to parameter $\omega$
for every $\omega \in (-\infty,\omega_n)$. Moreover, from monotonicity
of the period-to-energy maps (\ref{map-1}) and (\ref{map-2}), $C^1$ smoothness of the nonlinear terms in the
boundary-value problem
(\ref{stat-nls-scalar}) for $p > 0$, and $C^1$ smoothness of the energy invariant (\ref{energy-invariant})
with respect to $\omega$, it follows that the map $(-\infty,\omega_n) \ni \omega \mapsto
u^{\pm}_{n,\omega} \in H^2_{\rm per, odd}(-L,L)$ is $C^1$ in $\omega$.
\end{proof}

\begin{remark}
In the cubic case $p = 1$, one can obtain the explicit form of the family $u^{\pm}_{n,\omega}$.
The solutions are expressed by the Jacobian elliptic functions with parameters depending on $\omega$.
See \cite{CFN} and Section 7 below.
\end{remark}

\begin{remark}
The above construction of the countable double set of solutions $\{u^{\pm}_{n,\omega}\}_{n \in \mathbb{N}}$
can be extended to more general nonlinearities. The only property needed is the topological structure of
the level set $E$ for trajectories on the phase plane $(u,u')$.
\end{remark}

\begin{remark}
The smooth families of solutions $\{u^{\pm}_{n,\omega}\}_{n \in \mathbb{N}}$
bifurcate from the linear eigenstates of the operator $\Delta$
in the space of odd $(2L)$-periodic solutions. The bifurcation branches so obtained are
then globally extended to $\omega\in (-\infty,\omega_n)$.
\end{remark}

\begin{remark}
The boundary-value problem (\ref{stat-nls-scalar}) in the space of odd $2L$-periodic solutions
gives a subset of all solutions of the stationary NLS equation (\ref{eq})
on the interval $[-L,L]$ subject to the Dirichlet boundary conditions at $x = \pm L$.
Existence and stability of the latter solutions are considered in \cite{FHK}.
\end{remark}

In what follows, we are concerned with bifurcations of other branches of solutions of
the coupled boundary--value problem \eqref{stat-nls}. First, we observe that the spectrum of the linear operator
$-\Delta$ with the domain $\mathcal{D}(\Delta)$ given by (\ref{op-Delta}) and (\ref{domain-Delta}) is located on $[0,\infty)$
and includes the essential spectrum and the set of embedded simple eigenvalues $\{ \omega_n \}_{n \in \mathbb{N}}$.
As a result, for $\omega \geq 0$, the only square integrable solution $(u,v)$ of the coupled boundary--value problem \eqref{stat-nls}
is the solution with vanishing $v = 0$, that is, the solution of Proposition \ref{proposition-countable}.

At $\omega = 0$, the {\em edge bifurcation} takes place, when
new branches of solutions can bifurcate off the trivial solution $u = 0$ or
the countable double set of solutions $\{ u_{n,\omega}^{\pm} \}_{n \in \mathbb{N}}$
satisfying the scalar boundary--value problem (\ref{stat-nls-scalar}).
The new solutions occur for small negative $\omega$.
Correspondingly, we refer to the {\em primary branch} for
solutions of the coupled problem \eqref{stat-nls} bifurcating from the trivial solution
and to the {\em higher branches} for solutions of the coupled problem \eqref{stat-nls} bifurcating
from the nontrivial solutions in the set $\{ u_{n,\omega}^{\pm} \}_{n \in \mathbb{N}}$.

To study bifurcations of both the primary and higher branches,
we can set $\omega = -\epsilon^2$ and consider small values of $\epsilon$.  Without loss of
generality, we restrict $\epsilon$ to positive values.
For the $v$ component, we use the scaling transformation for the stationary NLS equation (\ref{stat-nls})
and express explicitly the dependence of the solution on $\epsilon$:
\begin{equation}
\label{scaling-v}
v(x) = \epsilon^{\frac{1}{p}} \phi(z), \quad z = \epsilon (x-L),
\end{equation}
where $\phi$ is a decaying solution of the second-order equation
\begin{equation}
\label{soliton}
-\phi''(z) + \phi - (p+1) |\phi|^{2p} \phi = 0, \quad z > 0.
\end{equation}

Let $\phi_0$ be a unique solitary wave of the second-order equation (\ref{soliton}) such that
$\phi_0(0) > 0$ and $\phi_0'(0) = 0$. Recall that $\phi_0(z) > 0$ for all $z > 0$.
Then, there exists a one-parameter family of positive
decaying solutions $\phi(z) = \phi_0(z+a)$, parameterized by the translation $a \in \mathbb{R}$.
Note that the family is known in the simple explicit form,
\begin{equation}
\label{soliton-sech}
\phi_0(z) = {\rm sech}^{\frac{1}{p}}(pz).
\end{equation}

The bifurcation problem for stationary solutions with nonzero $v$ can now be reduced to the closed, over-determined system of equations
\begin{equation}
\label{bif-nls}
\left\{ \begin{array}{l} - u''(x) + \epsilon^2 u - (p+1) |u|^{2p} u = 0, \quad x \in (-L,L), \\
u(L) = u(-L) = \epsilon^{\frac{1}{p}} \phi_0(a), \\
u'(L) - u'(-L) = \epsilon^{1+\frac{1}{p}} \phi'_0(a).
\end{array} \right.
\end{equation}
The additional boundary condition on the tadpole graph is supposed to specify
uniquely the additional parameter $a$ in terms of $\epsilon$.
Solutions to the boundary-value problem (\ref{bif-nls}) are different between the primary and
higher branches. Therefore, we proceed differently with the primary and higher
branches in the next two sections.

\section{Primary branch}

The following theorem specifies a unique primary branch of stationary solutions to
the boundary-value problem (\ref{bif-nls}) with $\epsilon > 0$ such that $u = 0$ as $\epsilon = 0$.

\begin{theorem}
For every $p > 0$ and every $\epsilon > 0$ sufficiently small, there exists a unique positive solution
$u \in C^{\infty}(-L,L)$ and $a \in \mathbb{R}$
of the boundary-value problem (\ref{bif-nls}) such that
$$
\| u \|_{L^{\infty}(-L,L)} = \mathcal{O}(\epsilon) \quad \mbox{\rm and} \quad
a = \mathcal{O}(\epsilon) \quad \mbox{\rm as} \quad \epsilon \to 0.
$$
Moreover, the following asymptotic expansions hold
\begin{equation}
u = {\epsilon}^{\frac{1}{p}} (1+ \mathcal{O}_{C^{\infty}(-L,L)}(\epsilon^2))  \quad
\mbox{\rm and} \quad a = 2L \epsilon + \mathcal{O}(\epsilon^3).
\end{equation}
\label{theorem-primary}
\end{theorem}

\begin{proof}
We make use of the scaling transformation
\begin{equation}
\label{scal-transf}
u(x) = \epsilon^{\frac{1}{p}} \psi(z), \quad z = \epsilon x,
\end{equation}
and rewrite the boundary-value problem (\ref{bif-nls}) in the form
\begin{equation}
\label{bif-nls-scale}
\left\{ \begin{array}{l} - \psi''(z) + \psi - (p+1) |\psi|^{2p} \psi = 0, \quad z \in (-\epsilon L, \epsilon L), \\
\psi(\epsilon L) = \psi(-\epsilon L) = \phi_0(a), \\
\psi'(\epsilon L) - \psi'(-\epsilon L) = \phi'_0(a).
\end{array} \right.
\end{equation}

Let us consider the initial-value problem for the second-order differential equation
in system (\ref{bif-nls-scale}) starting with the initial data $(\psi,\psi') = (\psi_0,\psi_0')$,
where $\psi_0 > 0$. By bootstrapping arguments, $\psi \in C^3$ near $z = 0$ if $p > 0$.
From the boundary condition $\psi(\epsilon L) = \psi(- \epsilon L)$, we realize
that $\psi_0' \to 0$ as $\epsilon \to 0$.
Now, for fixed $L > 0$ and small $\epsilon > 0$, the interval $(-\epsilon L,\epsilon L)$ is narrow
and the initial point is close to $(\psi_0,0)$, where $\psi_0 > 0$.
From symmetry of trajectories on the phase plane $(\psi,\psi')$, it follows
that for every small $\epsilon > 0$, the boundary condition $\psi(\epsilon L) = \psi(- \epsilon L)$ can be satisfied if and only if
the function $\psi$ is even in $z$. Therefore, we can consider
an initial-value problem starting with the initial data at $(\psi,\psi') = (\psi_0,0)$
parameterized by $\psi_0 > 0$.

By the existence and uniqueness theory, for any $\psi_0 > 0$, there exists a finite $z_0 > 0$
and a unique local solution $\psi \in C^1(-z_0,z_0)$ such that $\psi(0) = \psi_0$ and $\psi'(0) = 0$.
Note that $z_0$ depends on $\psi_0$ and it can be chosen to guarantee that
$\psi(z) > 0$ for all $z \in (-z_0,z_0)$. Since $L > 0$ is fixed, we have $\epsilon L < z_0$
for sufficiently small $\epsilon$, so that
the unique positive local solution exists on the interval $[-\epsilon L,\epsilon L]$. By the bootstrapping arguments,
because the nonlinear vector field in system (\ref{bif-nls-scale}) is smooth for positive $\psi$,
the unique local solution is $\psi \in C^{\infty}(-\epsilon L,\epsilon L)$.

The boundary condition $\psi(\epsilon L) = \phi_0(a)$ in system (\ref{bif-nls-scale}) yields
the following algebraic equation for the solution $\psi$ parameterized by $\psi_0 = \psi(0)$:
\begin{equation}
\label{alg-eq-1}
\phi_0(a) = \psi(\epsilon L) = \psi_0 + \frac{1}{2} \psi''(0) \epsilon^2 L^2 + \mathcal{O}(\epsilon^4),
\end{equation}
where $\psi''(0)$ is expressed in terms of $\psi_0$ by the differential equation in system (\ref{bif-nls-scale}).
For $\epsilon = 0$ and any $a \in \mathbb{R}$,
the algebraic equation (\ref{alg-eq-1}) has a unique solution $\psi_0 = \phi_0(a)$.
Moreover, the derivative of $\psi(\epsilon L)$ with respect to $\psi_0$
is $1 + \mathcal{O}(\epsilon^2)$. By the implicit function theorem,
for $\epsilon \in \mathbb{R}$ sufficiently small and any $a \in \mathbb{R}$,
there exists a unique root of the algebraic equation (\ref{alg-eq-1}) for $\psi_0$  such that
$\psi_0 = \phi_0(a) + \mathcal{O}(\epsilon^2)$.

The last boundary condition $\psi'(\epsilon L) - \psi'(-\epsilon L) = \phi_0'(a)$ in system (\ref{bif-nls-scale}) yields
another algebraic equation
\begin{eqnarray}
\label{alg-eq-2}
\phi'_0(a) = 2 \psi'(\epsilon L) = 2 \psi''(0) \epsilon L + \mathcal{O}(\epsilon^3).
\end{eqnarray}
Since $\phi_0'(0) = 0$ and $\phi_0''(0) \neq 0$, the implicit function theorem
implies that for every $\epsilon \in \mathbb{R}$ sufficiently small,
there exists a unique root $a$ of the algebraic equation (\ref{alg-eq-2}) such that
\begin{eqnarray*}
a = \frac{2 \epsilon L \psi''(0)}{\phi_0''(0)} + \mathcal{O}(\epsilon^3)
= \frac{2 \epsilon L (1 - (p+1) \psi_0^{2p}) \psi_0}{(1 - (p+1) \phi_0(0)^{2p})\phi_0(0)} + \mathcal{O}(\epsilon^3)
= 2 \epsilon L + \mathcal{O}(\epsilon^3),
\end{eqnarray*}
where the second-order differential equations (\ref{soliton}) and (\ref{bif-nls-scale}) are used
as well as the expansion $\psi_0 = \phi_0(a) + \mathcal{O}(\epsilon^2) = 1 + \mathcal{O}(\epsilon^2)$ since
$\phi_0(0)=1$ and $\phi_0'(0) = 0$. The theorem is proved by using (\ref{scal-transf}) with
$\psi(z) = 1 + \mathcal{O}_{C^{\infty}(-\epsilon L, \epsilon L)}(\epsilon^2)$.
\end{proof}

\section{Higher branches}

Here we consider bifurcations of stationary solutions of the perturbed problem (\ref{bif-nls})
from nontrivial solutions of the homogeneous boundary-value problem (\ref{stat-nls-scalar}) with $\omega = -\epsilon^2$.
By Proposition \ref{proposition-countable}, there exists a countable double set $\{u^{\pm}_{n,\omega}\}_{n \in
\mathbb{N}}$ of these solutions for every $\omega < 0$, that is, for every $\epsilon > 0$.  In what follows, we take one solution from
the countable double set and denote it by $u_{\epsilon}$.

Note that $u_{\epsilon}$ is odd and $(2L)$-periodic.
The scaling transformation (\ref{scal-transf}) cannot be used because $u_{\epsilon}$ does not
vanish in the limit $\epsilon \to 0$.  Nevertheless, we can immediately
construct a suitable solution of the inhomogeneous boundary-value problem (\ref{bif-nls})
from the solution $u_{\epsilon}$ of the homogeneous boundary-value problem (\ref{stat-nls-scalar})
with $\omega = -\epsilon^2$. The following theorem gives the relevant result.

\begin{theorem}
Let $p > 0$ be fixed and $\epsilon > 0$ be sufficiently small.
For each $u_{\epsilon} \in H^2_{\rm per, odd}(-L,L)$ that solves
the homogeneous problem (\ref{stat-nls-scalar}) with $\omega = -\epsilon^2$, there exists a
solution $u \in H^2_{\rm per}(-L,L)$ to the boundary-value problem
(\ref{bif-nls}) with $a = 0$ and $u(x) = u_{\epsilon}(x+b)$, where
$b$ is uniquely determined from the boundary condition
\begin{equation}
\label{bc-b}
u_{\epsilon}(L+b) = \epsilon^{\frac{1}{p}}.
\end{equation}
In particular, the following asymptotic expansion holds
\begin{equation}
b = \epsilon^{\frac{1}{p}} \left[ \frac{1}{u_0'(L)} + \mathcal{O}\left(\epsilon^{\min\left\{2,\frac{2}{p} \right\}}\right) \right].
\label{expansion-0}
\end{equation}
\label{theorem-higher-order}
\end{theorem}

\begin{proof}
By Proposition \ref{proposition-countable}, $u_{\epsilon} \in H^2_{\rm per, odd}(-L,L)$
exists and is $C^1$ in $\epsilon^2$. By bootstrapping arguments, $u_{\epsilon} \in C^3(-L,L)$ if $p > 0$.
The translation of this solution $u = u_{\epsilon}(x+b)$ for every $b \in \mathbb{R}$
satisfies the second-order differential equation in system (\ref{bif-nls}). If $a = 0$,
it also satisfies the boundary conditions in system (\ref{bif-nls}) if and only if $b$
can be found from the boundary condition (\ref{bc-b}), where we recall that $\phi_0(0) = 1$ and $\phi_0'(0) = 0$.
Since $u_{\epsilon}(L) = 0$, $u'_{\epsilon}(L) \neq 0$, and $\epsilon$ is small,
a unique solution for $b$ exists by the implicit function theorem such that
\begin{equation}
\label{expansions-1}
b = \epsilon^{\frac{1}{p}} \left[ \frac{1}{u_{\epsilon}'(L)} + \mathcal{O}\left( \epsilon^{\frac{2}{p}}\right) \right].
\end{equation}
where we have used $u''_{\epsilon}(L) = 0$, and
the $C^3$ smoothness of $u_{\epsilon}$ in $x$. Furthermore,
from the $C^1$ smoothness of $u_{\epsilon}$ in $\epsilon^2$, we have
\begin{equation}
\label{expansions-2}
u'_{\epsilon}(L) = u_0'(L) + \mathcal{O}(\epsilon^2).
\end{equation}
Expansions (\ref{expansions-1}) and (\ref{expansions-2}) yield
the asymptotic expansion (\ref{expansion-0}).
\end{proof}

\begin{remark}
The sign of $b$ coincides with the sign of $u'_0(L)$, which is different between the two members
of the double family $\{u^{\pm}_{n,\omega}\}_{n \in \mathbb{N}}$. More precisely, we have
$\text{sign}\  b^+_n = (-1)^{n+1}$ and $\text{sign}\ b^-_n = (-1)^{n}$. Correspondingly,
having fixed the sign of the solution $\phi$ of the differential equation (\ref{soliton})
on the half line as positive by convention,
the two different solutions on the tadpole graph are approximately related for small negative $\omega$
as follows:
$$
(u^{-}_{n,\omega}(x+b^{-}_n),\epsilon^{\frac{1}{p}}\phi_0) =
(-u^{+}_{n,\omega}(x+b^{-}_n),\epsilon^{\frac{1}{p}}\phi_0) \approx
(-u^{+}_{n,\omega}(x-b^{+}_n),\epsilon^{\frac{1}{p}}\phi_0).
$$
Notice that the two solutions in the set $\{u^{\pm}_{n,\omega}\}_{n \in \mathbb{N}}$ for a fixed $n \in \mathbb{N}$
belong to  the same $U(1)$ orbit of the stationary NLS equation (\ref{eq}), while
the two new solutions $(u,v)$ bifurcating from $(u^{\pm}_{n,\omega},0)$ do not belong to the same $U(1)$ orbit of
the stationary NLS equation (\ref{eq}) because ${\rm sign}(b^{+}_n) = - {\rm sign}(b^{-}_n)$.
\label{remark-4-1}
\end{remark}

\begin{remark}
For the same value of $\omega = -\epsilon^2 < 0$, all four solutions mentioned in Remark \ref{remark-4-1}
have the same $L^2(-L,L)$ norm for the component $u$ because the mean value of periodic functions does not depend
on the initial point of integration over the period. At the same time,
the $L^2(L,\infty)$ norm for the component $v$ is zero for the two solutions in
$\{u^{\pm}_{n,\omega}\}_{n \in \mathbb{N}}$
and nonzero for the two bifurcating solutions in Theorem \ref{theorem-higher-order}.
\end{remark}

In the rest of this section, we will prove that the solution to the perturbed
problem (\ref{bif-nls}) near $u_{\epsilon}$ is uniquely continued for small values
of $\epsilon$. By uniqueness, this continuation coincides with the solution given in
Theorem \ref{theorem-higher-order}.

Let us consider the associated linearized operator
$$
M_{\epsilon} := -\frac{d^2}{d x^2} + \epsilon^2 - (p+1)(2p+1) |u_{\epsilon}(x)|^{2p} :
\quad H^2_{\rm per}(-L,L) \to L^2_{\rm per}(-L,L),
$$
where $u_{\epsilon} \in H^2_{\rm per, odd}(-L,L)$ is a solution of the
boundary-value problem (\ref{stat-nls-scalar}) with $\omega = -\epsilon^2$.
For every $\epsilon \geq 0$ sufficiently small, let us continue $u_{\epsilon}$
in a family of odd functions with $u'(0) > 0$, which are parameterized
by the energy level $E$ given by the energy invariant (\ref{energy-invariant}),
that is,
\begin{equation}
\label{energy-invariant-higher}
E = \left( \frac{d u}{dx} \right)^2 - \epsilon^2 u^2 + |u|^{2p} u^2 = {\rm const}.
\end{equation}
Denote the continuation by $U_{\epsilon}(x;E)$ and the half-period
of this family by $T_{\epsilon}(E)$. By Lemma \ref{lemma-period-2},
both $U_{\epsilon}$ and $T_{\epsilon}$ are $C^1$ in $E$ and
$T_{\epsilon}'(E) < 0$ for every small $\epsilon \geq 0$.
Let $E_{\epsilon}$ be the level such that $L = T_{\epsilon}(E_{\epsilon})$
and $u_{\epsilon}(x) = U_{\epsilon}(x;E_{\epsilon})$. The level $E_{\epsilon}$
is unique due to the monotonicity of the period-to-energy map (\ref{map-2}).

By taking the derivatives of the second-order equation
$$
- U_{\epsilon}''(x;E) + \epsilon^2 U_{\epsilon}(x;E) - (p+1) |U_{\epsilon}(x;E)|^{2p} U_{\epsilon}(x;E) = 0,
$$
with respect to $x$ and $E$ at $E = E_{\epsilon}$, we verify that
$$
M_{\epsilon} u'_{\epsilon} = 0 \quad \mbox{\rm and} \quad
M_{\epsilon} \partial_E U_{\epsilon}|_{E = E_{\epsilon}} = 0,
$$
where the prime denotes the derivative of $u_{\epsilon}$ in $x$.
We note that $u'_{\epsilon}$ is even and $(2L)$-periodic, whereas
$\partial_E U_{\epsilon} |_{E = E_{\epsilon}}$ is
odd but not $(2L)$-periodic if $T_{\epsilon}'(E_{\epsilon}) < 0$, since
\begin{equation}
\label{bc-solutions}
\partial_E U_{\epsilon}(\pm L;E_{\epsilon})= \mp T_{\epsilon}'(E_{\epsilon}) u'_{\epsilon}(L) \neq 0.
\end{equation}
Since the Wronskian between the two particular solutions
of the homogeneous equation $M_{\epsilon} u = 0$ is constant, we have
\begin{equation}
\label{Wronskian}
W_{\epsilon} := \left| \begin{array}{cc} u'_{\epsilon} & \partial_E U_{\epsilon} |_{E = E_{\epsilon}} \\
u''_{\epsilon} & \partial_E U_{\epsilon}' |_{E = E_{\epsilon}} \end{array} \right| = {\rm const}.
\end{equation}
Because the two particular solutions of $M_{\epsilon} u = 0$
are linearly independent, we also have $W_{\epsilon} \neq 0$ for every $\epsilon \geq 0$.

Let us now decompose solution of the perturbed problem (\ref{bif-nls}) near $u_{\epsilon}$ by posing $u = u_{\epsilon} + w$, where the perturbation
$w$ satisfies the nonlinear boundary-value problem
\begin{equation}
\label{per-nls}
\left\{ \begin{array}{l} M_{\epsilon} w =
(p+1)\left(|u_{\epsilon}+w|^{2p}(u_{\epsilon}+w)-|u_{\epsilon}|^{2p}u_{\epsilon}-(2p+1)|u_{\epsilon}|^{2p} w\right)\\
w(L) = w(-L) = {\epsilon}^{\frac{1}{p}}\phi_0(a), \\
w'(L) - w'(-L) = \epsilon^{1+\frac{1}{p}} \phi'_0(a).
\end{array} \right.
\end{equation}
For $\epsilon = 0$, there exists a trivial solution $w = 0$
of the boundary-value problem (\ref{per-nls}). The following
results specify a unique continuation of the small solution $w$
to the boundary-value problem (\ref{per-nls}) with respect to small $\epsilon$.

\begin{lemma}
Let $p > 0$ be fixed and $\epsilon > 0$ be sufficiently small.
There exists a unique solution
$w \in C^1(-L,L)$ and $a \in \mathbb{R}$
to the boundary-value problem (\ref{per-nls}) such that $a = 0$ and
\begin{equation}
w(x) = \epsilon^{\frac{1}{p}} \frac{u_{\epsilon}'(x)}{u_{\epsilon}'(L)}
+ {\rm o}_{C^1(-L,L)}\left(\epsilon^{\frac{1}{p}}\right).
\label{expansion-1}
\end{equation}
Consequently, the solution in Theorem \ref{theorem-higher-order} is unique.
\label{lemma-higher-order}
\end{lemma}

\begin{proof}
Because ${\rm ker}(M_{\epsilon}) = {\rm span}\{u_{\epsilon}'\}$ in $L^2_{\rm per}(-L,L)$,
we consider the Lyapunov--Schmidt decomposition
\begin{equation}
\label{decomposition-w}
w(x) = c u_{\epsilon}'(x) + \psi(x),
\end{equation}
where $c \in \mathbb{R}$ and $\psi \in H^2(-L,L)$
are to be uniquely defined in what follows. In the standard Lyapunov--Schmidt
reduction method, the orthogonal projection $\langle u_{\epsilon}', \psi \rangle_{L^2_{\rm per}(-L,L)} = 0$
is typically used. However, because the boundary conditions in the problem (\ref{per-nls})
are not periodic, we will modify the conditions by requiring
\begin{equation}
\label{bc-psi}
\psi(L) = \psi(-L) = 0.
\end{equation}
Although it may seem that the two boundary conditions
for $\psi$ over-determine the decomposition (\ref{decomposition-w}) with only one parameter $c$, we shall recall here
that $w$ is required to satisfy the boundary condition $w(-L) = w(L)$, whereas
$u_{\epsilon}'$ is even in $x$. Therefore, $c$ is uniquely determined by
the boundary conditions (\ref{bc-psi}) if and only if $w$ is a solution of
the boundary-value problem (\ref{per-nls}). To be precise, for given small $a \in \mathbb{R}$
and $\epsilon \in \mathbb{R}$, parameter $c$ is uniquely determined by
\begin{equation}
\label{expression-c}
c =  \epsilon^{\frac{1}{p}} \frac{\phi_0(a)}{u'_{\epsilon}(L)},
\end{equation}
where we recall that $u'_{\epsilon}(L) \neq 0$.

There exists a unique solution of the inhomogeneous equation $M_{\epsilon} w = F$
subject to the boundary conditions (\ref{bc-psi}), where $F$ is a given function in $L^2(-L,L)$,
which does not need to be $2L$-periodic. Indeed, by the variation of constant formula,
we obtain
\begin{eqnarray*}
\psi(x) & = & c_1 u'_{\epsilon}(x) + c_2 \partial_E U_{\epsilon}(x;E_{\epsilon}) \\
& \phantom{t} & \phantom{text}
+ \frac{1}{W} \int_0^x F(y) \left[  u'_{\epsilon}(y)  \partial_E U_{\epsilon}(x;E_{\epsilon}) -
u'_{\epsilon}(x)  \partial_E U_{\epsilon}(y;E_{\epsilon}) \right] dy.
\end{eqnarray*}
The coefficients $c_1$ and $c_2$ are uniquely found from the boundary conditions
(\ref{bc-psi}). After
routine computations involving relations (\ref{bc-solutions}) and (\ref{Wronskian}), we obtain
a unique representation for $\psi$ in the form
\begin{eqnarray}
\nonumber
\psi(x) & = & \frac{1}{2 W_{\epsilon}}
\left( \int_{-L}^x - \int_x^L \right) F(y) \left[  u'_{\epsilon}(y)  \partial_E U_{\epsilon}(x;E_{\epsilon}) -
u'_{\epsilon}(x)  \partial_E U_{\epsilon}(y;E_{\epsilon}) \right] dy \\ \nonumber
& \phantom{t} & + \frac{T_{\epsilon}'(E_{\epsilon})}{2 W_{\epsilon}} \langle u'_{\epsilon}, F \rangle_{L^2(-L,L)} u'_{\epsilon}(x)\\
&\phantom{t} & - \frac{1}{2 W_{\epsilon} T_{\epsilon}'(E_{\epsilon})}
\langle \partial_E U_{\epsilon} |_{E=E_{\epsilon}}, F \rangle_{L^2(-L,L)} \partial_E U_{\epsilon}(x;E_{\epsilon}).
\label{integral-eq}
\end{eqnarray}
Note that $\psi$ is not a $(2L)$-periodic function in $H^2_{\rm per}(-L,L)$ unless $F$ satisfies the Fredholm
solvability condition $\langle u'_{\epsilon}, F \rangle_{L^2(-L,L)} = 0$.
Substituting the decomposition (\ref{decomposition-w})
to the differential equation in system (\ref{per-nls}), we obtain
$M_{\epsilon} \psi = F$ with
$$
F(c,\psi) := (p+1)\left(|u_{\epsilon} + c u_{\epsilon}' + \psi |^{2p}(u_{\epsilon}+ c u_{\epsilon}' + \psi)
- |u_{\epsilon}|^{2p} u_{\epsilon} - (2p+1)|u_{\epsilon}|^{2p} (c u_{\epsilon}' + \psi ) \right).
$$
Using this expression for $F = F(c,\psi)$, we can interpret (\ref{integral-eq}) as an integral
equation for $\psi$ for a given $c$. Note that $F$ is $C^1$ in $c$ and $\psi$ if $p > 0$
and that $F$ and its first partial derivatives are zero at $c = 0$ and $\psi = 0$.
By the implicit function theorem, for all $c \in \mathbb{R}$ sufficiently small, there exists a unique solution
$\psi \in C^1(-L,L)$ of the integral equation (\ref{integral-eq}), which is $C^1$ in $c$ and satisfies
$\psi = \partial_c \psi |_{c = 0} =0$.

It follows from (\ref{expression-c}) that for every $a \in \mathbb{R}$, we have $c \to 0$ as $\epsilon \to 0$.
Consequently, $\| \psi \|_{C^1(-L,L)} = {\rm o}(c) \to 0$ as $\epsilon \to 0$.

Let us now recall
that $u = u_{\epsilon} + c u'_{\epsilon} + \psi$ satisfies the homogeneous second-order
differential equation
with the energy invariant (\ref{energy-invariant-higher}).
Thanks to the boundary conditions (\ref{bc-psi}), we have
\begin{equation}
\label{bc-u-plus}
u(\pm L) = c u_{\epsilon}'(L), \quad u'(\pm L) = u'_{\epsilon}(L) + \psi'(\pm L),
\end{equation}
where $|c| + |\psi'(\pm L)| \to 0$ as $\epsilon \to 0$.
We shall now prove that $\psi'(L) = \psi'(-L)$.

Assume $\psi'(L) \neq \psi'(-L)$ so that $u'(L) \neq u'(-L)$ for small $\epsilon > 0$.
Thanks to the energy invariant (\ref{energy-invariant-higher}),
each orbit on the phase plane $(u,u')$ intersects any vertical curve $u = u_0$ for a fixed small $u_0$ only twice,
symmetrically in the upper and lower half planes. If $u'(L) \neq u'(-L)$, then $u'(L) = -u'(L)$.
However, this contradicts (\ref{bc-u-plus}) with $u_{\epsilon}'(L) \neq 0$ and $\psi'(L) \to 0$ as $\epsilon \to 0$.
Therefore, $u'(L) = u'(-L)$, which implies that $\psi'(L) = \psi'(-L)$.

Finally, the boundary conditions in system (\ref{per-nls}) yield expression
(\ref{expression-c}) for $c$ and the following equation for $a$:
\begin{equation}
\label{bc-4}
\epsilon^{1+\frac{1}{p}} \phi_0'(a) = \psi'(L) - \psi'(-L) = 0.
\end{equation}
There is only one solution for $a$ such that $\phi_0'(a) = 0$
and this is $a = 0$. Hence $c$ is uniquely defined by (\ref{expression-c}) with $\phi_0(0) = 1$,
after which $w$ is uniquely defined by the solution of the integral equation (\ref{integral-eq})
with $F = F(c,\psi)$. This yields the asymptotic expression (\ref{expansion-1}).
By uniqueness, this constructed solution
with small $c$ and $\psi$ corresponds to the solution of Theorem \ref{theorem-higher-order}.
\end{proof}

\begin{remark}
By the construction of $\psi$ in Lemma \ref{lemma-higher-order}, the parameters $b$ and $c$ in
Theorem \ref{theorem-higher-order}
and Lemma \ref{lemma-higher-order} are different from each other. However,
it follows from (\ref{expansions-1}) and (\ref{expression-c}) that
$b = c + \mathcal{O}\left(\epsilon^{\frac{3}{p}}\right)$,
where $c = \mathcal{O}\left(\epsilon^{\frac{1}{p}}\right)$. \label{remark-b-c}
\end{remark}

To illustrate Remark \ref{remark-b-c} with an example, let us consider the particular case of the cubic nonlinearity with
$p = 1$. Then, $F(c,\phi)$ is a smooth function near $c = 0$ and $\psi = 0$
with the expansion
$$
F(c,\psi) = 6 u_{\epsilon} (c u_{\epsilon}' + \psi)^2 +  2 c^3 (c u_{\epsilon}' + \psi)^3.
$$
In this case, $u_{\epsilon}$ is a smooth function in $x$ and the derivatives of $u_{\epsilon}$ satisfies
the linear inhomogeneous equations
$$
M_{\epsilon} u_{\epsilon}'' = 12 u_{\epsilon} (u_{\epsilon}')^2
$$
and
$$
M_{\epsilon} u_{\epsilon}''' = 36 u_{\epsilon} u_{\epsilon}' u_{\epsilon}'' + 12 (u_{\epsilon}')^3.
$$
Therefore, we can construct a near-identity transformation
for the solution $\psi$ of the integral equation (\ref{integral-eq}) with $F = F(c,\psi)$ such that
\begin{equation}
\label{near-identity-w}
\psi = \frac{1}{2} c^2 u_{\epsilon}'' + \frac{1}{6} c^3 (u_{\epsilon}''' - \epsilon^2 u_{\epsilon}') +
\tilde{\psi},
\end{equation}
where $\tilde{\psi} \in C^1(-L,L)$ is uniquely determined and
satisfies the bound $\| \tilde{\psi} \|_{C^1(-L,L)} = \mathcal{O}(c^4)$.
Note that we have used
$u_{\epsilon}''(\pm L) = 0$ and $u_{\epsilon}'''(\pm L) = \epsilon^2 u'_{\epsilon}(L)$
to satisfy the boundary conditions (\ref{bc-psi}) for the solution (\ref{near-identity-w}). By comparing
the solution $u(x) = u_{\epsilon}(x+b)$
and the solution given by (\ref{decomposition-w}) and (\ref{near-identity-w}),
we obtain the correspondence between $b$ and $c$:
$$
b = c - \frac{1}{6} c^3 \epsilon^2 + \mathcal{O}(c^4) = c + \mathcal{O}(\epsilon^4),
$$
because $c = \mathcal{O}(\epsilon)$.

In the next two sections, we consider spectral and orbital stability of the bifurcating standing wave
solutions of the NLS equation (\ref{eq}) along the primary and higher branches.

\section{Stability of the primary branch}

Here we consider the orbital stability of the primary branch, the existence of which
is given by Theorem \ref{theorem-primary} for $\omega = -\epsilon^2$ with $\epsilon > 0$ sufficiently small.
To this end, we shall count the number of negative eigenvalues in the
operators $L_-$ and $L_+$ given by the spectral problems (\ref{L-minus-nls}) and (\ref{L-plus-nls}), where $(u,v)$ is the solution of
the boundary-value problem (\ref{stat-nls}) along the primary branch. After counting of the number of
negative eigenvalues, it is straightforward to apply the orbital stability theory from \cite{GSS1}.
The main result of this section is formulated in the following theorem.

\begin{theorem}
\label{theorem-stability-primary}
For $\omega = -\epsilon^2$ with $\epsilon > 0$ sufficiently small,
the primary branch of Theorem \ref{theorem-primary} is orbitally stable with respect to
the time evolution of the NLS equation for every $p\in (0,2)$
and orbitally unstable for every $p \in (2,\infty)$.
\end{theorem}

Recall that $L_{\pm}$ are self-adjoint operator on $L^2(-L,L) \times L^2(L,\infty)$
with the domain $D(\Delta)$ given by (\ref{domain-Delta}). Since $L_{\pm}$ differs from $L_0 := -\Delta - \omega$ by a
bounded potential with the exponential decay to zero as $x \to \infty$ (hence, it is a relatively compact
perturbation to $L_0$), the absolutely continuous spectra of $L_{\pm}$ and $L_0$ (denoted
by $\sigma_c$) coincide. Moreover, the spectrum of $L_0$ is purely continuous,
so that $\sigma_c(L_{\pm}) = \sigma_c(L_0) = \sigma(L_0) = [-\omega,\infty)$.
Since the primary branch is defined for $\omega < 0$, the absolutely continuous spectrum of
$L_{\pm}$ is bounded from below by the number $-\omega > 0$.
Thus, for every $\omega < 0$, negative and zero eigenvalues of $L_{\pm}$ are isolated
from $\sigma_c(L_{\pm})$, hence we can count the number of these eigenvalues
with the account of their multiplicity. The following lemma reports the corresponding
result for the operator $L_-$.

\begin{lemma}
Let $\omega = -\epsilon^2$ and $\epsilon > 0$ be sufficiently small.
Operator $L_-$ is positive and $0$ is a simple isolated eigenvalue
with eigenfunction $(U,V) = (u,v)$. \label{lemma-L-minus}
\end{lemma}

\begin{proof}
By comparing (\ref{stat-nls}) and (\ref{L-minus-nls}),
we find that $(U,V) = (u,v)$ is an eigenvector of
the spectral problem (\ref{L-minus-nls}) for $\lambda = 0$.
Theorem \ref{theorem-primary} implies that for $\epsilon > 0$
sufficiently small, $u(x) > 0$ for all $x \in [-L,L]$ and $v(x) > 0$ for all $x \geq L$.

To show that $0$ is a simple isolated eigenvalue at the bottom of the
spectrum of $L_-$, we consider the energy quadratic form
associated with $L_-$:
\begin{eqnarray*}
E(U,V) & = & \int_{-L}^L \left[ \left( \frac{dU}{dx} \right)^2 + \epsilon^2 U^2 - (p+1) |u|^{2p} U^2 \right] dx \\
& \phantom{t} & + \int_{L}^{\infty} \left[ \left( \frac{dV}{dx} \right)^2 + \epsilon^2 V^2 - (p+1) |v|^{2p} V^2 \right] dx.
\end{eqnarray*}
Let us consider the representation
\begin{equation}
\label{representation-U-V}
U(x) = a(x) u(x), \quad V(x) = b(x) v(x).
\end{equation}
It is well-defined because $u$ and $v$ are positive for all admissible $x$. If $(U,V) \in \mathcal{D}(\Delta)$ is an eigenvector
of $L_-$ for $\lambda < 0$, then $b(x)$ and $b'(x)$ decay exponentially to zero as $x \to \infty$,
whereas if $(U,V) \in \mathcal{D}(\Delta)$  is an eigenvector of $L_-$ for $\lambda \in [0,\epsilon^2)$, then
$b(x)$ and $b'(x)$ may grow but $b(x) v(x)$ and $b'(x) v(x)$ still decay exponentially to zero as $x \to \infty$.

Substituting (\ref{representation-U-V}) into $E(U,V)$,
integrating by parts for any $(U,V) \in \mathcal{D}(\Delta)$, and
using the stationary system (\ref{stat-nls}), we obtain
$$
E(U,V) = \int_{-L}^L \left( \frac{da}{dx} \right)^2 u^2 dx
+ \int_{L}^{\infty} \left( \frac{d b}{dx} \right)^2 v^2 dx \geq 0.
$$
Therefore, no negative eigenvalues of $L_-$ exists and
the zero eigenvalue occurs if and only if
$a$ and $b$ are constant in $x$. Thus,
the eigenvector $(U,V) = (u,v)$ for the zero eigenvalue is unique
up to the constant multiplication factor.
\end{proof}

To deal with the spectral problem (\ref{L-plus-nls}) for the operator $L_+$, we use the scaling
transformation $\omega = -\epsilon^2$ and $\lambda = \epsilon^2 \Lambda$
together with the representations (\ref{scaling-v}) and (\ref{scal-transf}) for
the stationary solution $(u,v)$. As a result, the spectral problem (\ref{L-plus-nls}) is rewritten
in the equivalent form
\begin{equation}
\label{L-plus-nls-scale}
\left\{ \begin{array}{l} - U''(z) + U(z) - (2p+1)(p+1) |\psi(z)|^{2p} U(z) = \Lambda U(z), \quad z \in (-\epsilon L,\epsilon L), \\
- V''(z) + V(z) - (2p+1)(p+1) |\phi(z)|^{2p} V(z) = \Lambda V(z), \quad z \in (0,\infty), \\
U(\epsilon L) = U(-\epsilon L) = V(0), \\
U'(\epsilon L)-U'(-\epsilon L) = V'(0),
\end{array} \right.
\end{equation}
where we use the same notations $(U,V)$ for rescaled functions
$U(\epsilon x)$ and $V(\epsilon (x-L))$.

The absolute continuous spectrum of the operator $L_+$ for $\sigma_c(L_+) = [\epsilon^2,\infty)$
is now scaled to the absolutely continuous spectrum of the spectral problem (\ref{L-plus-nls-scale})
for $\Lambda \in [1,\infty)$.
Therefore, we shall focus on isolated eigenvalues of the spectral problem (\ref{L-plus-nls-scale}) for $\Lambda < 1$.

Recall that $\phi(z) = \phi_0(z+a)$, where $\phi_0(z) = {\rm sech}^{\frac{1}{p}}(pz)$,
see (\ref{soliton-sech}).
It is well known that the scalar Schr\"{o}dinger spectral problem on the line
\begin{equation}
\label{L-plus-scalar}
- V''(z) + V(z) - (2p+1)(p+1) {\rm sech}^2(pz) V(z) = \Lambda V(z), \quad z  \in \mathbb{R},
\end{equation}
admits a finite number of isolated eigenvalues (see, e.g., pp.103--105 in \cite{Titch}). Because $V(z) = \phi_0'(z)$
is the eigenfunction of the spectral problem (\ref{L-plus-scalar}) for $\Lambda = 0$
and $\phi_0'$ has only one zero on the real line, Sturm's nodal theorem
(see, e.g., Lemma 4.2 on p. 201 in \cite{Pel})  implies that
the spectral problem (\ref{L-plus-scalar}) has
exactly one negative eigenvalue, say $\Lambda_0 < 0$, a simple zero eigenvalue, and the rest of the spectrum
is bounded from below by a positive number $\Lambda_1$ (which coincides
with either the next positive eigenvalue or the bottom of the absolutely continuous spectrum at $1$).
The eigenfunction for the negative eigenvalue $\Lambda_0$ is even and strictly
positive and the eigenfunction for the zero eigenvalue is odd.
Given these preliminary facts, we prove the following technical result.

\begin{lemma}
For every $\Lambda \in (-\infty,1)$, there exists a unique $C^{\infty}$ solution of
the differential equation (\ref{L-plus-scalar}) on $(z_0,\infty)$ for every $z_0 \in \mathbb{R}$
that decays to zero as $z \to \infty$ and satisfies the boundary condition
\begin{equation}
\label{bc-plus-scalar}
\lim_{z \to \infty} V(z) e^{\sqrt{1-\Lambda} z} = 1.
\end{equation}
Denote this solution by $V_{\infty}(z;\Lambda)$. Then, the function
\begin{equation}
\label{Evans-function}
F(\Lambda) := \frac{V'_{\infty}(0;\Lambda)}{V_{\infty}(0;\Lambda)},
\end{equation}
where $V'_{\infty}$ is the derivative of $V_{\infty}$ with respect to the first argument,
is $C^{\infty}$ for every $\Lambda \in (-\infty,0)$ and admits a unique simple zero on $(-\infty,0)$
at $\Lambda = \Lambda_0 < 0$. \label{lemma-Sturm}
\end{lemma}

\begin{proof}
Using Green's function, we look for the decaying solution of
the differential equation (\ref{L-plus-scalar}) satisfying the boundary condition (\ref{bc-plus-scalar})
for any $\Lambda < 1$ from a suitable solution of the inhomogeneous integral equation
\begin{equation}
V(z) = e^{-\sqrt{1-\Lambda} z} - \frac{(2p+1)(p+1)}{\sqrt{1-\Lambda}} \int_{z}^{+\infty}
\sinh\left( \sqrt{1-\Lambda} (z-y) \right) {\rm sech}^2(py) V(y) dy.
\end{equation}
Denoting $W(z):= V(z) e^{\sqrt{1-\Lambda} z}$, we rewrite the integral equation
in the form
\begin{equation}
W(z) = 1 - \frac{(2p+1)(p+1)}{2 \sqrt{1-\Lambda}} \int_{z}^{+\infty}
\left( e^{2 \sqrt{1-\Lambda} (z-y)} - 1 \right) {\rm sech}^2(py) W(y) dy.
\end{equation}
Since the kernel of the integral equation is bounded for every $y \geq z$
and the potential term ${\rm sech}^2(py)$ is absolutely integrable,
existence and uniqueness of a bounded solution $W \in L^{\infty}(z_0, \infty)$
for every fixed $z_0 \in \mathbb{R}$ follows by the standard methods (see, e.g., Lemma 4.1 on pp.
199-200 in \cite{Pel}). The solution is $C^{\infty}$ for all $z$ on $(z_0,\infty)$ and
all $\Lambda$ on $(-\infty,1)$. Therefore, the unique smooth solution $V_{\infty}$  of
the differential equation (\ref{L-plus-scalar}) satisfying the boundary condition (\ref{bc-plus-scalar})
exists.

Next, we consider the function $F(\Lambda)$ defined by (\ref{Evans-function}).
This function is $C^{\infty}$ on $(-\infty,0)$ if and only if $V_{\infty}(0;\Lambda)$ is nonzero.
Assume that $V_{\infty}(0;\Lambda) = 0$ for some $\Lambda < 1$. Since the differential
equation (\ref{L-plus-scalar}) has even potential, the decaying
function $V_{\infty}(z;\Lambda)$ for $z \in [0,\infty)$ is extended as the odd
solution of the spectral problem (\ref{L-plus-scalar}) decaying at both $z \to \pm \infty$. Therefore,
it is an odd eigenfunction. However, as explained above, the smallest eigenvalue
with odd eigenfunction is located at $\Lambda = 0$. Therefore, $V_{\infty}(0;\Lambda) \neq 0$
for every $\Lambda \in (-\infty,0)$ and $F \in C^{\infty}(-\infty,0)$.

Finally, we prove that $F(\Lambda) = 0$ has only one simple zero on $(-\infty,0)$ and this zero coincides
with the negative eigenvalue $\Lambda_0$. Assume that $V_{\infty}'(0;\Lambda) = 0$ for some $\Lambda \in (-\infty,0)$.
Then, the decaying function $V_{\infty}(z;\Lambda)$ for $z \in [0,\infty)$ is extended as the even
solution of the spectral problem (\ref{L-plus-scalar}) decaying at both $z \to \pm \infty$. Therefore,
it is an even eigenfunction and $\Lambda$ is an eigenvalue. As explained above, there is only one negative eigenvalue
$\Lambda_0$ of the spectral problem (\ref{L-plus-scalar}).
Therefore, the zero of $V_{\infty}'(0;\Lambda) = 0$ occurs at $\Lambda = \Lambda_0$.

To prove that $\Lambda_0$ is a simple zero of $F$, we assume that $F'(\Lambda_0) = 0$ and obtain a contradiction.
Since $F(\Lambda_0) = 0$, the condition $F'(\Lambda_0) = 0$ is true if and only if $\partial_{\Lambda} V_{\infty}'(0;\Lambda_0) = 0$.
Define $\Psi(z) := \partial_{\Lambda} V_{\infty}(z;\Lambda_0)$. From the boundary condition at $z = 0$ and
the decay behavior (\ref{bc-plus-scalar}), we have $\Psi'(0) = 0$
and $\Psi(z;\Lambda_0) \to 0$ as $z \to \infty$. Simultaneously, differentiating
the spectral problem (\ref{L-plus-scalar}) in $\Lambda$, we obtain the inhomogeneous problem
for $\Psi$:
\begin{equation}
\label{L-plus-scalar-Psi}
- \Psi''(z) + \Psi(z) - (2p+1)(p+1) {\rm sech}^2(pz) \Psi(z) = \Lambda_0 \Psi(z) + V_{\infty}(z;\Lambda_0), \quad z  \in \mathbb{R}.
\end{equation}
Since $V_{\infty}(z;\Lambda_0)$ is even and $\Psi'(0) = 0$, $\Psi$ is extended as the even
solution of the inhomogeneous equation (\ref{L-plus-scalar-Psi}) decaying at both $z \to \pm \infty$.
Therefore, $\Psi \in L^2(\mathbb{R})$. However, existence of such solution contradicts
to the Fredholm theory for the self-adjoint spectral problem (\ref{L-plus-scalar})
with a simple eigenvalue $\Lambda_0$. Therefore, no $\Psi \in L^2(\mathbb{R})$
exists, and $F'(\Lambda_0) = 0$ is impossible. Thus, $\Lambda_0$ is a simple zero of $F$.
\end{proof}

We are now ready to count the negative and zero eigenvalue of the operator $L_+$
in the spectral problem (\ref{L-plus-nls}), which is rescaled as the spectral
problem (\ref{L-plus-nls-scale}).

\begin{lemma}
Let $\omega = -\epsilon^2$ and $\epsilon > 0$ be sufficiently small.
Operator $L_+$ has exactly one negative eigenvalue and no zero
eigenvalues. \label{lemma-L-plus}
\end{lemma}

\begin{proof}
We prove that the negative eigenvalue of the scalar spectral problem
(\ref{L-plus-scalar}) on the line persists in the spectral problem (\ref{L-plus-nls-scale}),
whereas the zero eigenvalue of (\ref{L-plus-scalar})
disappears for any $\epsilon > 0$ sufficiently small. Our proof relies on several claims.\\

{\bf Claim 1:} For every $\Lambda \in \mathbb{R}$, there exists a unique
even solution of the first equation in system (\ref{L-plus-nls-scale})
normalized by $U(0) = 1$. Denote it by $U_1(z;\Lambda)$. The solution $U_1$ is
$C^{\infty}$ both in $z$ and $\Lambda$.\\

{\bf Proof of Claim 1:} Because (\ref{L-plus-nls-scale}) is linear and $|\psi(z)|^{2p}$ is even in $z$,
the boundary condition $U(\epsilon L) = U(-\epsilon L)$ can be satisfied if and only if $U$ is even in $z$.
The even solution is uniquely determined by the initial value $U(0) = 1$ and $U'(0) = 0$.
Since $\psi(z) > 0$ for all $z \in [-\epsilon L,\epsilon L]$, as it follows
from the proof of Theorem \ref{theorem-primary}, the linear equation has smooth coefficients,
so that the unique even solution $U$ is smooth in $z$, that is, $U \in C^{\infty}(-\epsilon L, \epsilon L)$.
In particular, from $\psi(z) = 1 + \mathcal{O}(z^2)$ as $z \to 0$,
we can find the quadratic approximation for the solution:
\begin{equation}
\label{root-expansion}
U_1(z;\Lambda) = 1 - \frac{1}{2} (\Lambda - 1 + (p+1)(2p+1)) z^2 + \mathcal{O}(z^4) \quad \mbox{\rm as} \quad z \to 0.
\end{equation}
The solution $U_1$ is also smooth in $\Lambda$ because the linear equation is smooth in $\Lambda$. \\

{\bf Claim 2:} For every $\Lambda \in (-\infty,1)$, there exists a unique solution of the second equation
in system (\ref{L-plus-nls-scale}) that decays to zero as $z \to +\infty$ and satisfies
the boundary condition (\ref{bc-plus-scalar}). Denote it by $V_1(z;\Lambda)$. The solution is
$C^{\infty}$ both in $z$ and $\Lambda$.\\

{\bf Proof of Claim 2:} The existence of the unique smooth solution $V_1(z;\Lambda)$
for all $z \in \mathbb{R}_+$ that decays to zero
as $z \to +\infty$ and satisfies
(\ref{bc-plus-scalar}) follows by Lemma \ref{lemma-Sturm} since $\phi(z) = \phi_0(z+a)$
and $z_0$ in Lemma \ref{lemma-Sturm} is arbitrary. \\

{\bf Claim 3:}  For every $\Lambda \in (-\infty,1)$, there exists a unique square-integrable solution of the
spectral problem (\ref{L-plus-nls-scale}) in the form
\begin{equation}
\label{root-construction}
U = \frac{V_1(0;\Lambda)}{U_1(\epsilon L;\Lambda)} U_1(z;\Lambda), \quad V = V_1(z;\Lambda),
\end{equation}
if and only if the value of $\Lambda$ satisfies the algebraic equation
\begin{equation}
\label{root-Lambda}
\frac{V_1'(0;\Lambda)}{V_1(0;\Lambda)} = \frac{2 U_1'(\epsilon L;\Lambda)}{U_1(\epsilon L;\Lambda)}.
\end{equation}

{\bf Proof of Claim 3:} The solution of the first three equations of system (\ref{L-plus-nls-scale}) in the form (\ref{root-construction})
follows from Claims 1 and 2. It follows from expansion (\ref{root-expansion}) that
$U_1(\epsilon L; \Lambda) = 1 + \mathcal{O}(\epsilon^2) \neq 0$ as $\epsilon \to 0$,
hence the solution (\ref{root-construction}) is bounded and exponentially decaying as $z \to \infty$,
that is, it is square integrable. Finally, the algebraic equation (\ref{root-Lambda})
is obtained from the last equation in system (\ref{L-plus-nls-scale}).\\

We shall now use the construction in Claim 3 and prove that the
spectral problem (\ref{L-plus-nls-scale}) has a unique negative eigenvalue
and no zero eigenvalues. It follows from expansion (\ref{root-expansion}) and
the algebraic equation (\ref{root-Lambda}) that
\begin{equation}
\label{root-Lambda-expansion-1}
\frac{V_1'(0;\Lambda)}{V_1(0;\Lambda)} =  - 2L (\Lambda - 1 + (p+1)(2p+1)) \epsilon + \mathcal{O}(\epsilon^3) \quad
\mbox{\rm as} \quad \epsilon \to 0.
\end{equation}
Therefore, $V_1'(0;\Lambda) \to 0$ as $\epsilon \to 0$. Also recall that $\phi(z) = \phi_0(z+a)$
and $a = 2L \epsilon + \mathcal{O}(\epsilon^3)$ from Theorem \ref{theorem-primary} so that $a \to 0$
as $\epsilon \to 0$.
In the limit $\epsilon \to 0$, the condition $V_1'(0;\Lambda) = 0$ is satisfied
for the only value of $\Lambda$ on $(-\infty,\Lambda_1)$, where $\Lambda_1 \in (0,1)$ is defined above,
and this value coincides with the negative eigenvalue $\Lambda_0 < 0$ of the reduced spectral problem (\ref{L-plus-scalar})
on the line (in which case, the eigenfunction of (\ref{L-plus-scalar}) denoted by $V_0$ is even in $z$
and strictly positive for all $z \in \mathbb{R}$). Hence, no zero eigenvalue exists in the spectral
problem (\ref{L-plus-nls-scale}).

To prove persistence of the negative eigenvalue, we note again that
$\phi(z) = \phi_0(z+a)$, therefore, there exists a positive constant $C(a)$ such that
\begin{equation}
\label{representation-V-1}
V_1(z;\Lambda_0) = C(a) V_0(z+a),
\end{equation}
where $V_0$ is the eigenfunction of (\ref{L-plus-scalar}) for $\Lambda = \Lambda_0$.
The constant $C(a)$ is determined from the normalization condition (\ref{bc-plus-scalar}) for $V_1(z;\Lambda_0)$.
Since $V_0(z) > 0$ for every $z \in \mathbb{R}$, we note that
for any $a_0 > 0$ there is $C_0 > 0$ such that $C(a) \geq C_0$ for all $a \in [-a_0,a_0]$.

Now, using smoothness of the unique solution $V_1$ in Claim 2 in $\Lambda$
and the representation (\ref{representation-V-1}), we obtain
\begin{equation}
\label{root-Lambda-expansion-2}
\frac{V_1'(0;\Lambda)}{V_1(0;\Lambda)} = \frac{V_0'(a)}{V_0(a)} + (\Lambda - \Lambda_0) \frac{\partial}{\partial \Lambda}
\frac{V_1'(0;\Lambda_0)}{V_1(0;\Lambda_0)}
+ \mathcal{O}((\Lambda - \Lambda_0)^2) \quad \mbox{\rm as} \quad \Lambda \to \Lambda_0,
\end{equation}
where $V_0(a) > 0$ and $V_0'(a) = \mathcal{O}(a) = \mathcal{O}(\epsilon)$ as $\epsilon \to 0$.
By Lemma \ref{lemma-Sturm}, we have
$$
\lim_{\epsilon \to 0} \frac{\partial}{\partial \Lambda}
\frac{V_1'(0;\Lambda_0)}{V_1(0;\Lambda_0)} = F'(\Lambda_0) \neq 0.
$$
By the implicit function theorem, for $\epsilon > 0$
sufficiently small, there exists a unique root of
the algebraic equation (\ref{root-Lambda-expansion-1})
in $\Lambda$ such that $\Lambda = \Lambda_0 + \mathcal{O}(\epsilon)$.
\end{proof}

We can now proceed with the proof of Theorem \ref{theorem-stability-primary}.

\begin{proof1}{\em of Theorem \ref{theorem-stability-primary}.}
We apply the standard orbital stability theory from \cite{GSS1}.
The eigenvalue count in Lemmas \ref{lemma-L-minus} and \ref{lemma-L-plus}
gives exactly one negative eigenvalue of operator $L_+$
and a simple zero eigenvalue of operator $L_-$. The gauge invariance
of the NLS equation (\ref{eq0}) is used to construct a constrained $L^2$-space,
where the negative eigenvalue of $L_+$ becomes a positive eigenvalue if
$\partial_{\omega} (\| u \|_{L^2(-L,L)}^2 + \| v \|_{L^2(L,\infty)}^2) < 0$
and remains a negative eigenvalue if $\partial_{\omega} (\| u \|_{L^2(-L,L)}^2 + \| v \|_{L^2(L,\infty)}^2) > 0$,
where $(u,v)$ is the stationary solution along the primary branch.
The latter condition is sometimes referred to as the {\em slope condition}.

Therefore, we compute the slope condition for the primary branch in Theorem \ref{theorem-primary}:
\begin{eqnarray*}
\| u \|_{L^2(-L,L)}^2 = 2 L \epsilon^{\frac{2}{p}}(1 + \mathcal{O}(\epsilon^2))
\end{eqnarray*}
and
\begin{eqnarray*}
\| v \|_{L^2(L,\infty)}^2 = \epsilon^{\frac{2}{p}-1} \| \phi_0 \|^2_{L^2(a,\infty)},
\end{eqnarray*}
where $a = 2L \epsilon + \mathcal{O}(\epsilon^3)$ and $\phi_0$ is $\epsilon$-independent.

If $p \in (0,2)$, then $\partial_{\epsilon} (\| u \|_{L^2(-L,L)}^2 + \| v \|_{L^2(L,\infty)}^2) > 0$
for $\epsilon > 0$ sufficiently small, which implies that
$\partial_{\omega} (\| u \|_{L^2(-L,L)}^2 + \| v \|_{L^2(L,\infty)}^2) < 0$. This computation
yields the assertion on the orbital stability of the primary branch for $p \in (0,2)$.

If $p \in (2,\infty)$, then $\partial_{\epsilon} (\| u \|_{L^2(-L,L)}^2 + \| v \|_{L^2(L,\infty)}^2) < 0$
for $\epsilon > 0$ sufficiently small, which implies that
$\partial_{\omega} (\| u \|_{L^2(-L,L)}^2 + \| v \|_{L^2(L,\infty)}^2) > 0$. This
computation yields the assertion on the orbital instability of the primary branch if $p \in (2,\infty)$.
\end{proof1}

\begin{remark}
If $p = 2$, then we have the critical case with $\partial_{\epsilon} \| u \|_{L^2(-L,L)}^2 = 2L + \mathcal{O}(\epsilon^2)$
and $\partial_{\epsilon} \| v \|_{L^2(L,\infty)}^2 = -2L + \mathcal{O}(\epsilon^2)$.
Therefore, the test for orbital stability is inconclusive without computations
of the $\mathcal{O}(\epsilon^2)$ corrections in these expansions.
\end{remark}

\section{Stability of the higher branches}

Here we consider the linearized stability of the higher branches,
the existence of which is given by Theorem \ref{theorem-higher-order}
for $\omega = -\epsilon^2$ with $\epsilon > 0$ sufficiently small.

We linearize the NLS equation (\ref{eq0}) around the standing wave $e^{i\omega t} \Phi$, where
$\Phi=(u,v)$ is a suitable solution of the stationary NLS equation (\ref{eq}).
We write $\Psi = e^{i \omega t}(\Phi + U + i W)$, where
real-valued functions $U$ and $W$ are defined on the tadpole graph subject to
the same Kirchhoff boundary conditions. This yields the linearized evolution problem in the form
$$
\frac{d}{dt} \left[ \begin{array}{c} U \\ W \end{array} \right] = \left[ \begin{array}{cc} 0 & 1 \\ -1 & 0 \end{array} \right] \;
\left[ \begin{array}{cc} L_+ & 0 \\ 0 & L_- \end{array} \right] \left[ \begin{array}{c} U \\ W \end{array} \right]
= \left[  \begin{array}{cc}   0 & L_{-} \\   -L_{+}  & 0  \end{array}\right] \left[ \begin{array}{c} U \\ W \end{array} \right],
$$
where the operators $L_{\pm}$ are the same linear self-adjoint operators as before.
The spectral stability problem can be written as the coupled vector system
\begin{equation}
\label{spect-prob}
L_+ U = -\lambda W, \quad L_- W = \lambda U, \quad U, W \in \mathcal{D}(\Delta).
\end{equation}

The stationary solution $\Phi=(u,v)$ is said to be spectrally unstable if
there exist an isolated eigenvalue $\lambda$ with  ${\rm Re}(\lambda) > 0$
for the spectral problem (\ref{spect-prob}), in which case the eigenvalue is
referred to as unstable. The stationary solution $\Phi=(u,v)$ is said to be
weakly spectrally stable if the spectrum of the spectral problem (\ref{spect-prob}) is contained within the imaginary axis.
We note that isolated eigenvalues of the spectral problem (\ref{spect-prob}) are symmetric about the real
and imaginary axes.

The spectral problem (\ref{spect-prob}) is not self-adjoint because of the symplectic
matrix relating components $U$ and $W$. This is a well known source of difficulty but
important information on the unstable eigenvalues in the spectral problem (\ref{spect-prob})
can be derived from the spectral properties of operators
$L_-$ and $L_+$ (see, e.g., Chapter 4 in \cite{Pel}).
To proceed with this analysis, we count the number of negative eigenvalues of the operators $L_-$ and $L_+$
given by the spectral problems (\ref{L-minus-nls}) and (\ref{L-plus-nls}), where $(u,v)$ is the solution of
the boundary-value problem (\ref{stat-nls}) along the higher branches. After
counting of the number of negative eigenvalues, we apply the spectral instability theory from \cite{Gr}
to study eigenvalues of the spectral stability problem (\ref{spect-prob}).
The main result of this section is formulated in the following theorem.

\begin{theorem}
\label{theorem-stability-higher-order}
For $\omega = -\epsilon^2$ with $\epsilon > 0$ sufficiently small,
all higher branches of Theorem \ref{theorem-higher-order}
are spectrally unstable with at least one pair (two pairs) of real eigenvalues $\lambda$
in the spectral stability problem (\ref{spect-prob}) for $p \in (0,2]$ (respectively, $p \in (2,\infty)$).
\end{theorem}

To develop the count of negative eigenvalues associated to the higher branches $(u,v)$
(see Lemma \ref{lemma-L-minus-plus} below), we need to obtain the analogous count of negative eigenvalues
associated to the stationary solutions $(u^{\pm}_{n,\omega},0)$ described in Proposition
\ref{proposition-countable}. Recall that the higher branches $(u,v)$ of Theorem \ref{theorem-higher-order} bifurcate
from the standing wave solutions $(u_{n,\omega}^{\pm},0)$ of Proposition \ref{proposition-countable}.

In the following Propositions \ref{proposition-stability} and \ref{proposition-second-group},
we also count the number of negative eigenvalues in the operators $L_-$ and $L_+$
given by the spectral problems (\ref{L-minus-nls}) and (\ref{L-plus-nls}), where $(u,v) = (u^{\pm}_{n,\omega},0)$.
Unfortunately, the count does not give a conclusive stability result for these stationary states
(see Remark \ref{remark-branches} below), therefore, we only
formulate the conjecture on their instability for small negative $\omega$.
Motivations for posing this conjecture are explained in the end of this section.

\begin{conjecture}
\label{conjecture-stability-higher-order}
For $\omega = -\epsilon^2$ with $\epsilon > 0$ sufficiently small,
the branch $(u,v) = (u_{n,\omega}^{\pm},0)$ is spectrally unstable for any $n \in \mathbb{N}$
with at least $n$ quartets of complex eigenvalues $\lambda$
in the spectral stability problem (\ref{spect-prob}).
\end{conjecture}

Associated with the solution $(u,v) = (u_{n,\omega}^{\pm},0)$, we study negative and zero eigenvalues
of operators $L_-$ and $L_+$ given by the spectral problems
\begin{equation}
\label{L-scalar}
\left\{ \begin{array}{l} - U''(x) - \omega U - (p+1) |u^{\pm}_{n,\omega}|^{2p} U = \lambda U, \quad x \in (-L,L), \\
- V''(x) - \omega V = \lambda V, \quad \quad \quad \quad \quad \quad \quad \quad \quad \; x\in (L,\infty), \\
U(L) = U(-L) = V(L), \\
U'(L) - U'(-L) = V'(L),
\end{array} \right.
\end{equation}
and
\begin{equation}
\label{L-scalar-2}
\left\{ \begin{array}{l} - U''(x) - \omega U - (2p+1)(p+1) |u^{\pm}_{n,\omega}|^{2p} U = \lambda U, \quad x \in (-L,L), \\
- V''(x) - \omega V = \lambda V, \quad \quad \quad \quad \quad \quad \quad \quad \quad \quad \quad \quad \quad x\in (L,\infty),  \\
U(L) = U(-L) = V(L), \\
U'(L) - U'(-L) = V'(L).
\end{array} \right.
\end{equation}

Analysis of eigenvalues of the spectral problems (\ref{L-scalar}) and (\ref{L-scalar-2}) relies
on the analysis of Schr\"{o}dinger operators with $2L$-periodic coefficients:
\begin{equation}
\label{Sturm-operators}
M_- := -\partial_x^2 -\omega - (p+1) |u^{\pm}_{n,\omega}|^{2p} : \quad H^2_{\rm per}(-L,L) \to L^2_{\rm per}(-L,L).
\end{equation}
and
\begin{equation}
\label{Sturm-operators-2}
M_+ := -\partial_x^2 -\omega - (2p+1)(p+1) |u^{\pm}_{n,\omega}|^{2p} : \quad H^2_{\rm per}(-L,L) \to L^2_{\rm per}(-L,L).
\end{equation}

For the operator $M_-$, there exist two fundamental solutions of the second-order differential equation $M_- \psi = 0$.
One solution in the form $\psi(x) = u_{n,\omega}^{\pm}(x)$ is $2L$-periodic in $x$ and the other solution
is available in the explicit form
\begin{equation}
\label{psi-minus}
\psi(x) = u_{n,\omega}^{\pm \prime}(x) + p u_{n,\omega}^{\pm}(x) \int_0^x | u_{n,\omega}^{\pm}(y)|^{2p} dy.
\end{equation}
It is then obvious that the second solution (\ref{psi-minus}) is not $2L$-periodic for every $\omega \in (-\infty,\omega_n)$.

For the operator $M_+$ (which coincides with the operator $M_{\epsilon}$ in Section 4 if $\omega = -\epsilon^2$),
there exist again two fundamental solutions of the second-order differential equation $M_+ \psi = 0$.
One solution in the form $\psi(x) = u_{n,\omega}^{\pm \prime}(x)$ is $2L$-periodic in $x$ and the other solution
is available in the implicit form
\begin{equation}
\label{psi-plus}
\psi(x) = \partial_E U_{n,\omega}^{\pm}(x;E_{\omega}),
\end{equation}
where  $U_{n,\omega}^{\pm}(x;E)$ is the continuation of $u_{n,\omega}^{\pm}$ as the odd solution of the boundary--value
problem (\ref{stat-nls-scalar}) with respect to
the energy invariant (\ref{energy-invariant}) for fixed $\omega \in (-\infty,\omega_n)$
and $E_{\omega}$ is defined by the root of $T(E_{\omega}) = L$ with $T(E)$ being the half-period
of the odd solution $U_{n,\omega}^{\pm}(x;E)$.
Recall from Lemmas \ref{lemma-period} and \ref{lemma-period-2}
that the period-to-energy map is a $C^1$ diffeomorphism with
$T'(E) < 0$ for every fixed $\omega \in (-\infty,\omega_n)$. It follows from the boundary conditions (\ref{bc-solutions})
that the second linear independent solution (\ref{psi-plus})
is not $2L$-periodic for every $\omega \in (-\infty,\omega_n)$.

Equipped with these preliminary facts, we analyze the spectral problems (\ref{L-scalar}) and (\ref{L-scalar-2}).
Eigenvalues of these spectral problems can be divided into two groups.
The first group is characterized by the reduction $V(x) = 0$ for all $x \geq L$
and the other group has eigenfunctions with nonzero $V$. The following two propositions
give the relevant counts of negative and zero eigenvalues in these two groups.

\begin{proposition}
For any $\omega < \omega_n := \frac{\pi^2 n^2}{L^2}$, there exist
exactly $(n-1)$ negative simple eigenvalues and a zero simple eigenvalue
in the spectral problem (\ref{L-scalar}) with $V \equiv 0$ and
exactly $n$ negative simple eigenvalues and no zero eigenvalue
in the spectral problem (\ref{L-scalar-2}) with $V \equiv 0$.
The corresponding eigenfunctions $U \in H^2(-L,L)$ are odd and $(2L)$-periodic, that is,
$U \in H^2_{\rm per, odd}(-L,L)$.
\label{proposition-stability}
\end{proposition}

\begin{proof}
We use the Sturm theory to identify negative eigenvalues of the operators $M_-$ and $M_+$
given by (\ref{Sturm-operators}) and (\ref{Sturm-operators-2}).
Note that since $u_{n,\omega}^- = -u_{n,\omega}^+$, the spectra of these operators are identical for the
two members of the double set. Since $u^{\pm}_{n,\omega} \in H^2_{\rm per, odd}(-L,L)$,
the operators $M_-$ and $M_+$ are invariant under the change $x \to -x$, therefore,
their $(2L)$-periodic eigenfunctions are either even or odd.
If $V \equiv 0$, then $U$ satisfies the boundary conditions
$U(-L) = U(L) = 0$ and $U'(-L) = U'(L) \neq 0$ if and only if $U$ is odd and $(2L)$-periodic,
that is, if $U \in H^2_{\rm per, odd}(-L,L)$. Also note that $|u^{\pm}_{n,\omega}|^{2p}$
is actually $L$-periodic, therefore, the $2L$-periodic eigenfunctions
are either $L$-periodic or $L$-antiperiodic.

Regarding operator $M_-$, we have $M_- u^{\pm}_{n,\omega} = 0$, where $u_{n,\omega}^{\pm}$ is $(2L)$-periodic, odd, and
has $2n-1$ zeros on $(-L,L)$. By Floquet--Sturm's theory (see Theorem 1.3.4 in \cite{Eastham}),
there exist at least $2n-1$ and at most $2n$ negative eigenvalues of the operator $M_-$
corresponding to $(2L)$-periodic eigenfunctions. The lowest negative eigenvalue corresponds to an even
positive state. The other negative eigenvalues occur in pairs and each pair corresponds to eigenfunctions
of different parity (one even and one odd). Since $u_{n,\omega}^{\pm}$ is odd, exactly $n-1$ negative eigenvalues correspond
to odd eigenfunctions, whereas the other (either $n$ or $n+1$)
negative eigenvalues correspond to even eigenfunctions.
The assertion of the proposition about the negative and zero eigenvalues
of the spectral problem (\ref{L-scalar}) with $V \equiv 0$ is proved.

Regarding operator $M_+$, we have $M_+ u_{n,\omega}^{\pm \prime} = 0$, where
$u^{\pm \prime}_{n,\omega}$ is $(2L)$-periodic,
even, and has $2n$ zeros on $(-L,L)$. By the same Floquet--Sturm's theory
(see Theorem 1.3.4 in \cite{Eastham}), there exist exactly $n$ negative eigenvalues with even eigenfunctions and
either $n-1$ or $n$ negative eigenvalues with odd eigenfunctions.
Therefore, we only need to check in the last pair of eigenvalues
if the eigenvalue for the odd eigenfunction is located to the left or
to the right of the zero eigenvalue for the even eigenfunction $u^{\pm \prime}_{n,\omega}$.

Recall here from Proposition \ref{proposition-countable}
that the branch $(u_{n,\omega}^{\pm},0)$ originates
from the local bifurcation at $\omega = \omega_n$
with the limiting solution (\ref{solution-bifurcation}).
Because
$$
M_+ = M_- - 2p (p+1) |u^{\pm}_{n,\omega}|^{2p}
$$
and $M_- u_{n,\omega}^{\pm} = 0$ with odd $u_{n,\omega}^{\pm}$, it
is clear that the zero eigenvalue is double at $\omega = \omega_n$
but it splits for $\omega \lesssim \omega_n$ in such a way that
the eigenvalue for an odd eigenfunction of $M_+$ in the corresponding pair is located
on the left from the zero eigenvalue for the even eigenfunction $u^{\pm \prime}_{n,\omega}$.

As argued above, the operator $M_+$ given by (\ref{Sturm-operators-2}) admits exactly one
$2L$-periodic eigenfunction for the zero eigenvalue for every $\omega \in (-\infty,\omega_n)$.
Therefore, once the splitting happens for
$\omega \lesssim \omega_n$, the negative eigenvalue for an odd eigenfunction of $M_+$
cannot cross the zero eigenvalue and is hence located
on the left from the zero eigenvalue for the even eigenfunction $u^{\pm \prime}_{n,\omega}$
for every $\omega < \omega_n$.
As a result, exactly $n$ negative eigenvalues  of $M_+$ correspond to
odd eigenfunctions. The assertion of the proposition about the negative and zero eigenvalues
of the spectral problem (\ref{L-scalar-2}) with $V \equiv 0$ is proved.
\end{proof}

\begin{remark}
Recall that the operator $M_-$ given by (\ref{Sturm-operators}) has also exactly
one $(2L)$-periodic eigenfunction  for the zero eigenvalue for every $\omega \in (-\infty,\omega_n)$.
Also recall that
$$
M_- = M_+ + 2p (p+1) |u^{\pm}_{n,\omega}|^{2p}
$$
with $M_+ u_{n,\omega}^{\pm \prime} = 0$, where $u^{\pm \prime}_{n,\omega}$ is even.
Repeating the last argument in the proof of Proposition \ref{proposition-stability} for the operator $M_-$,
we also conclude that the double zero eigenvalue of $M_-$ at $\omega = \omega_n$
splits in such a way that the eigenvalue for an even eigenfunction of $M_-$ is located
on the right from the zero eigenvalue for the odd eigenfunction $u^{\pm}_{n,\omega}$ for every $\omega \in (-\infty,\omega_n)$.
Therefore, exactly $n$ negative eigenvalues of the operator $M_-$ correspond to even
eigenfunctions of the spectral problem (\ref{L-scalar}) with $V \equiv 0$. \label{remark-splitting}
\end{remark}

\begin{proposition}
For every $\omega < \omega_n$, isolated eigenvalues of the spectral problems (\ref{L-scalar}) and (\ref{L-scalar-2}) with
nonzero $V$ correspond to the real eigenvalues $\lambda$ such that $\lambda + \omega < 0$
of the following spectral problems
\begin{equation}
\label{L-scalar-second-group}
\left\{ \begin{array}{l} - U''(x) - \omega U - (p+1) |u^{\pm}_{n,\omega}|^{2p} U = \lambda U, \quad x \in (-L,L), \\
U(L) = U(-L), \\
U'(L) - U'(-L) = - U(L) \sqrt{|\omega + \lambda|}
\end{array} \right.
\end{equation}
and
\begin{equation}
\label{L-scalar-second-group-2}
\left\{ \begin{array}{l} - U''(x) - \omega U - (2p+1) (p+1) |u^{\pm}_{n,\omega}|^{2p} U = \lambda U, \quad x \in (-L,L), \\
U(L) = U(-L), \\
U'(L) - U'(-L) = - U(L) \sqrt{|\omega + \lambda|}.
\end{array} \right.
\end{equation}
The corresponding eigenfunctions $U$ are even but not $(2L)$-periodic, that is,
$$U \notin H^2_{\rm per, even}(-L,L).$$ Furthermore, for $\omega = -\epsilon^2$ with $\epsilon > 0$ sufficiently small,
there exist exactly $n$ negative eigenvalues in the spectral problems (\ref{L-scalar-second-group})
and (\ref{L-scalar-second-group-2}). No zero eigenvalues exist in either spectral problem.
\label{proposition-second-group}
\end{proposition}

\begin{proof}
We recall that the operators $L_+$ and $L_-$ are self-adjoint when they are considered
from the domain $\mathcal{D}(\Delta)$ to $L^2(-L,L) \times L^2(L,\infty)$.
Therefore, the values of $\lambda$ are real. The continuous spectrum is located for $\lambda > -\omega$. In what follows,
we shall only consider isolated eigenvalues $\lambda$ such that $\lambda + \omega < 0$.
If $\lambda + \omega < 0$, there exists a one-parameter family of
decaying solutions for the second equation of the spectral problems (\ref{L-scalar}) and (\ref{L-scalar-2})
as $x \to +\infty$, in fact, in the explicit form
$$
V(x) = V(L) e^{-\sqrt{|\omega+\lambda|} (x-L)}, \quad x \geq L.
$$
Using the boundary condition $V(L) = U(L)$, we arrive to the spectral problems (\ref{L-scalar-second-group})
and (\ref{L-scalar-second-group-2}).
Since $|u^{\pm}_{n,\omega}|^{2p}$ is even in $x$, eigenfunctions of these spectral problems are either even or odd.

If $U(L) = 0$, then $V(x) = 0$ for all $x \geq L$ and we are back
to the case considered in Proposition \ref{proposition-stability}.
Therefore, $U(L) \neq 0$. Odd eigenfunctions violate
the boundary condition $U(L) = U(-L) \neq 0$. Therefore, $U$
is even. Because $U(L) \neq 0$, it follows that $U'(L) \neq U'(-L)$,
therefore, the eigenfunction $U$ is not $(2L)$-periodic, that is,
$U \notin H^2_{\rm per, even}(-L,L)$.

Next, we give the precise count of negative and zero eigenvalues of the spectral problems
(\ref{L-scalar-second-group}) and (\ref{L-scalar-second-group-2}) for $\omega = -\epsilon^2$ with
$\epsilon > 0$ sufficiently small.

First, we note that for $\omega = \omega_n$ when $u_{n,\omega_n}^{\pm} = 0$,
there are no eigenvalues of the spectral problems
(\ref{L-scalar-second-group}) and (\ref{L-scalar-second-group-2}) with $\lambda < -\omega_n$.
Indeed, in this case, the even solution the differential equations
is known in the explicit form
\begin{equation}
\label{exact-U}
U(x) = U(0) \cosh\left(\sqrt{|\omega_n + \lambda|} x\right), \quad x \in (-L,L),
\end{equation}
whereas the last boundary condition yields the equation $2 \tanh\left(\sqrt{|\omega_n + \lambda|} L\right) = -1$,
which admits no solutions.

On the other hand, for $\omega = \omega_n$, there exist a certain number of
negative eigenvalues of the self-adjoint operators $M_-$ and $M_+$ in
(\ref{Sturm-operators}) and (\ref{Sturm-operators-2}),
which correspond to even and $(2L)$-periodic eigenfunctions.
From the proof of Proposition \ref{proposition-stability} and Remark \ref{remark-splitting},
we know that the operators $M_-$ and  $M_+$ admit exactly $n$ eigenvalues with even $2L$-periodic
eigenfunctions. Because the nonlinear terms in the spectral problems (\ref{L-scalar-second-group}) and (\ref{L-scalar-second-group-2})
are $C^1$ if $p > 0$, the negative eigenvalues of $M_-$ and $M_+$ are $C^1$ functions
of $\omega$ for $\omega \in (-\infty,\omega_n)$. All $n$ negative eigenvalues are above
the anti-diagonal $\lambda = -\omega$ at $\omega = \omega_n$, according to the previous result based
on the exact solution (\ref{exact-U}). By continuity, each negative eigenvalue intersects transversely
with the anti-diagonal $\lambda = -\omega$
at least once (and, in any case, an odd number of times) when $\omega$ changes from $\omega_n$
to $0$.

Let $\mu(\omega)$ denotes a particular negative eigenvalue of either $M_-$ or $M_+$ with the corresponding
even eigenfunction $U(x;\omega)$. For definiteness, let us consider operator $M_-$.
Let $\omega_0$ be the point of a particular intersection $\mu(\omega_0) = -\omega_0 < 0$
(which does not need to be transverse). According to the location
of $\mu(\omega)$ above or below the diagonal $\lambda = -\omega$, we claim
in the following Table \ref{TableEigenvalue} the location of a new isolated eigenvalue $\lambda$ of
the spectral problem (\ref{L-scalar-second-group}) with $\omega$
in the neighborhood of the intersection point $\omega = \omega_0$.

\begin{table}[htdp]
\begin{center}
\begin{tabular}{|c|c|c|}
\hline
$\omega \lesssim \omega_0$ & $\omega \gtrsim \omega_0$ & location of $\lambda$\\
\hline
$\mu(\omega) < -\omega$ & $\mu(\omega) < -\omega$ & $\omega \lesssim \omega_0$ \mbox{\rm and} $\omega \gtrsim \omega_0$\\
\hline
$\mu(\omega) < -\omega$ & $\mu(\omega) > -\omega$ & $\omega \lesssim \omega_0$ \\
\hline
$\mu(\omega) > -\omega$ & $\mu(\omega) < -\omega$ & $\omega \gtrsim \omega_0$\\
\hline
$\mu(\omega) > -\omega$ & $\mu(\omega) > -\omega$ & \\
\hline
\end{tabular}
\end{center}
\caption{Locations of a new isolated eigenvalue $\lambda$ of
the spectral problem (\ref{L-scalar-second-group}) with
$\lambda < -\omega$ near $\omega = \omega_0$.}
\label{TableEigenvalue}
\end{table}

Table \ref{TableEigenvalue} implies that if a particular eigenvalue $\mu(\omega)$
of $M_-$ has only one transverse intersection with the anti-diagonal $\lambda = -\omega$ when $\omega$
changes from $\omega_n$ to $0$,
then the spectral problem (\ref{L-scalar-second-group}) acquires one
negative eigenvalue $\lambda$ at $\omega = 0$. If $\mu(\omega)$ has an odd number of
transverse intersections  with the anti-diagonal $\lambda = -\omega$,
the spectral problem (\ref{L-scalar-second-group}) still acquires
only one negative eigenvalue $\lambda$ at $\omega = 0$. The other
intermediate intersections lead to an even number of
appearances and disappearances of negative eigenvalues $\lambda$ below the anti-diagonal $\lambda = -\omega$.
The tangential intersections, if they occur, do not change the outcome at $\omega = 0$.
By continuity of eigenvalues $\lambda$ of the spectral problem (\ref{L-scalar-second-group}) with respect to $\omega$,
this argument yields the last assertion of the lemma for negative eigenvalues
of the spectral problem (\ref{L-scalar-second-group}). A similar
count can be developed for the spectral problem (\ref{L-scalar-second-group-2}).\\

{\bf Proof of Table \ref{TableEigenvalue}:}
Let $\mu(\omega)$ be a particular negative eigenvalue of $M_-$  with the corresponding
even eigenfunction $U(\cdot;\omega)$ in $L^2_{\rm per}(-L,L)$.
Using $C^1$ smoothness of solutions of the first equation in system (\ref{L-scalar-second-group})
in $\lambda$, we write
\begin{equation}
\label{bif-1}
\lambda = \mu(\omega) + \Lambda, \quad U(x) = U(x;\omega) + \Lambda \tilde{U}(x;\omega) + {\rm o}(\Lambda),
\end{equation}
where $\tilde{U}$ is an even solution of the inhomogeneous equation
\begin{equation}
\label{bif-2}
M_- \tilde{U}(x;\omega) = U(x;\omega)
\end{equation}
subject to the orthogonality condition $\langle U(\cdot;\omega), \tilde{U}(\cdot;\omega) \rangle_{L^2(-L,L)} = 0$.
Because $U(\cdot;\omega)$ is even in $x$ and $C^1$, we have $U(L;\omega) \neq 0$ and
$U'(L;\omega) = 0$. On the other hand, since $\| U(\cdot;\omega) \|^2_{L^2_{\rm per}} \neq 0$,
the even function $\tilde{U}$ is not $(2L)$-periodic
because
$$
\tilde{U}'(-L;\omega) = -\tilde{U}'(L;\omega) \neq 0.
$$
Indeed, multiplying the
inhomogeneous equation (\ref{bif-2}) by $U(\cdot;\omega)$ and integrating on $(-L,L)$, we obtain
\begin{equation}
\label{bif-3}
-2 U(L;\omega) \tilde{U}'(L;\omega) = \| U(\cdot;\omega) \|^2_{L^2_{\rm per}} \neq 0.
\end{equation}

The first boundary condition $U(L) = U(-L)$ in system (\ref{L-scalar-second-group})
is satisfied by the construction of even functions. The second boundary condition
in system (\ref{L-scalar-second-group}) leads to the algebraic equation
\begin{equation}
\label{bif-4}
2 \Lambda \tilde{U}'(L;\omega) + {\rm o}(\Lambda) = -
\sqrt{|\omega + \mu(\omega) + \Lambda|} \left( U(L;\omega) + \mathcal{O}(\Lambda) \right).
\end{equation}
In view of the previous relation (\ref{bif-3}), equation (\ref{bif-4}) yields
\begin{equation}
\label{bif-5}
\Lambda \| U(\cdot;\omega) \|^2_{L^2_{\rm per}} + {\rm o}(\Lambda) =
\sqrt{|\omega + \mu(\omega) + \Lambda|} \left( U^2(L;\omega) + \mathcal{O}(\Lambda) \right).
\end{equation}

Recall that $\omega_0 + \mu(\omega_0) = 0$ and the mapping $\omega \mapsto \mu$ is $C^1$.
It follows from the algebraic equation (\ref{bif-5}) that $\Lambda = \mathcal{O}(\omega + \mu(\omega))$
is small if $\omega - \omega_0$ is small. Expanding $\Lambda$ in powers of $\omega + \mu(\omega)$
by using the algebraic equation (\ref{bif-5}), we obtain
\begin{equation}
\label{bif-6}
\Lambda = - (\omega + \mu(\omega)) - (\omega + \mu(\omega))^2
\frac{\| U(\cdot;\omega) \|^4_{L^2_{\rm per}}}{U^4(L;\omega)} + {\rm o}((\omega + \mu(\omega))^2).
\end{equation}
It follows from (\ref{bif-5}) that the new eigenvalue exists only if $\Lambda > 0$,
which implies that $\omega + \mu(\omega) < 0$ near $\omega_0 + \mu(\omega_0) = 0$.
Since
$$
\lambda + \omega = \omega + \mu(\omega) + \Lambda = - (\omega + \mu(\omega))^2
\frac{\| U(\cdot;\omega) \|^4_{L^2_{\rm per}}}{U^4(L;\omega)} + {\rm o}((\omega + \mu(\omega))^2) < 0,
$$
the new eigenvalue $\lambda$ is isolated from the continuous spectrum of the spectral problem (\ref{L-scalar-second-group}).
Thus, a new isolated eigenvalue of the spectral problem (\ref{L-scalar-second-group}) bifurcates at the intersection
$\mu(\omega_0) = -\omega_0 < 0$ near $\omega = \omega_0$ in the subset of $\omega$, where $\omega + \mu(\omega) < 0$.
This perturbation argument yields the statement of Table \ref{TableEigenvalue}. \\

It remains to consider the zero eigenvalue of the spectral problems
(\ref{L-scalar-second-group}) and (\ref{L-scalar-second-group-2}) for $\omega = -\epsilon^2$ for $\epsilon > 0$
sufficiently small. Since $M_- u_{n,\omega}^{\pm} = 0$
with the odd $(2L)$-periodic function $u_{n,\omega}^{\pm}$, no bifurcations
of a new isolated eigenvalue $\lambda$ corresponding to an even eigenfunction may occur in the spectral
problem (\ref{L-scalar-second-group}) near $\omega = 0$.
On the other hand, since $M_+ u_{n,\omega}^{\pm \prime} = 0$
with the even $(2L)$-periodic function $u_{n,\omega}^{\pm \prime}$, bifurcations
of new small negative eigenvalues $\lambda$ may occur in the spectral
problem (\ref{L-scalar-second-group-2}) near $\omega = 0$.
This bifurcation can be considered as in the proof of Table \ref{TableEigenvalue} but with
$\mu(\omega) = 0$ and $U(x;\omega) = u_{n,\omega}^{\pm \prime}(x)$.
With parametrization $\omega = -\epsilon^2$, the
algebraic equation (\ref{bif-6}) yields now
\begin{equation}
\label{bif-7}
\Lambda = \epsilon^2 - \frac{\| u_{n,\omega}^{\pm \prime}\|^4_{L^2_{\rm per}}}{u_{n,\omega}^{\pm \prime}(L)} \epsilon^4 + {\rm o}(\epsilon^4),
\end{equation}
which indicates that a new isolated eigenvalue $\lambda = \Lambda > 0$
does exists for $\omega \lesssim 0$. Since the new eigenvalue is positive,
no zero (or additional negative) eigenvalues in the spectral problem  (\ref{L-scalar-second-group-2})
bifurcates near $\omega = 0$. The statement of the lemma is proved.
\end{proof}

Next, we analyze negative and zero eigenvalues of the operators $L_{\pm}$
associated with the higher branches of Theorem \ref{theorem-higher-order}.
As in Theorem \ref{theorem-higher-order}, we only use $u_{\epsilon}$ to denote
the selected solution along the higher branch for $\omega = -\epsilon^2$ with small positive $\epsilon$.
Recall that the stationary solution $(u,v)$ is expressed by
$$
u = u_{\epsilon}(x + b), \quad v = \epsilon^{\frac{1}{p}} \phi_0(\epsilon (x-L)),
$$
where $b = \mathcal{O}(\epsilon^{\frac{1}{p}})$ is a unique root of $u_{\epsilon}(L+b) = \epsilon^{\frac{1}{p}}$
and $\phi_0$ is the normalized (even) solution given by (\ref{soliton-sech}).
The corresponding spectral problems for the operators $L_-$ and $L_+$ are given by
\begin{equation}
\label{L-plus-nls-scale-higher-order}
\left\{ \begin{array}{l} - U''(x) + \epsilon^2 U - (p+1) |u_{\epsilon}(x+b)|^{2p} U = \lambda U, \quad \quad \quad \; x \in (-L,L), \\
- V''(x) + \epsilon^2 V - \epsilon^2 (p+1) |\phi_0(\epsilon(x-L))|^{2p} V = \lambda V, \quad x \in (L,\infty), \\
U(L) = U(-L) = V(L), \\
U'(L)-U'(-L) = V'(L),
\end{array} \right.
\end{equation}
and
\begin{equation}
\label{L-plus-nls-scale-higher-order-2}
\left\{ \begin{array}{l} - U''(x) + \epsilon^2 U - (2p+1) (p+1) |u_{\epsilon}(x+b)|^{2p} U = \lambda U, \quad \quad \quad \; x \in (-L,L), \\
- V''(x) + \epsilon^2 V - \epsilon^2 (2p+1) (p+1) |\phi_0(\epsilon(x-L))|^{2p} V = \lambda V, \quad x \in (L,\infty), \\
U(L) = U(-L) = V(L), \\
U'(L)-U'(-L) = V'(L).
\end{array} \right.
\end{equation}
The following lemma summarizes the main result about the numbers of negative and zero eigenvalues
of operators $L_-$ and $L_+$.

\begin{lemma}
For $\epsilon > 0$ sufficiently small, the negative eigenvalues of
Propositions \ref{proposition-stability} and \ref{proposition-second-group}
persist as the negative eigenvalues of
the spectral problem (\ref{L-plus-nls-scale-higher-order}) and (\ref{L-plus-nls-scale-higher-order-2}).
In addition, the operator $L_+$ admits one more small negative eigenvalue
and the operator $L_-$ admits a simple zero eigenvalue.
\label{lemma-L-minus-plus}
\end{lemma}

\begin{proof}
We shall first prove persistence of negative eigenvalues of operators $L_{\pm}$
given by Propositions \ref{proposition-stability} and \ref{proposition-second-group}
for the branch $(u_{\epsilon},0)$. Recall that these negative eigenvalues are bounded away
from zero as $\epsilon \to 0$. We give an argument for $L_-$, the argument for $L_+$ is identical.

For any $\lambda < 0$, the first differential equation
of the system (\ref{L-plus-nls-scale-higher-order}) can be solved with two
linearly independent solutions
\begin{equation}
\label{decom-U-odd-even}
U(x) = c_1 U_{\rm odd}(x+b;\lambda) + c_2 U_{\rm even}(x+b;\lambda),
\end{equation}
where $(c_1,c_2)$ are arbitrary constants, $U_{\rm odd}$ and $U_{\rm even}$ are odd and even functions.
For uniqueness, we add the normalization conditions
$$
U_{\rm odd}'(0;\lambda) = 1, \quad U_{\rm even}(0;\lambda) = 1.
$$

For any fixed $\lambda < 0$ and sufficiently small $\epsilon > 0$,
the second differential equation of the system (\ref{L-plus-nls-scale-higher-order}) has a decaying solution
\begin{equation}
\label{decom-V-odd-even}
V(x) = d_1 V_{\rm dec}(x-L;\lambda),
\end{equation}
where $d_1$ is an arbitrary constant and $V_{\rm dec}$ is uniquely specified
by the decay condition
$$
\lim_{x \to +\infty} e^{\sqrt{\epsilon^2 - \lambda} x} V_{\rm dec}(x;\lambda) = 1.
$$
As $\epsilon = 0$, we have the unique representation $V_{\rm dec}(x;\lambda) = e^{-\sqrt{-\lambda} x}$ for $\lambda < 0$,
so that $V_{\rm dec}(0;\lambda) = 1$ for every $\lambda < 0$.

Substituting the representations (\ref{decom-U-odd-even}) and (\ref{decom-V-odd-even})
to the boundary conditions in the system (\ref{L-plus-nls-scale-higher-order}),
we obtain $d_1 = \frac{U(L)}{V_{\rm dec}(0;\lambda)}$, use symmetry properties for odd and even functions,
and derive the linear algebraic system for parameters $(c_1,c_2)$:
\begin{eqnarray*}
\left\{ \begin{array}{l}
c_1 \left[ U_{\rm odd}(L+b;\lambda) + U_{\rm odd}(L-b;\lambda) \right] +
c_2 \left[ U_{\rm even}(L+b;\lambda) - U_{\rm even}(L-b;\lambda) \right] = 0, \\
c_1 \left[ U'_{\rm odd}(L+b;\lambda) - U'_{\rm odd}(L-b;\lambda) \right] +
c_2 \left[ U'_{\rm even}(L+b;\lambda) + U'_{\rm even}(L-b;\lambda) \right] \\
\phantom{texttext} = \frac{V'_{\rm dec}(0;\lambda)}{V_{\rm dec}(0;\lambda)}
\left[ c_1 U_{\rm odd}(L+b;\lambda) + c_2 U_{\rm even}(L+b;\lambda) \right].
\end{array} \right.
\end{eqnarray*}
Note that $b = 0$ if $\epsilon = 0$.

The homogeneous linear system above at $\epsilon = 0$ has two groups of solutions, which
correspond to negative eigenvalues described in Propositions \ref{proposition-stability} and \ref{proposition-second-group}.
For the first group, $c_2 = 0$ and $U_{\rm odd}(L;\lambda_0) = 0$, where $\lambda_0 < 0$ is
an eigenvalue of Proposition \ref{proposition-stability}. For the second group,
$c_1 = 0$ and $U_{\rm even}'(L;\lambda_0) = -\sqrt{|\lambda_0|} U_{\rm even}(L;\lambda_0)$,
where $\lambda_0 < 0$ is an eigenvalue of Proposition \ref{proposition-second-group}.
Since all solutions used in the construction of the homogeneous linear system above are $C^1$ in $\lambda$ and $\epsilon$ at
$\lambda = \lambda_0 < 0$
and $\epsilon = 0$, persistence of negative eigenvalues follows from persistence of roots of
the characteristic equation associated with the homogeneous linear system for $(c_1,c_2)$.

It remains to consider small negative and zero eigenvalues of operators $L_{\pm}$,
which may bifurcate as $\epsilon \to 0$ from the zero eigenvalues in Propositions
\ref{proposition-stability} and \ref{proposition-second-group}. Here we first consider
the operator $L_+$ and then operator $L_-$.

By using the scaling transformation $z = \epsilon(x-L)$ and $\lambda = \epsilon^2 \Lambda$,
the second equation of the system (\ref{L-plus-nls-scale-higher-order-2}) can be rewritten
in terms of the variable $z$:
\begin{equation}
\label{L-scalar-nls-2}
-V''(z) + V - (2p+1) (p+1) |\phi_0(z)|^{2p} V = \Lambda V, \quad z > 0.
\end{equation}
By Lemma \ref{lemma-Sturm}, there exists a unique $C^{\infty}$ decaying
solution of the differential equation (\ref{L-scalar-nls-2}) for $\Lambda \in (-\infty,1)$
satisfying the decay condition
\begin{equation}
\label{bc-L-scalar-nls}
\lim_{z \to \infty} V(z)  e^{\sqrt{1 - \Lambda} z} = 1.
\end{equation}
Denote this solution by $V_{\infty}(z;\Lambda)$. Recall that $V_{\infty}$ is also $C^{\infty}$
in $\Lambda$ for every $\Lambda \in (-\infty,1)$. Also recall from Lemma \ref{lemma-Sturm}
that the spectral problem (\ref{L-scalar-nls-2}) with the boundary condition $V'(0) = 0$
has one negative eigenvalue $\Lambda_0 < 0$, no zero eigenvalue, and the rest of the spectrum
is bounded from below by a positive number $\Lambda_1 \in (0,1)$. Therefore, the only root of
$V_{\infty}'(0;\Lambda) = 0$ on $(-\infty,\Lambda_1)$
occurs at the negative eigenvalue $\Lambda_0 < 0$.

Given the unique $V$, the rest of the system (\ref{L-plus-nls-scale-higher-order-2}) is given by
\begin{equation}
\label{L-plus-nls-rest-2}
\left\{ \begin{array}{l} - U''(x) + \epsilon^2 U
- (2p+1) (p+1) |u_{\epsilon}(x+b)|^{2p} U = \epsilon^2 \Lambda U, \quad x \in (-L,L), \\
U(L) = U(-L) = V_{\infty}(0;\Lambda), \\
U'(L)-U'(-L) = \epsilon V_{\infty}'(0;\Lambda).
\end{array} \right.
\end{equation}
Small negative and zero eigenvalues of the spectral problem (\ref{L-plus-nls-rest-2}) such that
$\lambda = \epsilon^2 \Lambda \to 0$ as $\epsilon \to 0$
bifurcate from the zero eigenvalue of the linear operator
$$
M_+ := -\partial_x^2 - (2p+1) (p+1) |u_0(x)|^{2p} : \quad H^2_{\rm per}(-L,L) \to L^2_{\rm per}(-L,L),
$$
whose the only eigenfunction is known because of $M_+ u_0' = 0$.
If $\Lambda = 0$, the $(2L)$-periodic eigenfunction $u_0'$ is continued
as $u_{\epsilon}'(\cdot + b)$ but this eigenfunction does not satisfy the system
(\ref{L-plus-nls-rest-2}) because $V_{\infty}'(0;0) \neq 0$. Therefore,
zero is not an eigenvalue of the spectral problem (\ref{L-plus-nls-rest-2}).

Next, for every $\Lambda \in \mathbb{R}$, there exists a unique
solution of the differential equation in system (\ref{L-plus-nls-rest-2})
subject to the normalization $U(-L) = U(L) = 1$. This solution $U$ is $C^1$ in $x$ on $[-L,L]$
and in $\Lambda$ on $\mathbb{R}$.
From the solution for $\Lambda = 0$, the unique function is given by
$$
U(x;\Lambda) = \frac{u'_{\epsilon}(x+b)}{u'_{\epsilon}(L+b)} + \mathcal{O}_{C^1(-L,L)}(\epsilon^2 \Lambda)
\quad \mbox{\rm as} \quad \epsilon \to 0.
$$
Substituting this unique function to the boundary conditions in system (\ref{L-plus-nls-rest-2}),
we obtain the boundary condition
\begin{equation}
\label{boundary-value-L-scalar-nls}
F(\Lambda) := \frac{V_{\infty}'(0;\Lambda)}{V_{\infty}(0;\Lambda)}
= \frac{U'(L;\Lambda) - U'(-L;\Lambda)}{\epsilon} = \mathcal{O}(\epsilon \Lambda), \quad \mbox{\rm as} \quad \epsilon \to 0.
\end{equation}
By Lemma \ref{lemma-Sturm},
the function $F$ is $C^{\infty}$ in $\Lambda$ for every $\Lambda < 0$ and it admits only one simple
zero at $\Lambda = \Lambda_0$. By the implicit function theorem, the simple root persists in
algebraic equation (\ref{boundary-value-L-scalar-nls}) near $\Lambda_0$ with respect to small parameter
$\epsilon$. Therefore, the spectral problem (\ref{L-plus-nls-scale-higher-order-2})
admits no zero eigenvalue and exactly one negative eigenvalue $\Lambda$ near $\Lambda_0$ such that
$\Lambda = \Lambda_0 + \mathcal{O}(\epsilon)$.
The assertion of the lemma about the negative and zero eigenvalues of
the spectral problem (\ref{L-plus-nls-scale-higher-order-2}) is proved.

Now we consider the operator $L_-$. With the help of the same scaling transformation,
the second equation of the system (\ref{L-plus-nls-scale-higher-order}) can be rewritten
for the component $V$ as a function of $z$:
\begin{equation}
\label{L-scalar-nls}
-V''(z) + V - (p+1) |\phi_0(z)|^{2p} V = \Lambda V, \quad z > 0.
\end{equation}
By the same result as in Lemma \ref{lemma-Sturm}, there exists a unique $C^{\infty}$ decaying
solution of the differential equation (\ref{L-scalar-nls}) for $\Lambda \in (-\infty,1)$
satisfying the decay condition (\ref{bc-L-scalar-nls}). Again, we denote this solution
by $V_{\infty}(z;\Lambda)$ and recall that $V_{\infty}$ is also $C^{\infty}$
in $\Lambda$ for every $\Lambda \in (-\infty,1)$. Because the spectral problem (\ref{L-scalar-nls})
with the boundary condition $V'(0) = 0$ has no negative eigenvalues, a simple zero eigenvalue,
and the rest of the spectrum is bounded away from zero by a positive number $\Lambda_1 \in (0,1)$,
the only root of $V_{\infty}'(0;\Lambda)$ on $(-\infty,\Lambda_1)$ occurs at
the zero eigenvalue $\Lambda = 0$.

Given the unique $V$, the rest of the system (\ref{L-plus-nls-scale-higher-order}) is given by
\begin{equation}
\label{L-plus-nls-rest}
\left\{ \begin{array}{l} - U''(x) + \epsilon^2 U
- (p+1) |u_{\epsilon}(x+b)|^{2p} U = \epsilon^2 \Lambda U, \quad x \in (-L,L), \\
U(L) = U(-L) = V(0), \\
U'(L)-U'(-L) = \epsilon V'(0).
\end{array} \right.
\end{equation}
Small negative and zero eigenvalues of the spectral problem (\ref{L-plus-nls-rest})
with $\lambda = \epsilon^2 \Lambda \to 0$ as $\epsilon \to 0$
bifurcate from the zero eigenvalue of the linear operator
$$
M_- := -\partial_x^2 - (p+1) |u_0(x)|^{2p} : \quad H^2_{\rm per}(-L,L) \to L^2_{\rm per}(-L,L),
$$
whose the only $2L$-periodic eigenfunction is known because of $M_- u_0 = 0$.
If $\Lambda = 0$, the $(2L)$-periodic eigenfunction $u_0$ is continued
as $u_{\epsilon}(\cdot + b)$ and it satisfies (\ref{L-plus-nls-rest}) for any $\epsilon > 0$ because
$V_{\infty}'(0;0) = 0$. The simple zero eigenvalue persists in the spectral problem
(\ref{L-plus-nls-scale-higher-order}) with the exact solution
$(U,V) = (u_{\epsilon}(\cdot + b),\epsilon^{\frac{1}{p}} \phi_0(\epsilon(\cdot - L)))$
for $\lambda = 0$. No negative eigenvalues exists in (\ref{L-plus-nls-rest})
because $V_{\infty}'(0;\Lambda) \neq 0$ for $\Lambda \in (-\infty,0)$.
The assertion of the lemma about the negative and zero eigenvalues of
the spectral problem (\ref{L-plus-nls-scale-higher-order}) is proved.
\end{proof}

We can now proceed with the proof of Theorem \ref{theorem-stability-higher-order}. We note
the following result (see, e.g., Proposition 20 and decomposition (12) in \cite{FHK}).

\begin{proposition}
For every $p \in (0,2]$ and every $n \in \mathbb{N}$, the following slope condition is satisfied
for the entire family of the standing wave solutions $u_{n,\omega}^{\pm}$:
\begin{equation}
\label{asumption-negative}
\frac{d}{d \omega} \| u_{n,\omega}^{\pm} \|_{L^2(-L,L)} < 0, \qquad \omega \in (-\infty,\omega_n).
\end{equation}
\label{prop-slope}
\end{proposition}

\begin{proof1}{\em of Theorem \ref{theorem-stability-higher-order}.}
By Lemma \ref{lemma-L-minus-plus},
we have $2n-1$ negative eigenvalues of the
operator $L_-$ and $2n+1$ negative eigenvalues of the operator $L_+$,
associated with the higher-order branch of Theorem \ref{theorem-higher-order}.
Furthermore, it follows from Proposition \ref{proposition-countable} and Theorem \ref{theorem-higher-order} that
\begin{eqnarray*}
\| u \|_{L^2(-L,L)}^2 = \| u_{\epsilon}(\cdot + b) \|^2_{L^2(-L,L)} = \| u_{\epsilon} \|^2_{L^2(-L,L)} =
\| u_0 \|^2_{L^2(-L,L)} + \mathcal{O}(\epsilon^2)
\end{eqnarray*}
and
\begin{eqnarray*}
\| v \|_{L^2(L,\infty)}^2 = \epsilon^{\frac{2}{p}-1} \| \phi_0 \|^2_{L^2(0,\infty)}.
\end{eqnarray*}
In the constrained $L^2$-subspace, where $(U,V) \perp (u,v)$, the additional small negative eigenvalue of the operator $L_+$
in Lemma \ref{lemma-L-minus-plus} disappears
if
$$
\partial_{\epsilon} \left( \| u \|_{L^2(-L,L)}^2 + \| v \|_{L^2(L,\infty)} \right) > 0,
$$
which is only true for $p \in (0,2)$. In this case, the difference between the number of
negative eigenvalues of the operators $L_+$ and $L_-$ is exactly one and the stationary solution is spectrally unstable
with at least one pair of real eigenvalues in the spectral stability problem (\ref{spect-prob}),
according to the spectral instability theory in \cite{GSS2,Gr}.

If $p \in (2,\infty)$, we have $\partial_{\epsilon} \left( \| u \|_{L^2(-L,L)}^2 + \| v \|_{L^2(L,\infty)} \right) < 0$,
so that the difference  between the number of
negative eigenvalues of the operators $L_+$ and $L_-$
is exactly two. In this case, the stationary solution is spectrally unstable with
at least two pairs of real eigenvalues in the spectral stability problem (\ref{spect-prob}),
according to the spectral instability theory in \cite{Gr}.

Finally, if $p = 2$, then $\| v \|_{L^2(L,\infty)}$ is independent of $\epsilon$,
whereas $\| u_{\epsilon}\|^2_{L^2(-L,L)}$ satisfies the slope condition
(\ref{asumption-negative}) in Proposition \ref{prop-slope}. Therefore,
the difference between the number of
negative eigenvalues of the operators $L_+$ and $L_-$
is exactly one and the stationary solution is spectrally unstable
with exactly one pair of real eigenvalues in the spectral stability problem (\ref{spect-prob}).
\end{proof1}

\begin{remark}
By the count of negative eigenvalues in Propositions \ref{proposition-stability}
and \ref{proposition-second-group}, we have $2n-1$ negative eigenvalues of the
operator $L_-$ and $2n$ negative eigenvalues of the operator $L_+$, associated with
the family of the standing wave solutions $(u,v) = (u_{n,\omega}^{\pm},0)$ near $\omega = 0$.
Using the slope condition (\ref{asumption-negative}), we obtain that in the constrained $L^2$-subspace,
where $U \perp u_{n,\omega}$, the operator $L_+$ has $2n-1$ negative eigenvalues.
Hence, the difference between negative eigenvalues of $L_+$ and $L_-$ is exactly zero,
and the instability test for the branch $(u,v) = (u_{n,\omega}^{\pm},0)$ is inconclusive if $p \in (0,2]$.

If $p \in (2,\infty)$, the instability test is even more undecided because the slope condition (\ref{asumption-negative})
only holds for $\omega$ near $\omega_n$ and is definitely violated if $\omega \to -\infty$ \cite{FHK}.
When the slope condition (\ref{asumption-negative}) is violated, there exists at least one pair of real eigenvalues
in the spectral problem (\ref{spect-prob}).
\label{remark-branches}
\end{remark}

In the rest of this section, we motivate why the unperturbed branch $(u,v) = (u_{n,\omega}^{\pm},0)$
is expected to be unstable for small negative $\omega$, according to Conjecture \ref{conjecture-stability-higher-order}.
If the slope condition (\ref{asumption-negative}) is satisfied (which is the case for every $n \in \mathbb{N}$
if $p \in (0,2]$ \cite{FHK}), there exist exactly $2n-1$ negative eigenvalues of the operators $L_-$ and
$L_+$ (the latter operator is considered in the constrained space, where $U \perp u_{n,\omega}$).
Nevertheless, the stationary solution can be spectrally stable if there exist $2n-1$ pairs of
purely imaginary eigenvalues $\lambda$ of negative Krein signature in the spectral stability
problem (\ref{spect-prob}) (see, e.g., Theorem 4.5 in \cite{Pel}). Hence, in order to claim the spectral
instability of the stationary solutions, we shall rule out this possibility.

Thanks to the symmetry decompositions of eigenvectors
to the even and odd parts on the interval $[-L,L]$ (the same as in
Propositions \ref{proposition-stability} and \ref{proposition-second-group}),
one can divide all eigenvalues of the spectral stability problem (\ref{spect-prob})
into two groups corresponding to even and odd eigenfunctions. The odd eigenvectors
satisfy the following spectral stability problem
\begin{equation}
\label{L-vector-1}
\left\{ \begin{array}{l} - U''(x) - \omega U -  (2p+1) (p+1) |u^{\pm}_{n,\omega}|^{2p} U = -\lambda W, \quad x \in (0,L), \\
- W''(x) - \omega W -  (p+1) |u^{\pm}_{n,\omega}|^{2p} W = \lambda U, \quad \quad \quad \quad \;\; x \in (0,L), \\
U(0) = U(L) = 0, \\
W(0) = W(L) = 0. \end{array} \right.
\end{equation}
If the slope condition (\ref{asumption-negative}) is satisfied, the spectral problem
(\ref{L-vector-1}) can have at most $2n-2$ unstable eigenvalues, depending on the parameter $\omega$
in $(-\infty,\omega_n)$,
according to the count of $n-1$ negative eigenvalues of the operators $L_-$ and
$L_+$ (the latter operator is considered in the constrained space, where $U \perp u_{n,\omega}^{\pm}$),
see Proposition \ref{proposition-stability} and Remark \ref{remark-branches}.
For $n = 1$, the spectral stability problem (\ref{L-vector-1}) produces no unstable eigenvalues.
For $n \geq 2$, unstable eigenvalues may appear or disappear for various values of $\omega$
and the exact count of unstable eigenvalues for small negative values of $\omega$ is more difficult.

The even eigenvectors satisfy the following spectral stability problem
\begin{equation}
\label{L-vector-2}
\left\{ \begin{array}{l} - U''(x) - \omega U -  (2p+1)(p+1) |u^{\pm}_{n,\omega}|^{2p} U = -\lambda W, \quad x \in (0,L), \\
- W''(x) - \omega W - (p+1) |u^{\pm}_{n,\omega}|^{2p} W = \lambda U, \quad \quad \quad \quad \;\;  x \in (0,L), \\
- V''(x) - \omega V = -\lambda Z, \quad \quad \quad \quad \quad \quad \quad \quad \quad \quad \quad \quad \;\;  x \in (L,\infty), \\
- Z''(x) - \omega Z = \lambda W, \quad \quad \quad \quad \quad \quad \quad \quad \quad \quad \quad \quad \quad \;  x \in (L,\infty),\\
U'(0) = 0, \;\; W'(0) = 0,\\
U(L) = V(L), \;\; W(L) = Z(L),\\
2U'(L) = V'(L), \;\; 2 W'(L) = Z'(L).\end{array} \right.
\end{equation}
The spectral problem (\ref{L-vector-2}) can have at most $2n$ unstable eigenvalues for small negative $\omega$, according to the
count of $n$ negative eigenvalues in Proposition \ref{proposition-second-group}, unless $n$ pairs
of purely imaginary eigenvalues $\lambda$ of negative Krein signature occurs in the spectral stability
problem (\ref{L-vector-2}).
However, the continuous spectrum of the spectral stability problem (\ref{L-vector-2})
is located for two symmetric segments $\pm i \Sigma_c$, where
$$
\Sigma_c = \{ -\omega + k^2, \quad k \in \mathbb{R} \}.
$$
Therefore, for $\omega \geq 0$, $\Sigma_c \cup (-\Sigma_c) = \mathbb{R}$ contains no gap, so that
all purely imaginary eigenvalues $\lambda$, if they exist, are embedded into the continuous spectrum.
Since pairs of embedded eigenvalues of negative Krein signature are structurally unstable and bifurcate
into quartets of complex eigenvalues according to the spectral instability
theory in \cite{Gr}, the spectral stability problem (\ref{L-vector-2}) is expected to
have generically $n$ quartets of complex eigenvalues for small negative $\omega$.
This argument is reflected in Conjecture \ref{conjecture-stability-higher-order}.

\begin{remark}
Since the higher branch of Theorem \ref{theorem-higher-order} bifurcates off the branch
$(u,v) = (u^{\pm}_{n,\omega},0)$ for $\omega \lesssim 0$, the spectral problem (\ref{spect-prob}) is expected to have
as many quartets of complex eigenvalues or pairs of purely imaginary
eigenvalues of negative Krein signatures as the branch $(u,v) = (u^{\pm}_{n,\omega},0)$ does near $\omega = 0$.
In addition, Theorem \ref{theorem-stability-higher-order}
guarantees that the higher branch of Theorem \ref{theorem-higher-order} has one (two) pairs of real unstable eigenvalues
$\lambda$ in the spectral stability problem (\ref{spect-prob}) for $p \in (0,2]$ (respectively,
for $p \in (2,\infty)$).
\end{remark}

\section{Numerical results}

In this final section, we study existence and stability of standing waves on the tadpole graph numerically.
We shall confirm the results of Theorems \ref{theorem-stability-primary} and \ref{theorem-stability-higher-order}.
In addition, we shall illustrate
the validity of Conjecture \ref{conjecture-stability-higher-order}. For the NLS model (\ref{eq0}),
we consider the case $p = 1$, which corresponds to the cubic NLS equation on the tadpole.

For $p = 1$, the second-order differential equation
\begin{equation}
\label{num-1}
u''(x) + \omega u + 2 u^3 = 0
\end{equation}
can be solved analytically. Indeed, using the transformation
\begin{equation}
\label{num-0}
u(x) = k a v(\xi), \quad \xi = a x, \quad a = \sqrt{\frac{\omega}{1 - 2 k^2}},
\end{equation}
where $k \in (0,1)$ is a parameter, we transform equation (\ref{num-1}) to the form
\begin{equation}
\label{num-2}
v''(\xi) + (1-2k^2) v + 2k^2 v^3 = 0.
\end{equation}
The second-order equation (\ref{num-2}) is satisfied by the Jacobi elliptic function $v(\xi) = {\rm cn}(\xi;k)$ associated
with the parameter $k$ \cite{OLBC10}.
The Jacobi elliptic function ${\rm cn}(\cdot;k)$ is $4K(k)$-periodic, where
$K(k)$ is the complete elliptic integral of the first kind. If $u$ is $2L$-periodic, then
parameters $k$, $\omega$, and $L$ satisfy the relationship
\begin{equation}
\label{num-3}
4 n K(k) = 2 a L \quad \Rightarrow \quad \omega L^2 = 4 n^2 (1-2k^2) K(k)^2,
\end{equation}
where $n \in \mathbb{N}$ is the index for the corresponding branch of the $2L$-periodic solution.
As $k \to 0$, we have $K(k) \to \frac{\pi}{2}$, hence $\omega \to \omega_n = \frac{\pi^2 n^2}{L^2}$, according to the
result of Proposition \ref{proposition-countable}. As $k \to 1$, we have $K(k) \to \infty$,
hence $\omega \to -\infty$. At $k = \frac{1}{\sqrt{2}}$, we have $\omega = 0$.

Using the relation (\ref{num-3}) and translating the Jacobi elliptic function ${\rm cn}(\xi;k)$ in $\xi$
to satisfy the boundary condition $u(0) = 0$, we obtain the exact $2L$-periodic solution of
the second-order differential equation (\ref{num-1}) in the form
\begin{equation}
\label{num-4}
u^{\pm}_{n,\omega}(x) = \pm \frac{2n k K(k)}{L} \; {\rm cn}\left(\frac{2 n K(k)}{L} \; x + K(k);k\right).
\end{equation}

The exact solution (\ref{num-4}) is used as the seed solution for the Newton iterative algorithm to
approximate standing wave solutions of the boundary-value problem (\ref{stat-nls}) with $p = 1$.
We discretize the second-order differential equations with a second-order central difference method
and incorporate the Kirchhoff boundary conditions into the discretization method.
Figure \ref{fig0} shows the numerical approximations of the standing wave solutions $(u_{n,\omega}^+,0)$
with $n = 1$ (a) and $n = 2$ (b) corresponding to $\omega = -1$.
We have set $L = \pi$ and used $N = 100$ grid points on the interval $[-L,L]$.
For the same value $\omega = -1$, Figure \ref{fig-sol} shows the standing wave solutions $(u,v)$ with nonzero $v$
along the primary branch (a) and two representatives of the higher branches with $n = 1$ (b,c)
and $n = 2$ (d,e) bifurcating from the standing wave solutions $(u_{n,\omega}^{\pm},0)$ at $\omega = 0$.
To truncate the semi-infinite line $[L,\infty)$ on the finite interval $[L,L_{\infty}]$,
we have used $L_{\infty} = 2\pi$ and the Dirichlet boundary condition at $L_{\infty}$. The grid spacing
is uniform between $[-L,L]$ and $[L,L_{\infty}]$.

\begin{figure}[htbp]
\begin{center}
\includegraphics[scale=0.47]{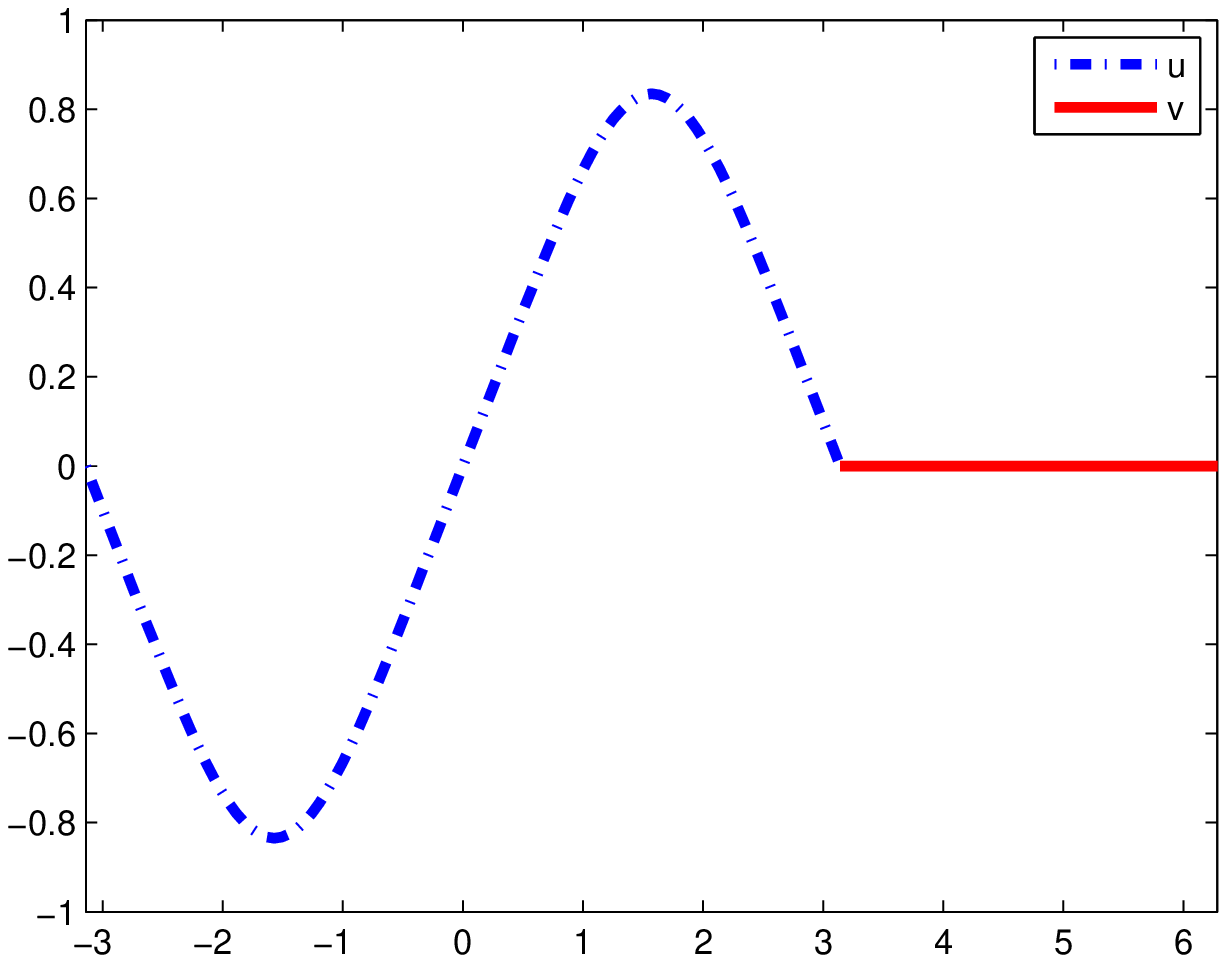}
\includegraphics[scale=0.47]{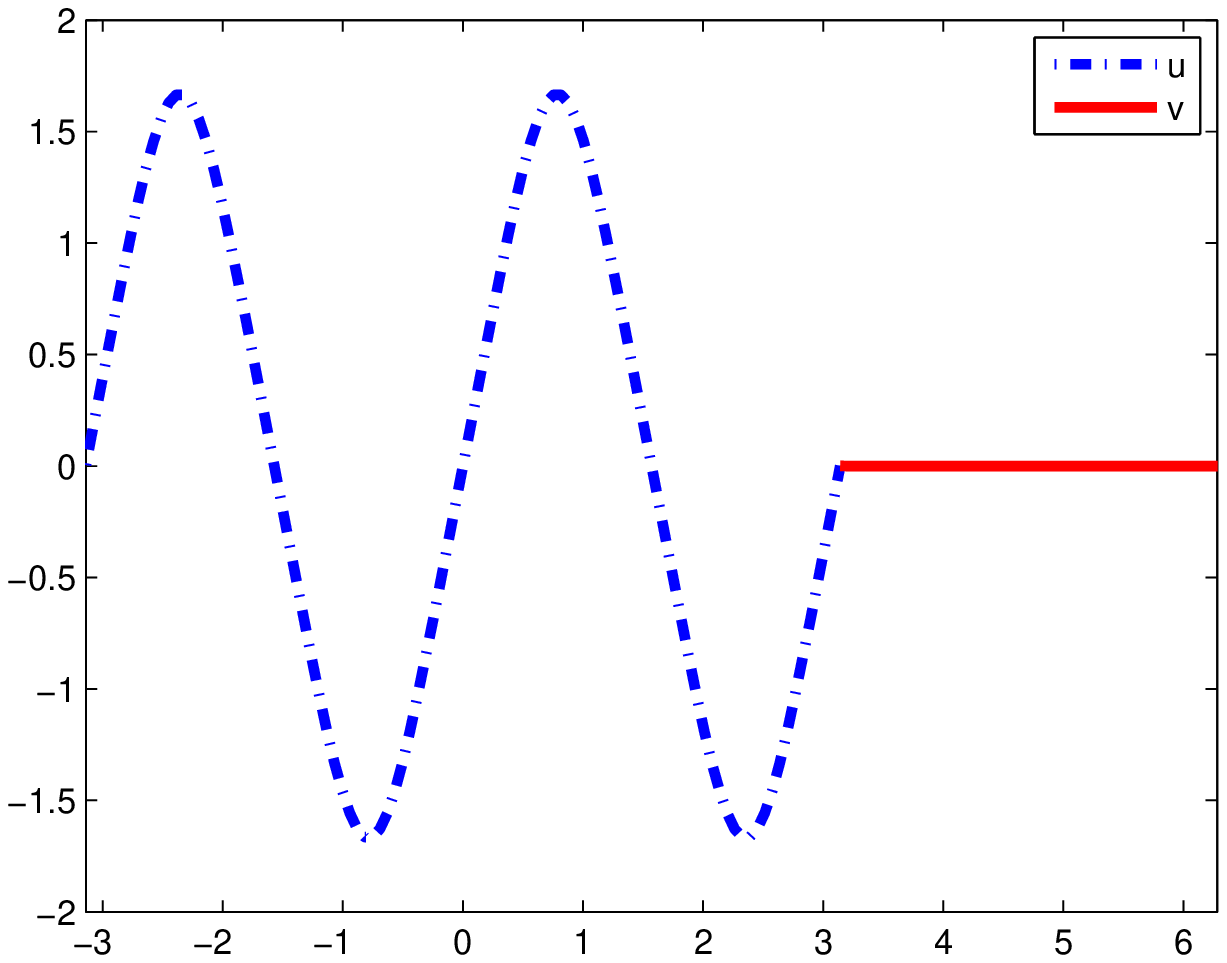}
\end{center}
\caption{Standing wave solutions $(u_{n,\omega}^+,0)$ versus $x$ for $n = 1$ (a) and $n = 2$ (b)
corresponding to $\omega = -1$.}
\label{fig0}
\end{figure}

\begin{figure}[htbp]
\begin{center}
\includegraphics[scale=0.53]{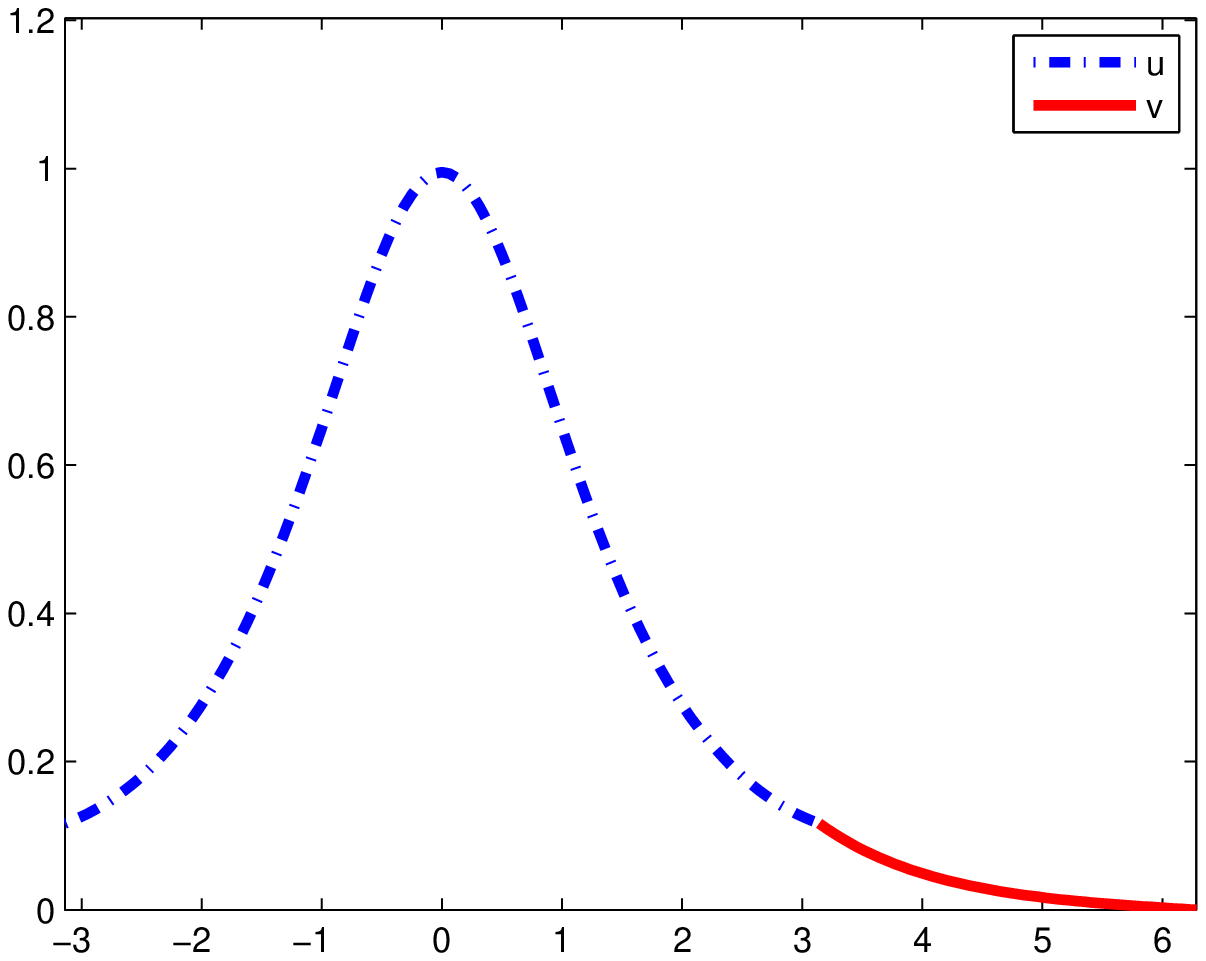}\\
\includegraphics[scale=0.47]{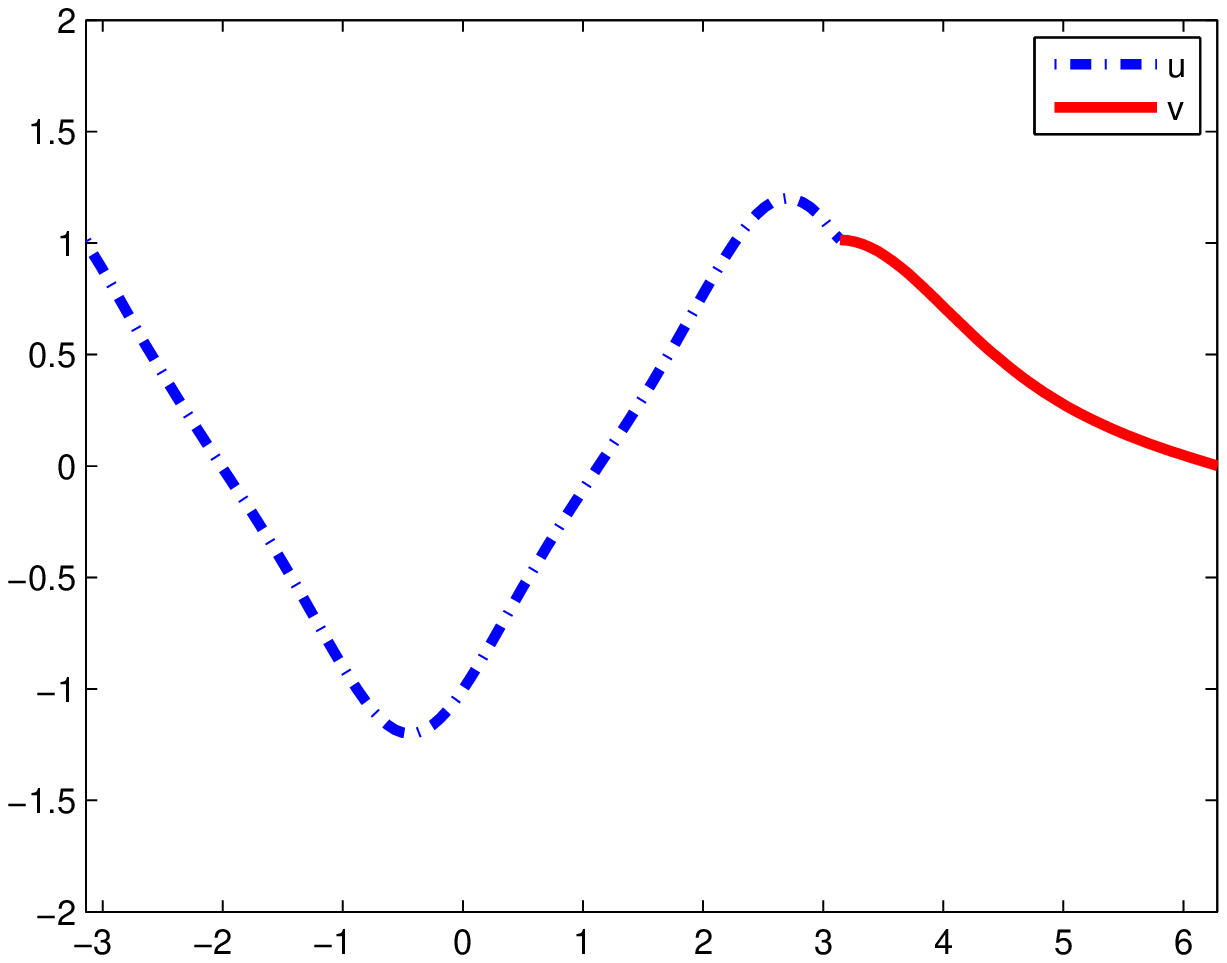}
\includegraphics[scale=0.47]{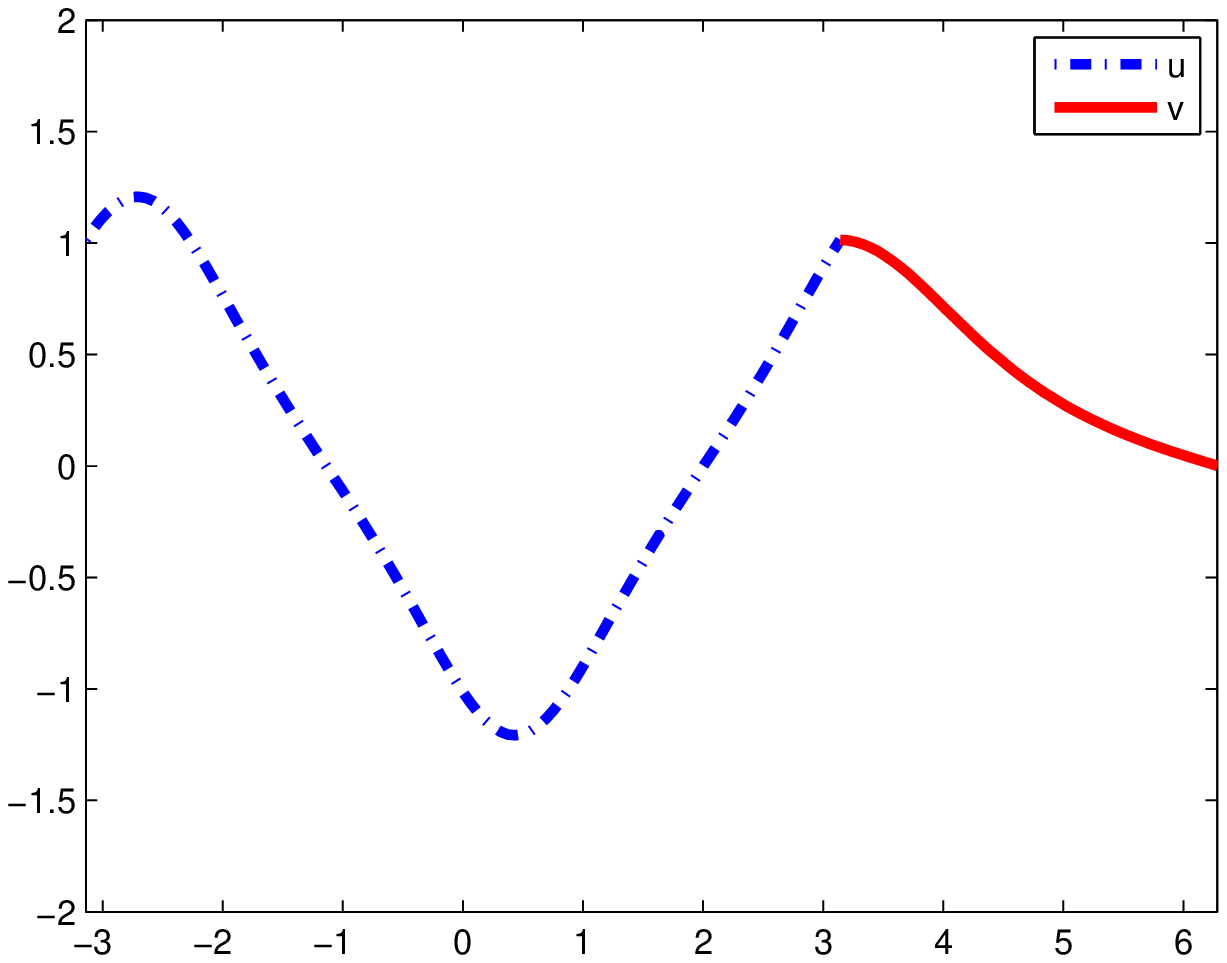}\\
\includegraphics[scale=0.47]{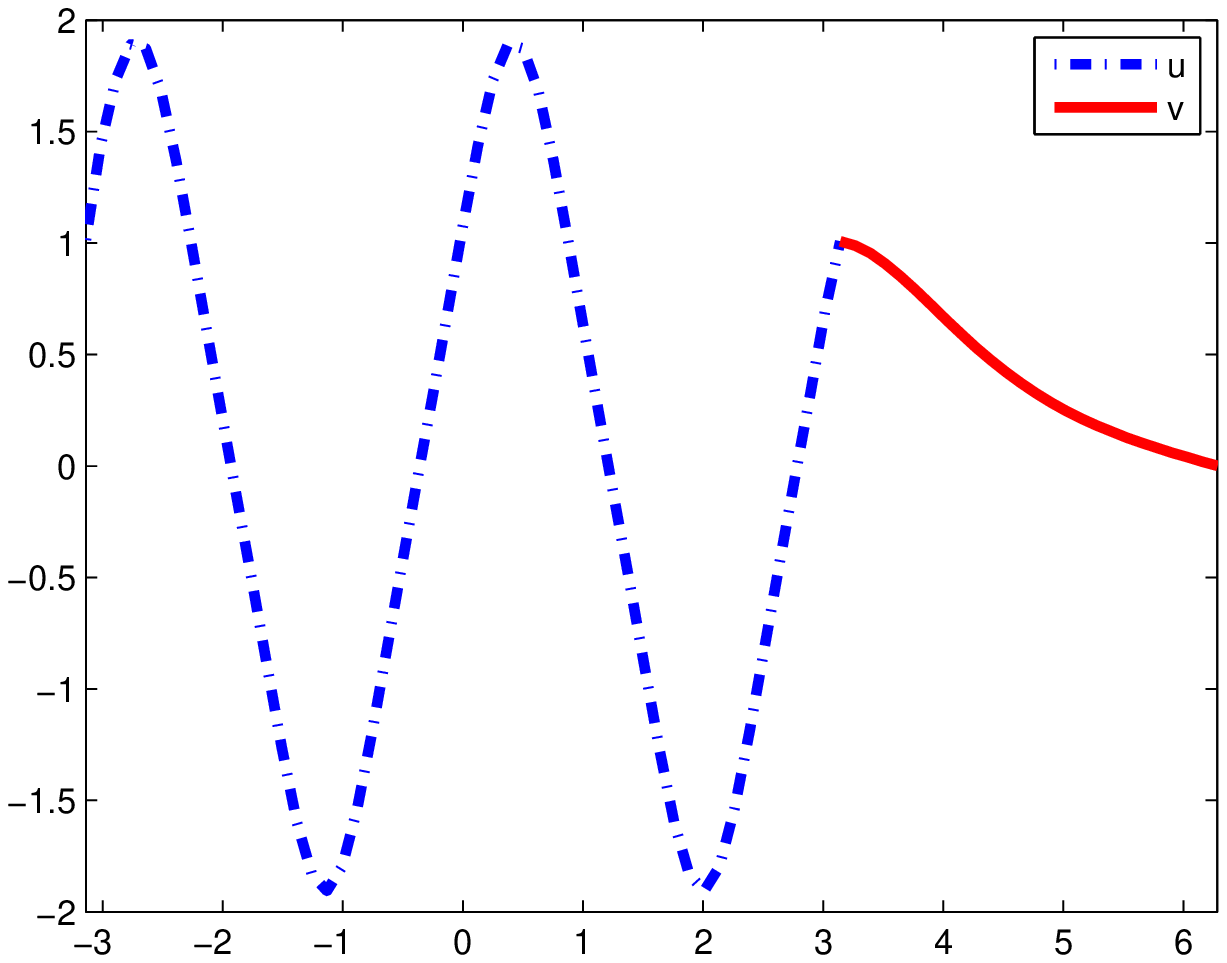}
\includegraphics[scale=0.47]{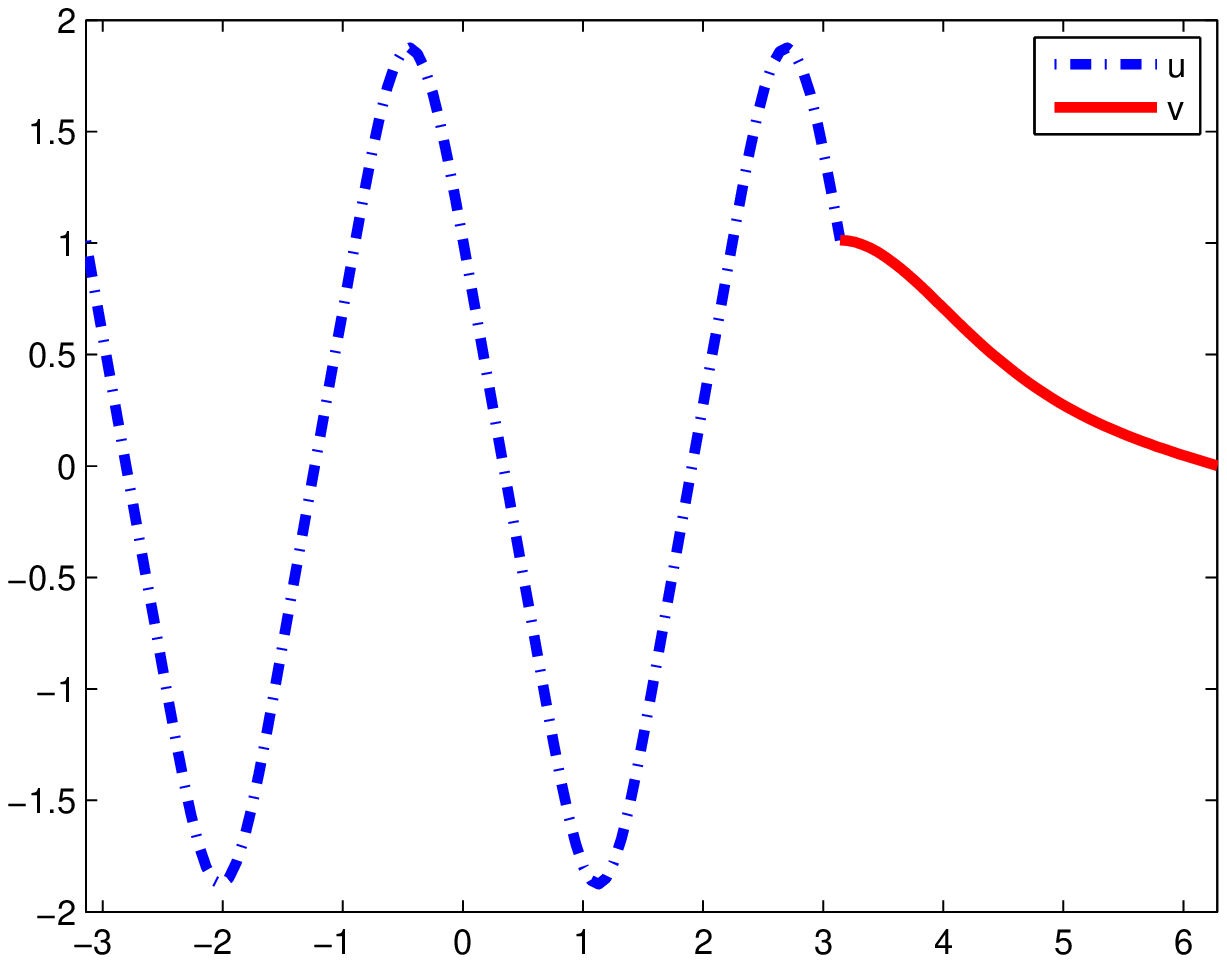}
\end{center}
\caption{Standing wave solutions $(u,v)$ versus $x$ for $\omega = -1$ along the primary branch (a) and two representatives
of the higher branches for $n = 1$ (b,c)
and $n = 2$ (d,e).}
\label{fig-sol}
\end{figure}

We then discretize the spectral problems (\ref{L-minus-nls}) and (\ref{L-plus-nls}) with the same second-order difference method
to obtain the negative and zero eigenvalues of the self-adjoint operators $L_-$ and $L_+$
for $\omega < 0$.
For eigenvalue computations, we use the MATLAB eigenvalue solver.
Figure \ref{fig1} shows the lowest six eigenvalues of these operators versus parameter $\omega$
for the standing wave solutions $(u_{n,\omega}^+,0)$ with $n = 1$ (a) and $n = 2$ (b).
The results are also identical for the standing wave solutions $(u_{n,\omega}^-,0)$.
In agreement with Propositions \ref{proposition-stability} and \ref{proposition-second-group},
we count $2n-1$ negative eigenvalues for the operator $L_-$ and $2n$ negative eigenvalues for the operator $L_+$
for every $\omega < 0$. In addition, the operator $L_-$ has a simple zero eigenvalue
and the operator $L_+$ admits no zero eigenvalue for $\omega < 0$. Note that the operator $L_+$ has a small positive eigenvalue,
which stays above zero for every $\omega < 0$.

Figure \ref{fig2} shows similar results for the lowest six eigenvalues of operators $L_-$ and $L_+$
versus parameter $\omega$ for the standing wave solutions $(u,v)$ along the higher branches
bifurcating from the solutions $(u_{n,\omega}^+,0)$ with $n = 1$ (a) and $n = 2$ (b).
The results are not identical but very similar for the standing wave solutions $(u,v)$
bifurcating from the solutions $(u_{n,\omega}^-,0)$ (not shown). In agreement with Lemma \ref{lemma-L-minus-plus},
we count $2n-1$ negative eigenvalues and one simple zero eigenvalue for the operator $L_-$ and $2n+1$ negative eigenvalues
for the operator $L_+$ for every $\omega < 0$. We also checked (not shown) that
the operator $L_-$ for the standing wave solution ($u,v$)
along the primary branch has no negative eigenvalues and a simple zero eigenvalue,
whereas the operator $L_+$ has a simple negative eigenvalue and no zero eigenvalues,
in accordance with Lemmas \ref{lemma-L-minus} and \ref{lemma-L-plus}.

\begin{figure}[htbp]
\begin{center}
\includegraphics[scale=0.47]{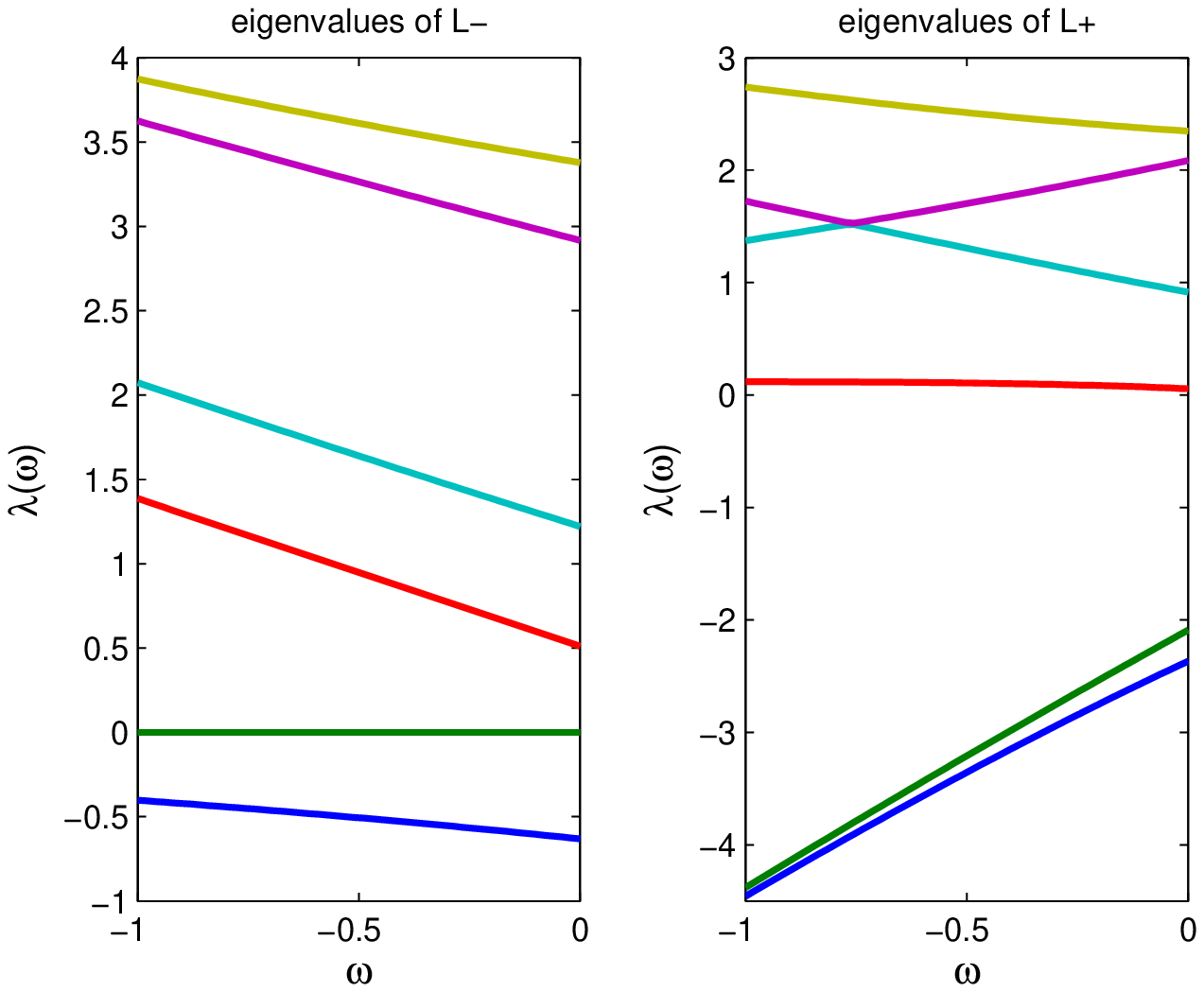}
\includegraphics[scale=0.47]{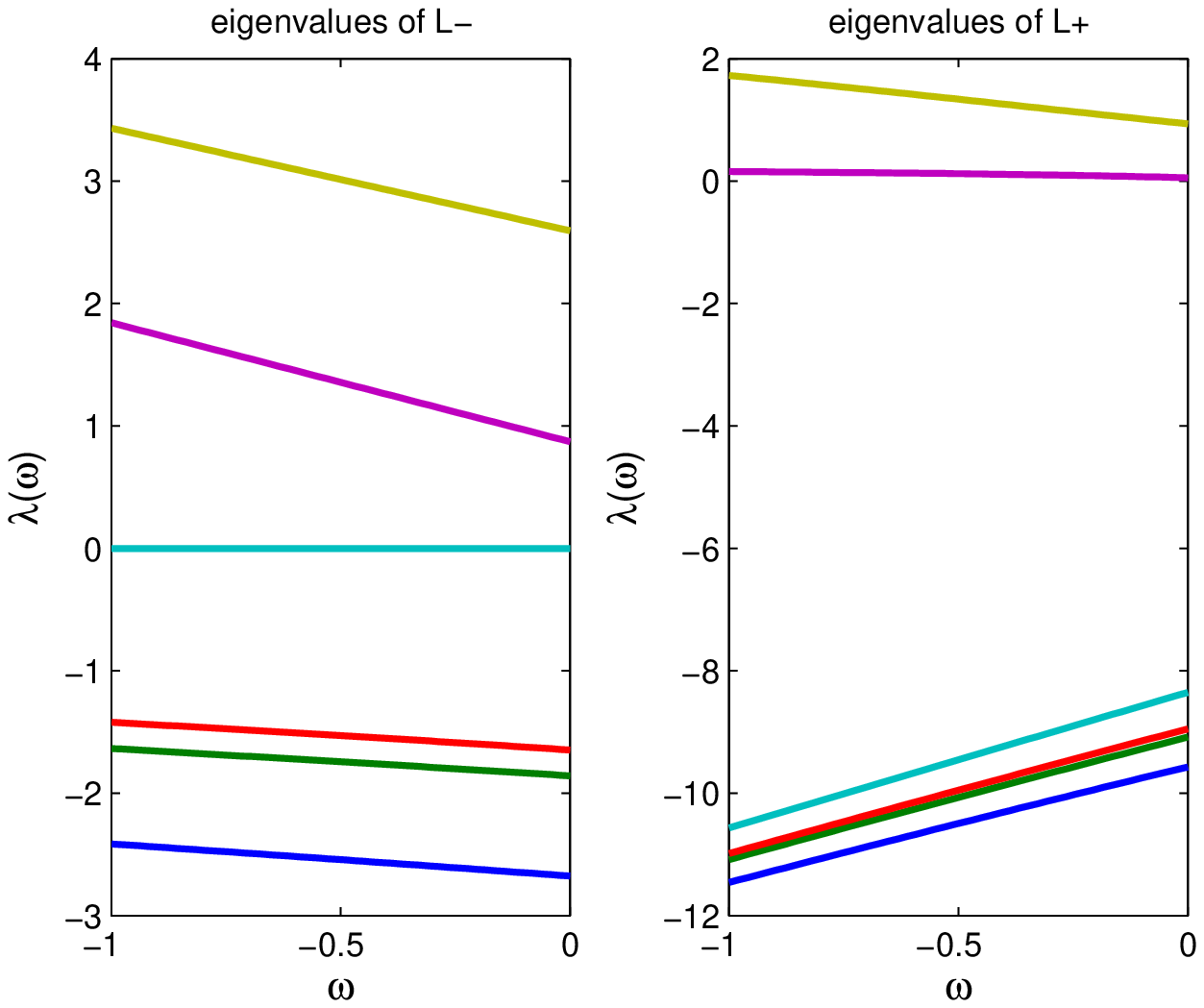}
\end{center}
\caption{Lowest six eigenvalues of operators $L_-$ and $L_+$ versus parameter $\omega$ for the
 standing wave solution $(u_{n,\omega}^+,0)$ with $n = 1$ (a) and $n = 2$ (b).}
\label{fig1}
\end{figure}

\begin{figure}[htbp]
\begin{center}
\includegraphics[scale=0.47]{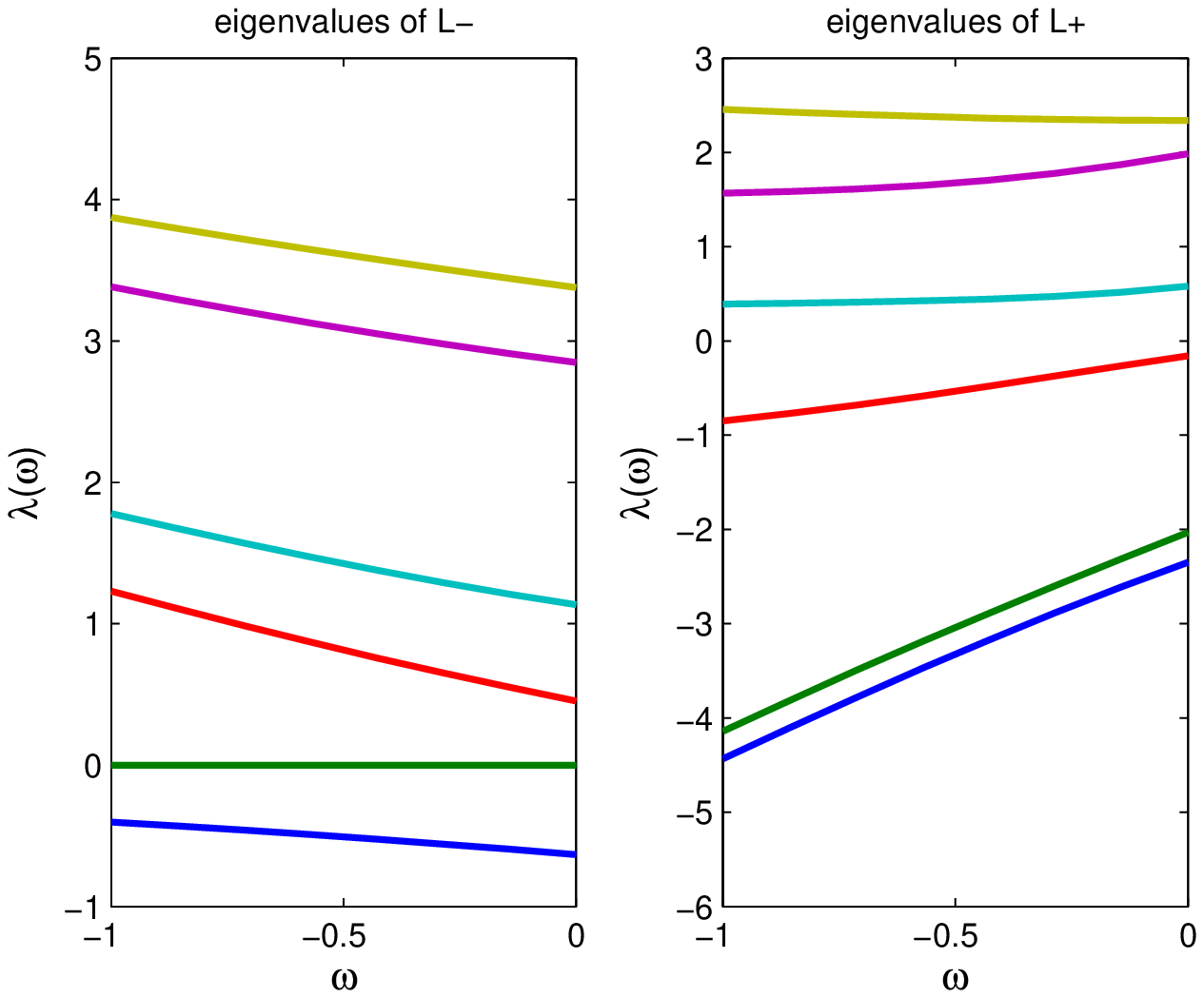}
\includegraphics[scale=0.47]{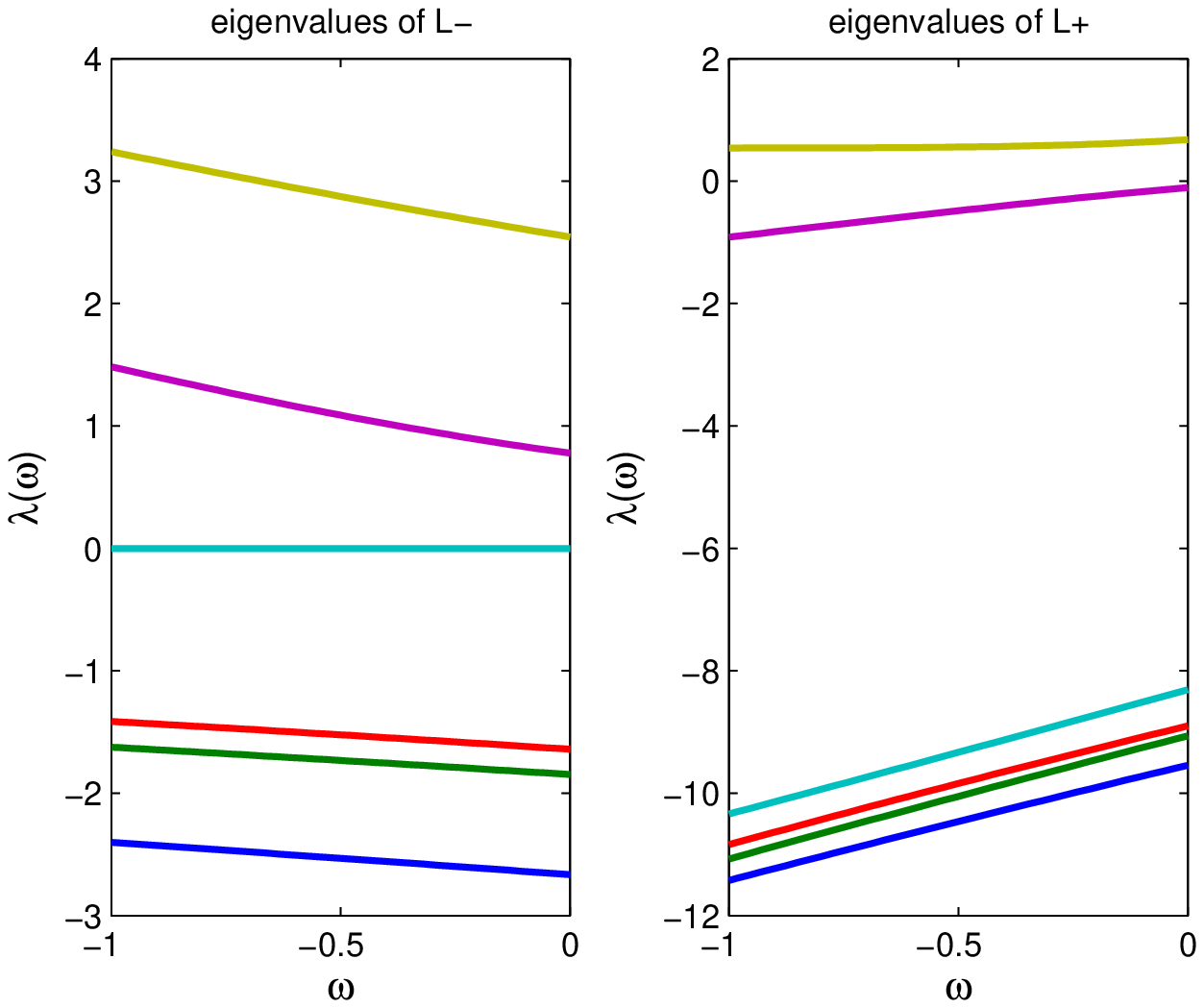}
\end{center}
\caption{Lowest six eigenvalues of operators $L_-$ and $L_+$ versus parameter $\omega$ for the
 standing wave solution $(u,v)$ along the higher branches with $n = 1$ (a) and $n = 2$ (b).}
\label{fig2}
\end{figure}

Finally, we discretize the spectral stability problem (\ref{spect-prob}) with the same second-order difference method
to obtain the unstable eigenvalues associated with the standing wave solutions for $\omega < 0$.
Figure \ref{fig3} shows all eigenvalues on the complex plane for
the standing wave solutions $(u_{n,\omega}^+,0)$ with $n = 1$ (a) and $n = 2$ (b) corresponding to $\omega = -1$.
Figure \ref{fig4} shows real and imaginary parts of the corresponding unstable eigenvalues versus parameter $\omega$.
For small negative $\omega$, we observe $2n-1$ quartets of complex eigenvalues, which means that
in addition to $n$ complex quartets predicted from the spectral problem (\ref{L-vector-2}),
there exists $n-1$ complex quartets in the spectral problem (\ref{L-vector-1}). Indeed, this does not
contradict to the theory and indicates that Conjecture \ref{conjecture-stability-higher-order} is true but
the actual number of quartets of complex eigenvalues exceed $n$ for $n \geq 2$.

We also note from Figure \ref{fig4} that all quartets of complex eigenvalues disappear for larger negative values
of $\omega$. The complex eigenvalues coalesce at the imaginary parts in the spectral gap away from the continuous
spectrum and then split into two pairs of purely imaginary eigenvalues of opposite Krein signatures.
This can be explained from the fact that the standing wave solution $(u_{n,\omega}^+,0)$ shown on Figure \ref{fig0}
look like a sequence of $2n$ NLS solitary waves of opposite polarity as $\omega \to -\infty$.
As is well known from the qualitative theory of soliton interactions \cite{NLS-sol},
pairs of NLS solitary waves of opposite polarity repel each other, so that the standing wave solution
$(u_{n,\omega}^+,0)$ represents an equilibrium configuration under a balance of repulsive force between solitary waves of
opposite polarity (which include $2n$ solitary waves on the interval $[-L,L]$ and additional
solitary waves outside $[-L,L]$, which are reflected
anti-symmetrically by the Dirichlet boundary conditions at the boundaries). Such equilibrium
configurations are spectrally stable, which explains qualitatively
disappearance of the complex unstable eigenvalues
in the limit $\omega \to -\infty$.

\begin{figure}[htbp]
\begin{center}
\includegraphics[scale=0.47]{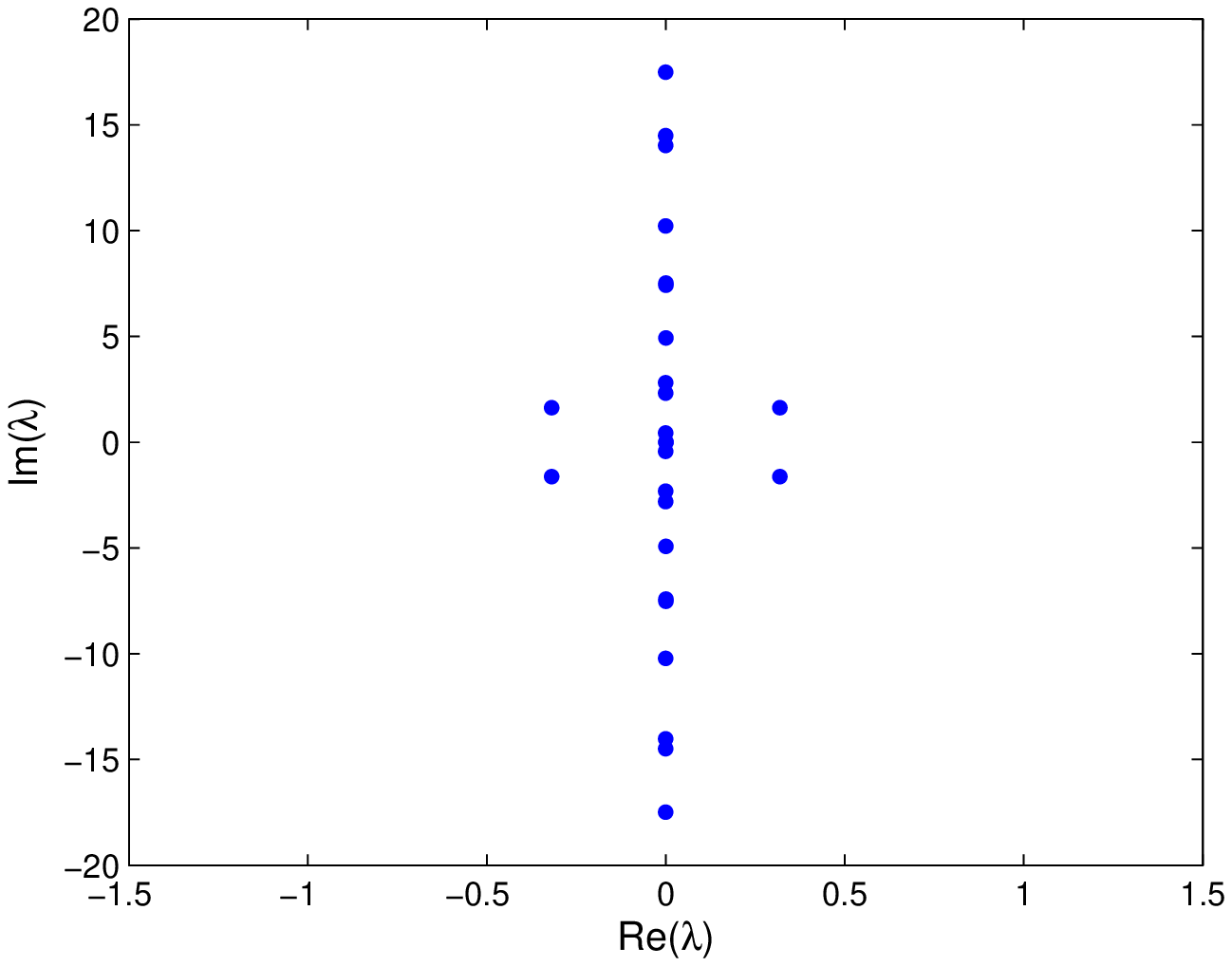}
\includegraphics[scale=0.47]{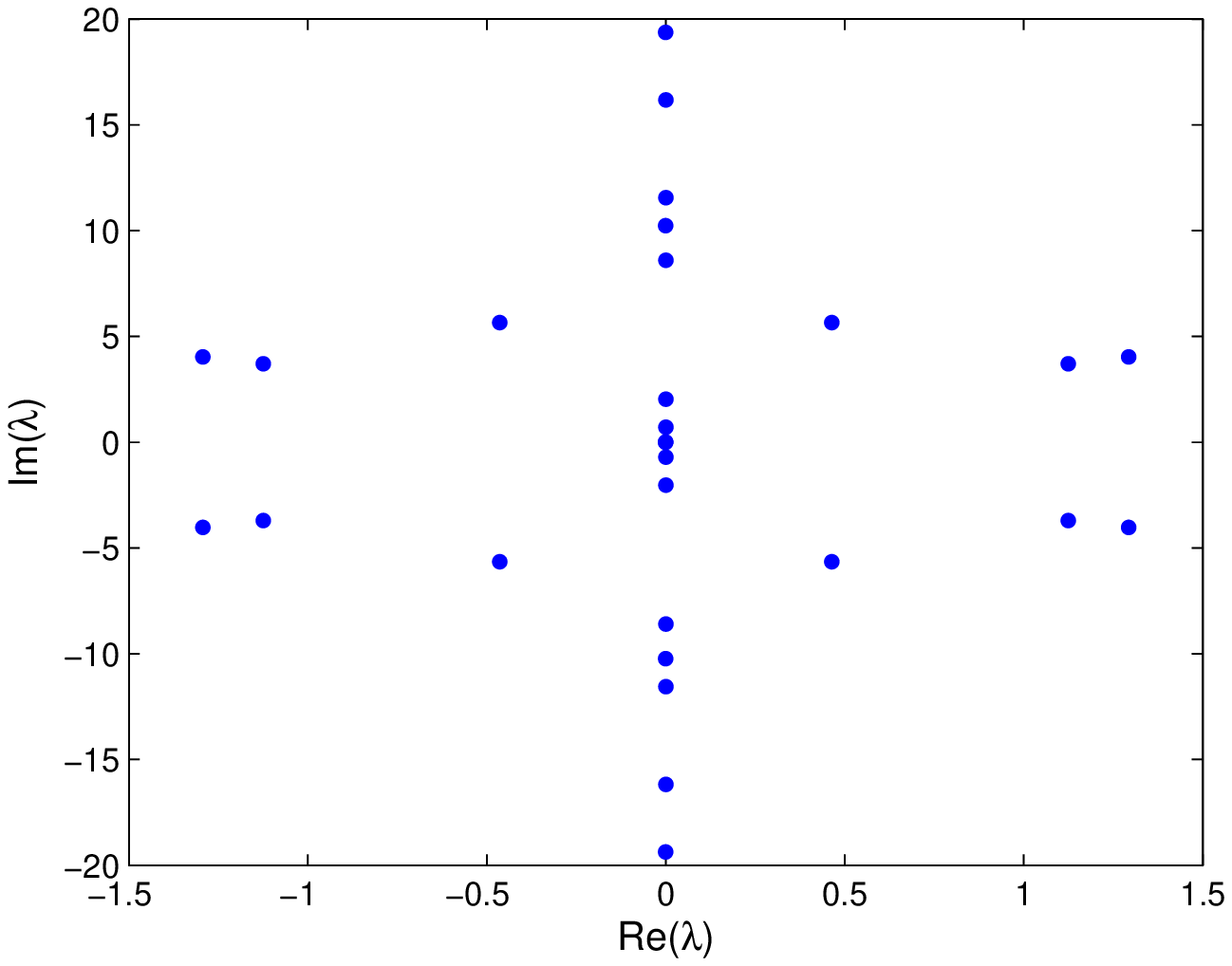}
\end{center}
\caption{Eigenvalues $\lambda$ of the spectral problem (\ref{spect-prob}) on the complex plane
for the standing wave solutions $(u_{n,\omega}^+,0)$ with $n = 1$ (a) and $n = 2$ (b) corresponding to $\omega = -1$.}
\label{fig3}
\end{figure}

\begin{figure}[htbp]
\begin{center}
\includegraphics[scale=0.45]{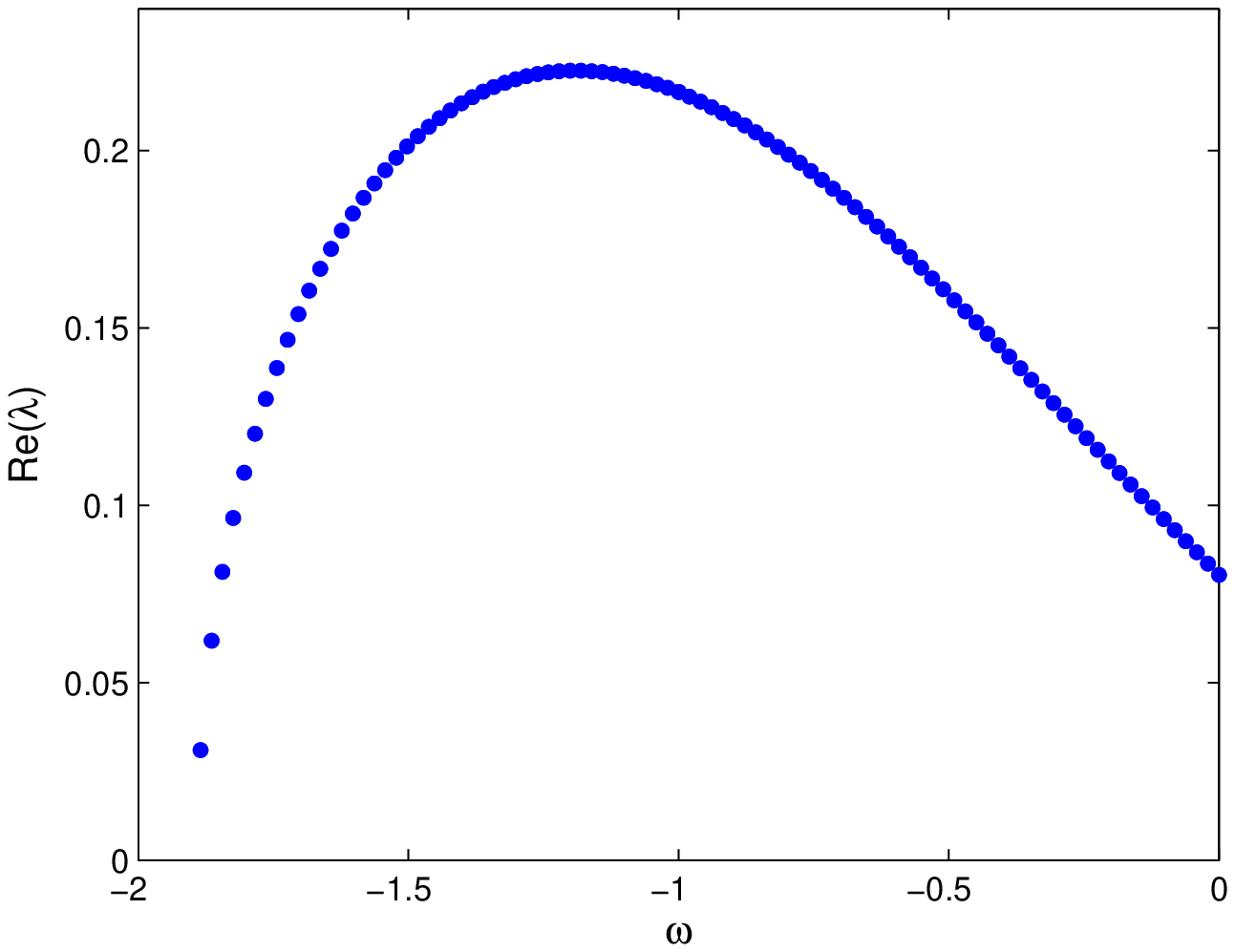}
\includegraphics[scale=0.45]{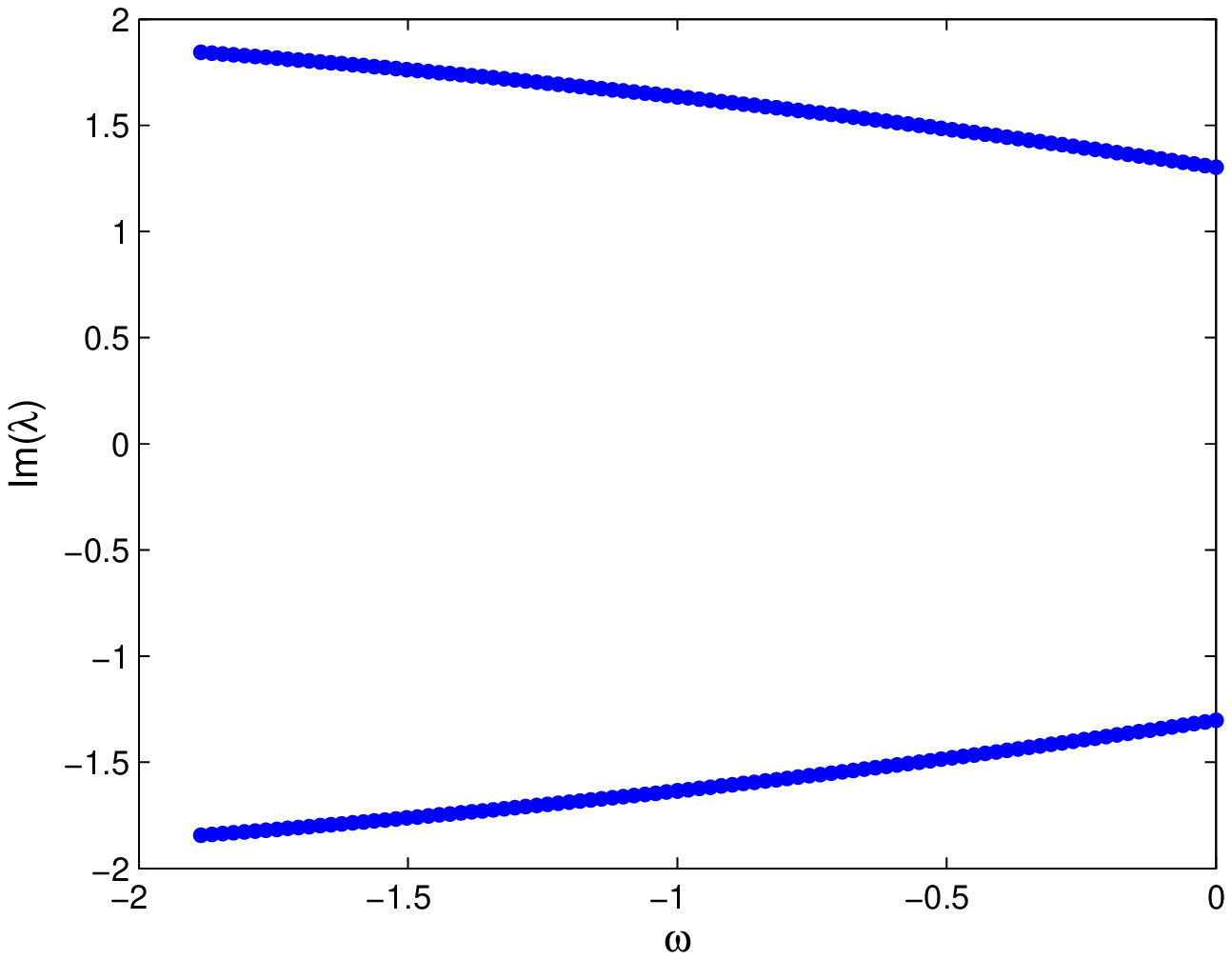}\\
\includegraphics[scale=0.45]{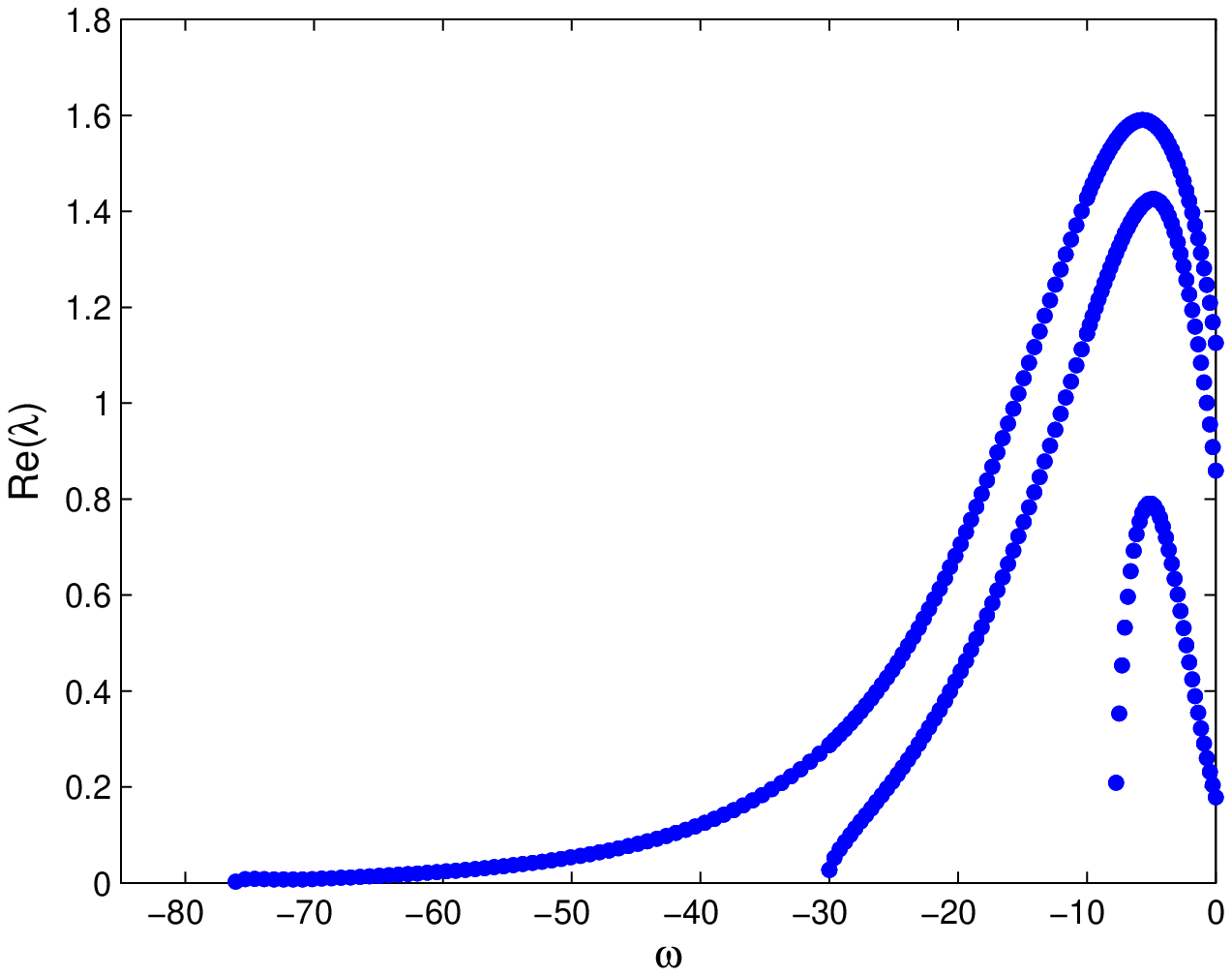}
\includegraphics[scale=0.45]{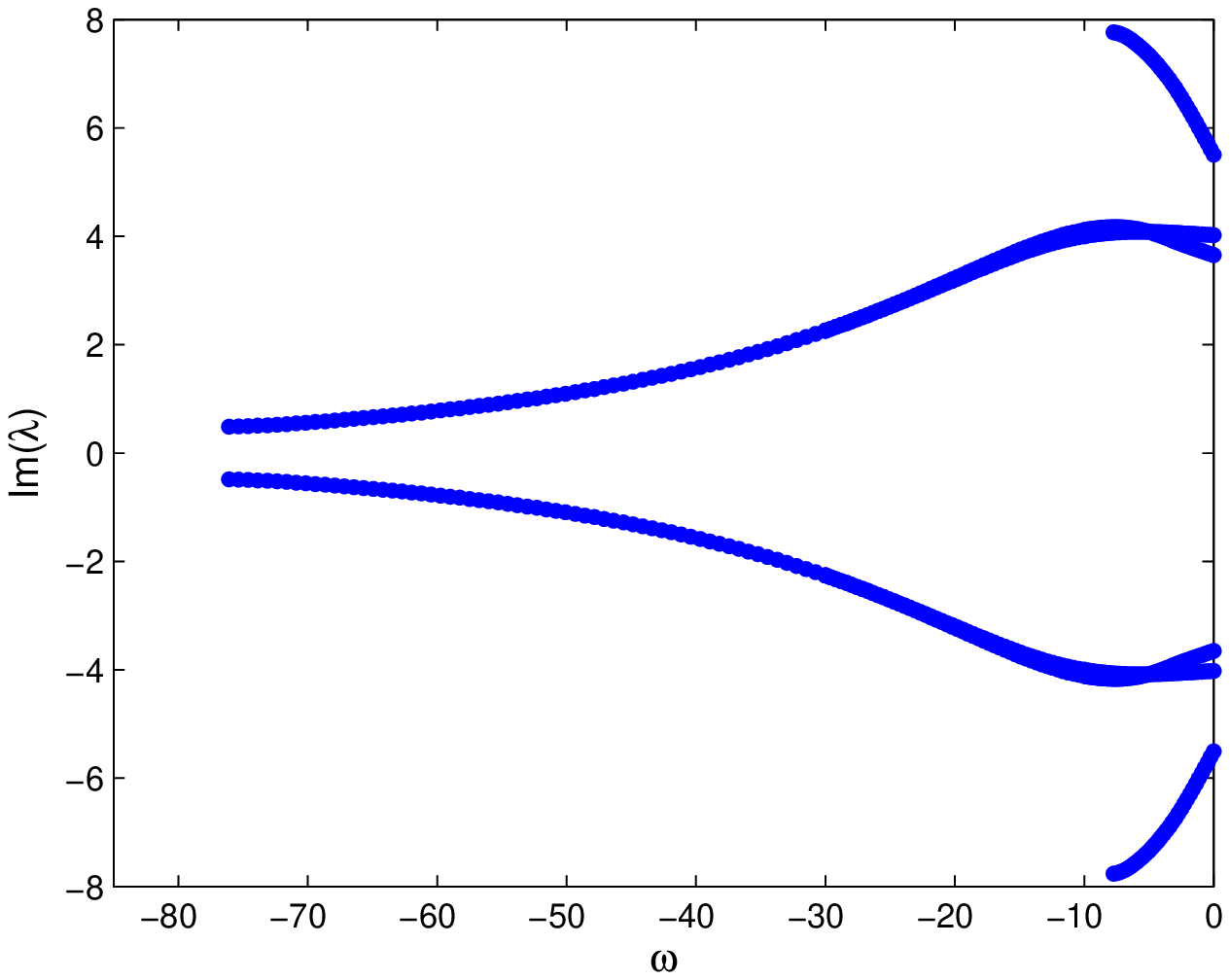}
\end{center}
\caption{Real and imaginary parts of the unstable eigenvalues $\lambda$ of the spectral problem (\ref{spect-prob})
 versus parameter $\omega$ for the standing wave solutions $(u_{n,\omega}^+,0)$ with $n = 1$ (a) and $n = 2$ (b).}
\label{fig4}
\end{figure}

Regarding the standing wave solutions along the primary and higher branches, we checked numerically
the validity of Theorems \ref{theorem-stability-primary} and \ref{theorem-stability-higher-order} (the corresponding
numerical results are not shown).
The standing wave solution $(u,v)$ along the primary branch is stable for every $\omega < 0$, in agreement
with the number of negative eigenvalues and
the orbital stability theory in \cite{GSS1}. The standing wave solutions $(u,v)$ along the higher
branches have a pair of real unstable eigenvalues in addition to the $(2n-1$) quartets of complex
eigenvalues, which are inherited from the standing wave solutions $(u_{n,\omega}^+,0)$,
from which the higher branches bifurcate off. The pair of real unstable eigenvalues
persists for every $\omega < 0$, whereas the complex quartets split into pairs of purely imaginary
eigenvalues for large negative values of $\omega$, similar to what is observed on
Fig. \ref{fig4}. These numerical results are in agreement with the number of
negative eigenvalues on Fig. \ref{fig2} and the spectral instability theory in \cite{GSS2,Gr}.

\newpage

{\bf Acknowledgements:} The work of D.N. is supported by the Ministry of University and
Research of Italian Republic, FIRB project 2012 (code RBFR12MXPO),
The work of D.P. is supported by the Ministry of Education
and Science of Russian Federation (the base part of the state task No. 2014/133).
The work of G.S. is supported by the Bolashak visiting fellowship of the government of Kazakhstan.

\end{document}